\documentclass[11pt]{article}

\usepackage{xr-hyper}
\usepackage{titlesec}
\usepackage{mathtools}
\titlelabel{\thetitle.\quad}
\usepackage{xcolor}
\usepackage{adjustbox}
\usepackage{amsmath, amsthm, amsfonts, amssymb}
\usepackage{graphicx}
\usepackage{setspace}
\usepackage{longtable}
\usepackage{breqn}
\usepackage{lscape}
\usepackage{indentfirst}
\usepackage[labelsep=period,justification=justified,singlelinecheck=false,font = footnotesize, labelfont=bf]{caption}
\usepackage{booktabs}
\usepackage{tabularx,ragged2e}
\usepackage{natbib}
\usepackage{rotating}
\usepackage{placeins}
\usepackage{subcaption}

\usepackage{bbm}
\usepackage[margin=1in]{geometry}

\usepackage{enumerate}
\usepackage{array}
\usepackage[T1]{fontenc}
\usepackage[font=small,labelfont=bf,tableposition=top]{caption}
\usepackage{mathtools}

\newcommand{\Inter}{\mathcal{I}}
\newcommand{\Hinter}{H_0^\Inter}

\newcommand{\nn}{\nonumber}

\newcommand{\CB}{\color{black}}

\DeclareCaptionLabelFormat{andtable}{#1~#2  \&  \tablename~\thetable}

\usepackage{fullpage, amsmath, amssymb, amsthm, bbm, color}
\usepackage{graphicx,caption,subcaption,placeins}
\usepackage{dsfont}

\usepackage{natbib}
\bibpunct{(}{)}{;}{a}{,}{,}

\newtheorem{definition}{Definition}
\newtheorem{theorem}{Theorem}
\newtheorem*{theorem*}{Theorem}
\newtheorem{example}{Example}
\newtheorem{procedure}{Procedure}
\newtheorem{lemma}{Lemma}

\usepackage{tikz}
\usetikzlibrary{positioning,chains,fit,shapes,calc}

\definecolor{myblue}{RGB}{80,80,160}
\definecolor{mygreen}{RGB}{80,160,80}

\newcommand{\bipartiteTwo}{
\begin{figure}[ht!]
\centering
\begin{tikzpicture}[thick,
  every node/.style={draw,circle},
  fsnode/.style={fill=myblue},
  ssnode/.style={fill=mygreen},
  every fit/.style={ellipse,draw,inner sep=-2pt,text width=2cm}, ,shorten >= 3pt,shorten <= 3pt
]
}

\newcommand{\Ind}[1]{\mathbbm{1}\{#1\}}

\newcommand{\Uset}{U}

\newcommand{\Zset}{\mathcal{Z}}
\newcommand{\Cset}{\mathcal{C}}
\newcommand{\Zobs}{Z^{\mathrm{obs}}}
\newcommand{\Yobs}{Y^{\mathrm{obs}}}

\newtheorem{proposition}{Proposition}

\pdfoutput=1

\usepackage{epstopdf}
\usepackage{wrapfig}
\usepackage{float}
\usepackage[justification=centering]{caption}

\let\ORGsidewaysfigure\sidewaysfigure
\let\ORGendsidewaysfigure\endsidewaysfigure
\renewenvironment{sidewaysfigure}
  {\ORGsidewaysfigure\vspace*{.5\textwidth}\begin{minipage}{\textheight}\centering}
  {\end{minipage}\ORGendsidewaysfigure}

\usepackage{epstopdf}
\newcommand{\appropto}{\mathrel{\vcenter{
  \offinterlineskip\halign{\hfil$##$\cr
    \propto\cr\noalign{\kern2pt}\sim\cr\noalign{\kern-2pt}}}}}

\newcolumntype{Y}{>{\centering\arraybackslash}X}
\providecommand{\U}[1]{\protect\rule{.1in}{.1in}}
\providecommand{\U}[1]{\protect\rule{.1in}{.1in}}
\providecommand{\U}[1]{\protect\rule{.1in}{.1in}}
\providecommand{\U}[1]{\protect\rule{.1in}{.1in}}
\theoremstyle{plain}

\newcommand{\comm}[1]{\text{\footnotesize [#1]}}
\newcommand{\tobs}{T^{\mathrm{obs}}}

\newcommand{\qa}{q_{\alpha}}
\newcommand{\Qa}{\hat q_{\alpha, m}}
\newcommand{\Fom}{\hat F_{0, n, m}}

\newcommand{\Udom}{\mathbb{U}}
\newcommand{\Zdom}{\mathbb{Z}}
\newcommand{\Ydom}{\mathbb{R}}
\newcommand{\Edom}{\mathbb{F}}
\newcommand{\Real}{\mathbb{R}}

\newcounter{remark}[section]
\usepackage[colorlinks = true,
            linkcolor = blue,
            urlcolor  = blue,
            citecolor = blue,
            anchorcolor = blue]{hyperref}

\newcommand{\Tobs}{T^{\mathrm{obs}}}

\newcommand{\pr}{P}
\newcommand{\obs}{\mathrm{obs}}

\renewcommand{\aa}{\mathsf{a}}
\newcommand{\bb}{\mathsf{b}}
\newcommand{\Hf}{H_0^{\mathcal{F}}}
\newcommand{\Hab}{H_0^{\{\aa, \bb\}}}
\newcommand{\Haa}{H_0^{\{\aa\}}}
\newcommand{\Hbb}{H_0^{\{\bb\}}}

\newcommand{\EX}{\mathbb{E}}
\newcommand{\Hex}{H_0^{\mathrm{ex}}}

\newcommand{\Fset}{\mathcal{F}}
\newcommand{\GfF}{G^{\Fset}_f}

\newcommand{\ThmOne}{
Consider the null hypothesis $\Hf$ in Equation~\eqref{eq:H0}. 
Construct the corresponding null exposure graph, $\GfF$, and compute a biclique
decomposition~$\Cset$. Let $C\in \Cset$ be the unique biclique such that $\Zobs \in C$.
Then, the randomization test described in Procedure~\ref{proc:test} is valid conditionally at any level, i.e., 
the p-value defined in Equation~\eqref{eq:pval} satisfies:
$$
\EX\left(\Ind{\mathrm{pval}(\Zobs, \Yobs; C) \le \alpha} \mid \Cset, \Hf\right) = \alpha,
$$
where the expectation is with respect to the design, $\pr(\Zobs)$, 
and $\alpha\in(0, 1)$.
}

\newcommand{\ThmTwo}{
Consider the intersection hypothesis  $\Hinter$ defined in~\eqref{eq:H0_inter}.
Then, the biclique test in Procedure~\ref{proc:test} operating on the biclique decomposition $\Cset$ from Procedure~\ref{proc:multi} is a conditionally valid test for $\Hinter$.
}

\newcommand{\Fo}{F_{0, n}}
\newcommand{\Fn}{F_{1, n}}

\newcommand{\ThmPowerText}{
In Procedure~\ref{proc:test}, let $C = (U, \Zset)\in\Cset$ be the conditioning biclique, where $\Cset$ is a fixed biclique decomposition. Let  $|C| = (n, m)$ denote the size of the biclique, 
with $n=|U|$ and $m=|\Zset|$.
Let the randomization distribution and the null distribution be denoted, respectively, by
\begin{align}\label{eq:def_hatF}
t(Z, Y(Z); C) \sim \hat F_{1, n, m},~\text{and}~t(Z, \Yobs; C) \sim \hat F_{0, n, m},
~\text{where}~Z\sim P(Z | C).
\end{align}
Suppose that for any fixed $n>0$:
\begin{enumerate}
\item[(A1)]  There exist continuous cdfs $F_{1, n}$ and $F_{0, n}$ such that $\hat F_{1, n, m}$ and $\hat F_{0, n,m}$ in~\eqref{eq:def_hatF} are  the empirical distribution functions over $m$ independent samples 
from $F_{1, n}$ and $F_{0, n}$, respectively.

\item[(A2)] There exists $\sigma_n>0$, and a continuous cdf $F$, such that 
$\Fo(t) = F(t/\sigma_n)$, for all $t\in\mathbb{R}$.
\item[(A3)] The treatment effect~(e.g., spillover contrast) is additive, that is, there exists a fixed $\tau\in\mathbb{R}$ such that $\Fn(t) = \Fo(t-\tau)$, for all $t\in\mathbb{R}$.
\end{enumerate}
Let $\phi_{n, m}$ be the power of the biclique randomization test conditional on a biclique of size $(n, m)$, i.e., 
$$
\phi_{n, m} = \EX \left(\Ind{\mathrm{pval}(Z, \Yobs; C) \le \alpha} \mid |C|=(n, m)\right).
$$ 

Fix any small $\delta>0$. Then, for large enough $m$,
\begin{equation}\label{eq:phi_power}
\phi_{n, m}
\ge 1- F(F^{-1}\big(1-\alpha) - \tau/\sigma_n\big)- O(m^{-0.5+\delta}).\nonumber
\end{equation}
If, in addition, $\sup_{x\in\mathbb{R}} |F(x) - 1/(1+e^{-b x})| \le \epsilon$ for fixed $b, \epsilon>0$, and 
$\sigma_n = O(1/\sqrt{n})$, then for some fixed $A, a>0$
}

\newcommand{\ThmPower}{
\ThmPowerText
\begin{equation}\label{eq:power}
\phi_{n, m}
\ge \frac{1}{1 + A e^{- a\tau \sqrt{n}}}  - O(m^{-0.5+\delta}) - \epsilon.
\end{equation}
}

\newcommand{\ThmPowerApp}{
\ThmPowerText
\begin{equation}
\phi_{n, m}
\ge \frac{1}{1 + A e^{- a\tau \sqrt{n}}}  - O(m^{-0.5+\delta}) - \epsilon.\nonumber
\end{equation}
}

\newcommand{\ThmDesignAid}{
The randomization implied by Procedure~\ref{proc:test2} is valid at the nominal level for the null hypothesis of Equation~\eqref{eq:H0cluster} 
under the two-stage experimental design of Section~\ref{sec:clustered_sim}.
}

\newcommand{\commentEq}[1]{
~\text{\footnotesize~[ #1 ]}
}

\begin{document}

\title{A Graph-Theoretic Approach to Randomization Tests of \\ 
Causal Effects
Under General Interference\thanks{email: \texttt{David.Puelz@chicagobooth.edu}. We thank Peng Ding, Connor Dowd, Colin Fogarty, 
Sam Pimentel, Stephen Raudenbush, Paul Rosenbaum, and Santiago Tob{\'o}n, as well as conference participants at Polmeth, Advances with Field Experiments, Design and Analysis of Experiments, 
Atlantic Causal Inference, and seminar  participants at the University of Chicago for helpful comments and discussion. AF gratefully acknowledges support from a National Academy of Education/Spencer Foundation postdoctoral fellowship. PT is grateful for the John E. Jeuck Fellowship at University of Chicago, Booth School of Business.}\vspace{12mm}}

\author{
David Puelz\thanks{The University of Chicago, Booth School of Business}
\and 
Guillaume Basse\thanks{Stanford University}
\and
Avi Feller\thanks{The University of California, Berkeley}
\and
Panos Toulis$^{\dagger}$}

\date{}
\maketitle
\vspace{5pt}

\begin{center}
\textsc{This Version: \today \\ \vspace{2mm}
}
\end{center}

\thispagestyle{empty}

\singlespacing
\begin{abstract}
Interference exists when a unit's outcome depends on another unit's treatment assignment.
For example, intensive policing on one street could have a spillover effect on neighboring streets.
Classical randomization tests typically break down in this setting because many null hypotheses of interest are no longer sharp under interference.
A promising alternative is to instead construct a conditional randomization test on a subset of units and assignments for which a given null hypothesis is sharp. 
Finding these subsets is challenging, however, and existing methods are limited to special cases or have limited power.
In this paper, we propose valid and easy-to-implement randomization tests for a general class of null hypotheses under arbitrary interference between units.  
Our key idea is to represent the hypothesis of interest as a bipartite graph between units and assignments, and to find an appropriate biclique of this graph.  
Importantly,  the null hypothesis is sharp within this biclique, enabling conditional randomization-based tests. We also connect the size of the biclique to statistical power.
Moreover, we can apply off-the-shelf graph clustering methods to find such bicliques efficiently and at scale.
We illustrate our approach in settings with clustered interference and show advantages over methods designed specifically for that setting.  We then apply our method to a large-scale policing experiment in Medell\'in, Colombia, where interference has a spatial structure.
\end{abstract}


\setcounter{page}{0}\thispagestyle{empty}\baselineskip18.99pt\newpage 

\newpage
\onehalfspacing

\section{Introduction}

The assumption of  ``no interference'' between units~\citep{cox1958planning} underlies 
most classical approaches to causal inference.
The key premise is that a unit's treatment does not affect another unit's outcome, so that 
each unit's outcome depends only on its own treatment status.
This is implausible in many settings, however, as it precludes peer effects, treatment spillovers, and other forms of treatment interference~\citep{halloran2016dependent}.

Classical approaches often break down under interference. The canonical Fisher Randomization Test~(FRT), for example, is valid for testing the global sharp null hypothesis of no treatment effect but fails when testing non-sharp null hypotheses, such as tests of treatment spillovers. Several recent proposals address this issue by restricting the randomization test to a subset of units, often called focal units,
and a subset of assignments~\citep{aronow2012general,athey2018exact,basse2019}. 
The central idea is that, conditional on these subsets, the specified null hypothesis is sharp for every focal unit, and thus the \emph{conditional} randomization test is valid. 
These randomization-based approaches have many advantages over model-based alternatives~\citep{manski2013identification, graham2008identifying, jackson2010social, graham2005identification, brock2001interactions, blume2015linear, toulis2013estimation, bowers2013reasoning, auerbach2016identification, leung2015two} because  they make minimal assumptions and are typically more robust.
However, existing randomization-based methods are usually tailored to a specific interference structure, 
which limits their power and widespread adoption.
For example,~\citet{basse2019} develop a permutation-based approach tailored to clustered interference, which cannot be easily generalized, and~\citet{athey2018exact} restrict the type of conditioning information used for constructing the test.

In this paper, we propose a general procedure for identifying subsets of units and assignments for which the null hypothesis of interest is sharp, and use it to develop a method for randomization tests under arbitrary interference.
Our key methodological contribution is the {\em null exposure graph}. 
This is a bipartite graph with units and assignments as the nodes, and an edge between any unit-assignment pair if we observe a 
null exposure, i.e., a treatment exposure specified in the null hypothesis, for the unit and assignment in the pair.
A biclique of the null exposure graph is a complete bipartite subgraph in which all units on one side are connected to all assignments on the other side.
Crucially for testing, this biclique constitutes a subset of units and assignments for which the null hypothesis is sharp and a valid randomization test is possible.


Our proposed method offers  three main benefits over existing approaches. 
First, it allows conditioning 
on the observed assignment, which can increase power over methods that suggest random conditioning~\citep{athey2018exact, aronow2012general}.
Second, since a null hypothesis uniquely determines a null exposure graph,
our method is constructive and defines, step by step, how to conduct randomization tests under 
general forms of interference. This is an improvement over methods that are tailored to specific
patterns of interference~\citep{basse2019}.  Finally, our method  translates questions of computation and statistical power 
into properties and operations on the null exposure graph, 
separating these considerations from test validity. 
Our approach is therefore modular, and  can benefit from separate advances in graph algorithms for biclique computation.

To illustrate our method, we consider two natural structures of interference: clustered interference and spatial interference.
Under clustered interference, units can be separated into well-defined clusters, such as households or classrooms, where we assume that units interact within clusters but not between clusters. Our motivating example here is a two-stage randomized trial of a student attendance intervention in which households are assigned to treatment or control and then, within each treated household, one student is randomly treated; see~\citet{basse2018analyzing}.
Under spatial interference, we assume that interactions ``pass through'' neighboring units, but without the simpler structure of clustered interference. Here, we focus on re-analyzing a large-scale experiment in Medell{\'i}n, Colombia studying the impact of ``hotspot policing'' on crime \citep{collazos2019hot}.
Our analysis of spatial interference considers hypotheses on 
any specified spillovers, which 
differs from other design-based methods
that consider a marginal spillover effect over the design~\citep{aronow2019design}.
We also extend our framework to test null hypotheses that restrict the level of interference between units, such as the null hypothesis of no interference.

Our paper is structured as follows. 
Section~\ref{sec:setup} introduces the problem setup and all necessary notation. 
Section~\ref{sec:method} outlines our methodology, comprised of the 
null exposure graph~(Section~\ref{sec:null_graph}) and biclique decompositions of the graph~(Section~\ref{sec:cliques}). 
Section~\ref{sec:test} presents the proposed randomization test.
Section~\ref{sec:power_main} gives the main results on statistical power.
We then illustrate our method in two applications. 
Section~\ref{sec:clust} considers settings with clustered interference, and Section~\ref{sec:spatial} considers settings with spatial interference, specifically in the context of a large-scale policing experiment in Medell\'in, Colombia.  
Finally, Section~\ref{sec:composite} extends our results to complex null hypotheses. The Appendix contains additional empirical and theoretical results, as well as the proofs.

\section{Overview of causal inference under interference and problem setup}
\label{sec:setup}

\subsection{Setup and notation}
\label{sec:setup_sub}
Consider a finite population of $N$ units indexed by $i=1, \ldots, N$. Let $Z_i$ denote unit $i$'s treatment, which we assume to be binary 
without loss of generality. Let $Z  = (Z_1, Z_2, \ldots, Z_N)\in\{0, 1\}^N$ denote the population treatment assignment, and $\pr(Z)$ its distribution according to the experiment design. We focus on experimental studies,  and so $P(Z)$ is known.
Let $Y_i(z)\in\Ydom$ denote the  potential outcome of unit $i$ under 
population assignment $z\in\{0, 1\}^N$.
For the observed quantities we use the modifier ``$\mathrm{obs}$'' as a superscript for emphasis. Thus, 
$\Zobs\in\{0, 1\}^N$ is the observed population treatment, and
$\Yobs = (Y_1(\Zobs), \ldots, Y_N(\Zobs))\in\mathbb{R}^N$ is the vector of observed outcomes. In the randomization framework, $\Zobs$ is random according to the design, $\Zobs\sim\pr(\Zobs)$, whereas the potential outcomes are fixed.
Let $\Udom = \{1, \ldots, N\}$ denote the set of all units,
 and $\Zdom = \{z\in\{0, 1\}^N : \pr(z) > 0\}$ denote the set of 
 all population treatment assignments supported by the design.

 The main challenge is that causal inference is infeasible without restrictions on the potential outcomes; unrestricted, each unit has $2^N$ possible potential outcomes. For instance, the common \emph{no interference} assumption states that the outcome for each unit depends only on its own
 treatment assignment, so each unit has only two potential outcomes.

 When this assumption is implausible, one strategy is to define a \emph{treatment exposure}, which is a low-dimensional summary of $Z$, such as ``spillovers'' or ``peer effects".
In particular, we assume a finite set of possible treatment exposures, $\Edom = \{\aa, \bb, ...\}$, 
and exposure mapping functions, $f_i :\Zdom\to\Edom$, 
for each unit $i\in\Udom$. 
Thus, the definition of $f_i$ is application-specific, as it needs to consider the interference structures that are appropriate for a given problem; see \citet{aronow2017estimating}. 
To make this concrete, we briefly introduce three examples.

\begin{example}[Clustered interference]\label{example:clustered}
Following the setting in \citet{basse2019}, each unit belongs to a fixed cluster, such as individuals within households. The key assumption is that units interact within each cluster, but not between clusters, also known as partial interference \citep{sobel2006randomized}. 
In the experiment, clusters are assigned to treatment or control, and, within treated clusters, one unit is randomly assigned to treatment.
Here, the exposure function may be defined as
$f_i(z) = z_i + \sum_{j\in[i]} z_j$, where $[i]$ denotes all units in $i$'s cluster. Thus, $\Edom=\{0, 1, 2\}$, and there are three exposure levels: $f_i=0$ denotes a control unit in a control cluster; 
$f_i=1$ denotes a control unit in a treated cluster; and $f_i=2$ denotes a treated unit 
in a treated cluster.~We explore this setting further in Section~\ref{sec:clust}. 
\end{example}

\begin{example}[Spatial interference]
In this setting, units interact with one another locally through shared space, like street segments in a city. 
The goal is to test whether the outcomes of untreated units close to treated units (i.e., spillover units) are affected by the treatment. For example, let $g_{ij}^r = \Ind{d(i, j) < r}$ denote whether units $i$ and $j$ are within distance $r$ of each other, where $d(i, j)$ denotes their spatial distance.
We may define $f_i(z) = (w_i, z_i)$, where 
$w_i = \Ind{\sum_{j\neq i} g_{ij}^rz_j > 0 }$ indicates whether $i$ 
is within a radius $r$ of any other treated unit. 
Here, $\Edom=\{(0, 0), (0, 1), (1, 0), (1, 1)\}$, and there are four exposure levels. 
We explore this setting further in Section~\ref{sec:spatial}.
\end{example}

\begin{example}[Extent of interference]\label{example_extent}
In this setting, units are linked in a social network. Specifically, we assume an integer-valued function $d(i, j)\in\mathbb{N}^+$  that measures the distance between units $i$ and $j$ in the network, 
such that  $d(i, j)$ is symmetric, $d(i, i)=0$ for each $i$, and $d(i, j)=\infty$ if $i$ and $j$ are not connected.
For some integer $k\ge 0$, let $G^k = (g_{ij}^k)$ be the social network such that 
$g_{ij}^k = \mathbb{I}\{d(i, j) \le k\}$, and define $f_i(z) = G_i^k z$, where $G_i^k$ 
is the $N\times N$ diagonal matrix having as diagonal the $i$-th row of matrix $G^{k}$. 
Here, the exposure for every unit is an $N$-length binary vector,
and so $\Edom \subseteq\{0, 1\}^N$. 
With this setup, different values of $k$ restrict the extent of interference.
For example, when $k=0$, the exposure $f_i(z)$ for some unit  $i$  depends only on 
individual treatment $z_i$. This can describe settings where the extent of interference is minimal. However, when $k=1$, $f_i(z)$ depends on the individual treatment of $i$ and the treatments of 
$i$'s immediate neighbors; when $k=2$, $f_i(z)$ additionally depends on the treatments of second-order neighbors, and so on. Thus, for sufficiently large $k$, $f_i(z)$ could depend on the treatments 
of all units connected to $i$.
We explore this setting further in Section \ref{sec:composite}.
\end{example}

\subsection{Null hypotheses on exposure mappings}
\label{sec:H0}
Our primary goal is to test null hypotheses that specify relations between treatment assignments and potential outcomes. First, consider the following general hypothesis:
\begin{align}
\label{eq:H0}
\Hf : Y_i(z) = Y_i(z')~\text{for all}~i=1, \ldots, N,~\text{and any}~z, z'\in\Zdom~\text{such that}~f_i(z), f_i(z')\in
\Fset\subseteq\Edom.
\end{align}
%
In words, $\Hf$ states that for each unit the outcomes under any exposure in $\Fset$ are 
identical, regardless of the particular population treatment assignment, $z \in \Zdom$. 
The formulation in Equation~\eqref{eq:H0} is quite general, and can express different kinds of hypotheses under interference for various choices of $\Fset$.
At one extreme, $H_0^{\Edom}$ is the ``global null hypothesis'' of no treatment effect. This 
is a sharp null hypothesis that
can be tested by the classical Fisher randomization test, since all potential outcomes can be imputed under the null. At the other extreme, $H_0^{\emptyset}$ is degenerate and is fundamentally untestable.

An important special case is the singleton null hypothesis, $H_0^{\{\aa\}}$, where $\Fset = \{\aa\}$. Formally:
\begin{align}
\label{eq:H0_singleton}
H_0^{\{\aa\}} : Y_i(z) = Y_i(z')~\text{for all}~i=1, \ldots, N,~\text{and any}~z, z'\in\Zdom~\text{such that}~f_i(z) = f_i(z') = \aa.
\end{align}
This is a \emph{stability} or \emph{exclusion restriction} hypothesis: under $H_0^{\{\aa\}}$, the potential outcomes for exposure $f_i(z) = \aa$ are only functions of exposure $\aa$ and not the underlying population treatment assignment, $z$.  Versions of this statement often appear as assumptions about properly specified exposure mappings; see~\citet{aronow2017estimating, basse2019randomization}.
By itself, this singleton hypothesis $H_0^{\{\aa\}}$ is not typically testable. However, it is the building block for many more general hypotheses. We focus on two classes: (1) contrast hypotheses that compare two exposures, e.g., $\Hab$; and (2) intersection hypotheses that consider sets of singleton hypotheses, e.g., $\bigcap_{\aa\in\Edom} \Haa$. 

First, to compare two exposures, $\aa$ and $\bb$, we set $\Fset = \{\aa, \bb\}$, 
and refer to the resulting hypothesis $\Hab$ as a {\em contrast hypothesis}:
\begin{align}
\label{eq:H0_contrast}
\Hab : Y_i(z) = Y_i(z')~\text{for all}~i=1, \ldots, N,~\text{and any}~z, z'\in\Zdom~\text{such that}~f_i(z), f_i(z')\in\{\aa,\bb\}.
\end{align}
A subtle issue with such contrast hypotheses is that $\Hab$ implies 
both $\Haa$ and $\Hbb$. As such, rejecting $\Hab$ does not necessarily mean that 
treatment exposures $\aa$ are $\bb$ are different, but it could be that 
either $\Haa$ or $\Hbb$ is not true. We will return to this issue of interpretation in the application of Section~\ref{sec:FRT_med}.

\addtocounter{example}{-3}
\begin{example}[Clustered interference (cont.)]
In this setting, we consider testing $\Hab$ with 
$\aa = 0$ and $\bb = 1$; that is, 
we test whether there is a spillover effect on control units
from a treated unit in the same cluster.
%
\end{example}
\begin{example}[Spatial interference (cont.)]
In this setting, we consider testing $\Hab$ with 
$\aa = (0, 0)$ and $\bb = (1, 0)$ for some value of the distance threshold, $r$. That is, we test whether there is a spillover effect on an untreated unit from having one or more treated units within the specified distance, $r$~(e.g., $r=125$m in the application of Section~\ref{sec:spatial}).
\end{example}

An alternative direction is to instead consider intersections of singleton hypotheses. 
For example, the intersection hypothesis of all possible singletons, namely $\bigcap_{\aa\in\Edom} \Haa =: \Hex$, is a full {\em exclusion restriction} condition on the exposures:
\begin{align}\label{eq:H0_ex}
\Hex: Y_i(z) = Y_i(z')~\text{for all}~i=1, \ldots, N,~\text{and any}~z, z'\in\Zdom~\text{such that}~f_i(z) = f_i(z').
\end{align}
This covers many standard hypotheses under interference 
in the literature~\citep{toulis2013estimation, bowers2013reasoning, rosenbaum2007interference, aronow2012general, basse2019randomization,basse2019,athey2018exact}. 
In Section~\ref{sec:composite}, we show that our main procedure can also test intersection null hypotheses of the form $H_0^{\Fset_1} \cap  H_0^{\Fset_2}$, where $\Fset_1 \cap \Fset_2 = \emptyset$. 


\begin{example}[Extent of interference]\label{example_extent}
In the social network setting, we consider
the null hypothesis that for all $i=1, \ldots, N$:
\begin{equation}\label{eq:H0_extent}
 Y_i(z) = Y_i(z')~\text{for any}~z, z'\in\Zdom~\text{such that}~
 z_j=z_j'~\text{for all}~j=1, \ldots, N,~\text{for which}~d(i, j) \le k.
\end{equation}
This hypothesis posits that a unit's outcome may depend 
only on treatments of units up to $k$ hops away in the social network but no further.
This is standard in the literature; see, for example, Hypotheses 2, 3, \& 4 in \citet{athey2018exact}.
The hypothesis in Equation~\eqref{eq:H0_extent} is equivalent to $\Hex$ with an exposure function as described in Example~\ref{example_extent} of Section \ref{sec:setup_sub} above.
\end{example}

As we discuss next, the main challenge in testing these hypotheses
is that the treatment exposures in $\Fset$ are not independently manipulated in the experiment---the design only 
determines the joint distribution of the population treatment assignment, $Z$. This precludes a simple randomization test for $\Hf$ since we cannot impute the outcomes of 
units assigned to exposure levels other than those in  $\Fset$.\footnote{
In Appendix~\ref{appendix:H0}, we discuss an equivalent formulation 
of this imputability problem. Specifically, we express $\Hf$ as a composite null hypothesis on the full schedule of potential outcomes. 
Then, the imputability problem 
becomes essentially a problem of identification, where several alternative hypotheses may lead to a randomization distribution that is equal to the null. }
%
Similar to other approaches discussed below, we propose to address this issue by appropriately conditioning the randomization test.  A contribution of this paper is to constructively find such conditioning for general null hypotheses.
We give a general overview of conditional randomization tests, and in Section~\ref{sec:method} we describe our proposed conditioning method based on the concept of null exposure graphs to test $\Hf$. 

\subsection{Conditional randomization tests under interference: A review}
\label{sec:workhorse}

We briefly review the general framework proposed by \cite{basse2019} for valid
randomization tests under interference. This framework builds on the key insight first formulated by
\cite{aronow2012general} and developed by \cite{athey2018exact} that, although
the null hypothesis $\Hf$ is not sharp in general, it can be ``made sharp'' if we restrict
our attention to a well chosen subset of units, $\Uset\subseteq \Udom$, and subset of assignments, $\Zset\subseteq\Zdom$.
\citet{basse2019} formalized the idea as sampling a {\em conditioning event}, $C = (\Uset, \Zset)$, from a carefully constructed distribution $P(C | \Zobs)$, called a
{\em conditioning mechanism}, and then running a test conditional on $C$.
This requires the use of a test statistic {\em restricted} on $C$:
\begin{definition}[Restricted test statistic]
Let $C = (\Uset, \Zset)$ be a conditioning event, where $\Uset\subseteq \Udom$ and $\Zset\subseteq\Zdom$. A test statistic, 
$t(z, y; C) : \{0, 1\}^N \times \Real^n \to \Real$  is restricted on $C$ iff:
$$
t(z', y'; C) = t(z, y; C),~\text{for all}~y,y',\text{and } z,z'\in \Zset, ~\text{such that}~z_U'=z_U~\text{and}~y_U'=y_U.
$$
Here, subscript ``$U$'' denotes the corresponding subvector 
restricted only to units $U$ in event $C$.
\end{definition}

\begin{theorem}[\cite{basse2019}]
\label{thm1}
Let $H_0$ be a null hypothesis. 
For some conditioning event $C=(\Uset,\Zset)$, where $U\subseteq \Udom$ and $\Zset\subseteq \Zdom$, let $t(z, y; C)$ be a test statistic restricted on $C$. 
Suppose that 
the test statistic is conditionally imputable under $H_0$ given $C$, i.e.,  
$Y_U(z') = Y_U(z)$
for any $z, z'\in\Zset$ such that
$P(C|z) > 0~\text{and}~P(C | z') >0$.
The p-value obtained from the following procedure:
  \begin{enumerate}
  \item Draw $Z^{\obs} \sim \pr(Z^{\obs})$, and observe $Y^{\obs} = Y(Z^{\obs})$.
  \item Draw $C \sim P(C | \Zobs)$, and compute $T^{\obs} = t(\Zobs, Y^{\obs}; C)$.
  \item Compute
  $\text{p-value} = \EX
\left[\Ind{t(Z, Y^{\obs}; C) > T^{\obs}} | C \right]$,
    where the expectation is with respect to 
    the correct randomization distribution, $P(Z |C) \propto 
    P(C | Z) \pr(Z)$,
  \end{enumerate}
  is valid conditionally and marginally.
\end{theorem}
Conditional validity means that, under any conditioning event $C$, the test in the third step of Theorem 1 has the correct level (i.e., has Type-I error probability equal to $\alpha$ under the null)  with respect to the conditional randomization distribution of $\Zobs$ given $C$.
    Marginal validity means that the test has the correct level marginally over $C$, and is thus weaker than conditional validity. Indeed, it is possible that the the test does not have the correct level for some conditioning events, but nonetheless has correct Type-I error on average over $C$.
    On the other hand, if the test is conditionally valid then it is also marginally valid.
    
The conditioning event $C$ is an abstract device used to construct valid randomization tests and is generated from the conditioning mechanism $P(C | \Zobs)$. Moreover, the analyst can construct the conditioning mechanism of their choosing so that $C$ depends stochastically or deterministically on $\Zobs$.  A key contribution of this paper is to propose an approach such that $C$ is easily constructed.

We can see from Theorem~\ref{thm1} that there are two main challenges in constructing a conditioning mechanism that leads to valid conditional randomization tests.
First, the test statistic should be imputable under the null hypothesis $H_0$. Generally,
this means that based only on the observed value of the outcomes, $\Yobs$, we can compute the null distribution of the test statistic $t(z, Y(z); C) = t(z, \Yobs; C)$ that is induced by
the randomization distribution, $P(Z | C)$. 
Second, we must be able to draw samples from this randomization distribution, given by its conditional-marginal decomposition $P(Z|C) \propto P(C |Z) \pr(Z)$ in the third step. Ensuring that
this distribution is computationally tractable can be challenging~\citep{basse2019randomization, basse2019}.

In addition, the null hypotheses of interest are defined for all units $i = 1, \ldots, N$, even if they are defined for a subset of assignments, e.g., $f_i \in \Fset$. As a result, rejecting a null hypothesis for a subset of units logically implies rejecting that same null hypothesis for all units.
Thus, rejecting a null hypothesis with a conditional randomization test has the same interpretation as rejecting a null hypothesis with a (possibly infeasible) unconditional test. In addition, this interpretation holds regardless of the specific conditioning event, although the power of the test might depend on the conditioning mechanism, as we discuss in Section \ref{sec:power_main}.

To gain more intuition, we can also describe the approaches of \cite{basse2019}, \cite{athey2018exact} and \cite{aronow2012general}
within this framework, each corresponding to different choices of the conditioning mechanism. 
In particular,~\citet{basse2019} propose a conditioning mechanism under clustered
interference, such that the implied randomization distribution, $P(Z | C)$, leads
to a permutation test, which is easy to implement. However, their approach does not readily generalize to other settings, such as spatial
interference. The methods of~\citet{athey2018exact} and \citet{aronow2012general} correspond to
conditioning mechanisms of the form $P(C | Z) = P(C)$, where conditioning is either random, or guided by
known auxiliary information, but is not conditioned on the observed assignment. In contrast to the approach of \citet{basse2019}, these methods can be applied in general interference settings. However, they are usually underpowered because they do not use the observed assignment
to do the conditioning; as an extreme example, these approaches cannot test $\Hf$ if there are no units in $C$ 
exposed to any exposure value in $\Fset$. This may happen because $C$ is randomly sampled, and does not use the exposure information 
in $\Zobs$.
We discuss this issue further in Section~\ref{sec:comparison}.

In this paper, we propose a testing method that is both general and powerful.
 In the following sections, we develop the core concepts of our approach, 
 using the framework of conditioning mechanisms presented here.
 Our main contribution is an algorithm that automatically constructs a tractable conditioning mechanism, $P(C|Z)$, through the concept of null exposure graph presented in the upcoming section.
Our proposed randomization test for $\Hf$ is presented later in Section~\ref{sec:test}, and in Section~\ref{sec:comparison} we 
follow up on this discussion of related methods and describe the benefits of our approach.

\section{The null exposure graph and bicliques}
\label{sec:method}
We now introduce some preliminary concepts underlying our test for
$\Hf$ in~Equation~\eqref{eq:H0}. 
The first key idea is to represent the imputability of outcomes under the null
hypothesis through a graph between units and assignments, which we call
{\em the null exposure graph}.
The conditioning event $C$ will then be taken to be a biclique in that graph, 
and the conditioning mechanism $P(C | \Zobs)$ will determine the biclique that
contains $\Zobs$. This transforms the analytical task of defining $P(C | \Zobs)$ into a computational task.

\subsection{The null exposure graph}

\label{sec:null_graph}
The first component of our method is the null exposure
graph, which is a graph that encodes the units' treatment
exposures under different population treatment assignments. 
As we will show in the next section, the null exposure graph determines the appropriate conditioning for the randomization
test of $\Hf$.
\begin{definition}[Null exposure graph]
\label{def:nullG}
Let $\Udom=\{1, \ldots, N\}$ and $\Zdom=\{z_1, \ldots, z_J\}$ denote the sets
of units and assignments, respectively. Define the vertex set as
$V = \Udom\cup \Zdom$, and the edge set as 
\begin{equation}\label{eq:E}
E = \big\{(i, z) \in \Udom\times\Zdom : f_i(z)\in\Fset\big\}.
\end{equation}
That is, an edge between unit $i$ and assignment $z$ exists if and only if $i$
is exposed to $\Fset$ under $z$. There are no edges
between units or between assignments. Then, $\GfF=(V, E)$ is the null exposure
graph of $\Hf$ with respect to exposure mapping $f$ and exposure set $\Fset$.
\end{definition}

The null exposure graph is a bipartite graph since there are not edges between units or between assignments. In order to visualize the null exposure graph in a concrete setting, we return to the example of clustered interference~(Example 1).
\addtocounter{example}{-3}
\begin{example}[Clustered interference (cont.)]
For simplicity, suppose we have four units in two clusters, namely $\{1,2\}$
and $\{3, 4\}$; and  that we treat exactly one cluster leading to four possible  assignment vectors, depending on which unit is treated within the treated cluster.
Following Section~\ref{sec:setup}, 
\begin{align}
f_1(z) = 2 z_1 + z_2,~f_2(z) = 2 z_2 + z_1;\nonumber\\
f_3(z) = 2 z_3 + z_4,~f_4(z) = 2 z_4 + z_3.
\end{align}
We test $H_0^{\{0, 1\}}$, i.e., whether outcomes are equal between exposures
$f_i=0$ and $f_i=1$.
Figure~\ref{fig2} displays the null exposure graph for this scenario.  
For example, when $z=(1, 0, 0, 0)$---denoted as population assignment ``1''---unit 1 is treated, and so 
the exposures are as follows:~$f_1=2, f_2=1, f_3=0, f_4=0$.
Since only units 2, 3, and 4 are exposed to the exposure levels in the null hypothesis, we draw edges from assignment ``1'' only to those units.  
This process repeats for all assignments to produce the null exposure graph shown in Figure~\ref{fig2}~(left).
On the right side of Figure~\ref{fig2}, the blue edges connecting units 1 and 4 to assignments ``2'' and ``3'' highlight a complete bipartite subgraph (biclique) of the null exposure graph.  Importantly, within this biclique, all missing potential outcomes are imputable under the null. 
\end{example}

\begin{figure}[t!]

\centering
\begin{minipage}{.45\textwidth}
\begin{tikzpicture}[thick,
  every node/.style={draw,circle},
  fsnode/.style={fill=darkgray},
  ssnode/.style={fill=darkgray},
  every fit/.style={ellipse,draw,inner sep=-2pt,text width=2cm}, ,shorten >= 3pt,shorten <= 2pt
]

\begin{scope}[start chain=going below,node distance=2mm]
\foreach \i in {1,2}
  \node[fsnode,on chain] (f\i) [label=left: \i] {};
\end{scope}

\begin{scope}[yshift=-1.5cm, start chain=going below,node distance=2mm]
\foreach \i in {3,4}
  \node[fsnode,on chain] (f\i) [label=left: \i] {};
\end{scope}

\begin{scope}[xshift=3cm,yshift=0.45cm,start chain=going below,node distance=7mm]

  \node[ssnode,on chain] (s1) [label=right: {1 (1,0,0,0)}] {};
  \node[ssnode,on chain] (s2) [label=right: {2 (0,1,0,0)}] {};
  \node[ssnode,on chain] (s3) [label=right: {3 (0,0,1,0)}] {};
  \node[ssnode,on chain] (s4) [label=right: {4 (0,0,0,1)}] {};

\foreach \i in {1,2,3,4}
  ;
  \end{scope}

\node [label={[label distance=0.52cm]90:Units} ]  {};
\node [label={[label distance=2cm]13:Assignments} ] {};

\draw[ultra thick,gray] (f2) -- (s1);
\draw[ultra thick,gray] (f3) -- (s1);
\draw[ultra thick,gray] (f4) -- (s1);
\draw[ultra thick,gray] (f1) -- (s2);
\draw[ultra thick,gray] (f3) -- (s2);
\draw[ultra thick,gray] (f4) -- (s2);
\draw[ultra thick,gray] (f1) -- (s3);
\draw[ultra thick,gray] (f2) -- (s3);
\draw[ultra thick,gray] (f4) -- (s3);
\draw[ultra thick,gray] (f1) -- (s4);
\draw[ultra thick,gray] (f2) -- (s4);
\draw[ultra thick,gray] (f3) -- (s4);
\end{tikzpicture}
\end{minipage}
\begin{minipage}{.45\textwidth}
\begin{tikzpicture}[thick,
  every node/.style={draw,circle},
  fsnode/.style={fill=darkgray},
  ssnode/.style={fill=darkgray},
  every fit/.style={ellipse,draw,inner sep=-2pt,text width=2cm}, ,shorten >= 3pt,shorten <= 2pt
]

\begin{scope}[start chain=going below,node distance=2mm]
\foreach \i in {1,2}
  \node[fsnode,on chain] (f\i) [label=left: \i] {};
\end{scope}

\begin{scope}[yshift=-1.5cm, start chain=going below,node distance=2mm]
\foreach \i in {3,4}
  \node[fsnode,on chain] (f\i) [label=left: \i] {};
\end{scope}

\begin{scope}[xshift=3cm,yshift=0.45cm,start chain=going below,node distance=7mm]

  \node[ssnode,on chain] (s1) [label=right: {1 (1,0,0,0)}] {};
  \node[ssnode,on chain] (s2) [label=right: {2 (0,1,0,0)}] {};
  \node[ssnode,on chain] (s3) [label=right: {3 (0,0,1,0)}] {};
  \node[ssnode,on chain] (s4) [label=right: {4 (0,0,0,1)}] {};

\foreach \i in {1,2,3,4}
  ;
  \end{scope}

\node [label={[label distance=0.52cm]90:Units} ]  {};
\node [label={[label distance=2cm]13:Assignments} ] {};

\draw[ultra thick,gray] (f2) -- (s1);
\draw[ultra thick,gray] (f3) -- (s1);
\draw[ultra thick,gray] (f4) -- (s1);
\draw[ultra thick,blue,line width=3pt] (f1) -- (s2);
\draw[ultra thick,gray] (f3) -- (s2);
\draw[ultra thick,blue,line width=3pt] (f4) -- (s2);
\draw[ultra thick,blue,line width=3pt] (f1) -- (s3);
\draw[ultra thick,gray] (f2) -- (s3);
\draw[ultra thick,blue,line width=3pt] (f4) -- (s3);
\draw[ultra thick,gray] (f1) -- (s4);
\draw[ultra thick,gray] (f2) -- (s4);
\draw[ultra thick,gray] (f3) -- (s4);
\end{tikzpicture}
\end{minipage}
\vspace{-5px}
\caption{{\em Left:}~Depiction of the null exposure graph for a clustered interference setting with four units and four assignments. The left nodes represent the experimental units, and the right nodes represent the population treatment assignments.  There are only four assignments since we consider treating exactly one cluster.  The graph is bipartite because no units and assignments are connected with other like nodes. 
{\em Right:}~One biclique in the null exposure graph is highlighted in blue.
}

\label{fig2}
\end{figure}

\subsection{Bicliques and biclique decompositions}\label{sec:cliques}

The ``completeness'' of bicliques in  the null exposure graph means that, within a biclique, all units are connected to all assignments.  
As we note above, this implies that for the units and assignments that comprise the biclique, we can impute all potential outcomes for exposures $\Fset$ under $\Hf$.
Below, we give definitions for these important objects and discuss how to partition the set of treatment assignments in the null exposure graph.  The resulting decomposition will be essential for the proposed randomization test. We describe algorithms for finding bicliques in Section \ref{sec:decomposition}.

\begin{definition}[Biclique]\label{def:decomposition}
	A biclique in the null exposure graph, $\GfF=(V, E)$, where $V= \Udom \cup \Zdom$
	and $E$ is defined in Equation~\eqref{eq:E}, is a set-pair $C = (U, \Zset)$, with $U\subseteq \Udom$ and $\Zset\subseteq \Zdom$, 
	such that $(i, f_i(z))\in E$, for every $i\in U$ and every $z\in\Zset$.
	\end{definition}
As an example, $C=(\{1, 4\}, \{z_2, z_3\})$ is a biclique in Figure~\ref{fig2} since it is a complete bipartite subgraph.
Note that we use the same notation ``$C$'' for bicliques as we did for conditioning events in Section~\ref{sec:workhorse}.
This is intentional since our proposed test~(in Section~\ref{sec:test}) will condition on a biclique of the null exposure graph.  

We now formalize the intuition that bicliques in the null exposure graph allow imputation of the missing potential outcomes.

\begin{proposition}\label{prop:clique}
Consider a null exposure graph, $\GfF$, with some biclique $C = (U, \Zset)$.
If $\Zobs\in\Zset$, then
$Y_i(z) = Y_i(\Zobs)$ under $\Hf$, for all $i\in U$ and all $z\in\Zset$.
\end{proposition}
\begin{proof}
For any unit $i\in \Uset$ in the biclique, $f_i(\Zobs)\in\Fset$ since $\Zobs$ and $i$ are both in the biclique.
Fix any other assignment $z\in\Zset$ in the biclique, 
then there is also an edge between $i$ and $z$ by definition of the biclique.
This implies that $f_i(z)\in\Fset$ as well, by construction of the null exposure graph in Definition~\ref{def:nullG}. Hence, $f_i(z)\in\Fset$ as well, and so $Y_i(z) = Y_i(\Zobs)$ under $\Hf$  of Equation~\eqref{eq:H0}.
\end{proof}

Proposition~\ref{prop:clique} shows that we can condition our test on a biclique of $\GfF$ because we can impute all the missing potential outcomes in a biclique that contains the observed treatment assignment $\Zobs$. Since there exist many such bicliques in $\GfF$, 
we need to decide  how to condition on one through 
an appropriate conditioning mechanism, $P(C| Z)$.
Choosing this mechanism is not a trivial task, however, because in order to calculate the randomization distribution, $P(z | C)$, we also need to calculate $P(C|z)$, for any assignment $z$ in $C$. This can lead to meaningful statistical and computational challenges, even 
under seemingly reasonable choices of the conditioning mechanism.

 
In this paper, our approach is designed so that the conditioning mechanism operates on a restricted set of bicliques, such that, for each assignment $z\in\Zdom$, there is only one biclique that contains $z$.
The benefit of this approach is that conditioning on the unique biclique containing $\Zobs$ yields a simple mechanism of the form $P(C|\Zobs) = \Ind{\Zobs \in \Zset(C)}$. This approach only needs to test membership of $Z$ in biclique $C$, and, consequently, also yields a simple randomization test in which we only need to randomize the assignments in $C$ weighted by the design, $P(Z|C) \propto \Ind{Z \in \Zset(C)} P(Z) $.

Such special construction of the conditioning mechanism
can be implemented through the concept of a \emph{biclique decomposition}, which is a set of bicliques in the null exposure graph that fully partitions the population set of assignments, $\Zdom$. This concept is defined formally as follows.


\begin{definition}[Biclique Decomposition]

Let $d : \Zdom \to \mathbb{N}$ denote the degree of assignment $z$ in the null exposure graph.
Define $\Zdom_\Fset = \{z \in \Zdom: d(z) > 0\}$ as the set of all assignments connected to at least one unit in the null exposure graph.
A biclique decomposition, $\mathcal C = \{C_1, \ldots, C_K\}$, of the null exposure graph in Definition~\ref{def:nullG} is a finite
set of bicliques, $C_k=(\Uset_k, \Zset_k)$, $k=1, \ldots, K$, such that $\Zdom_\Fset$ is partitioned, i.e.,
$$
\bigcup_k \Zset_k = \Zdom_\Fset,~\text{and}~\Zset_k\cap \Zset_{k'} = \emptyset,~\text{for any}~k \neq k'.
$$
	\end{definition}

Notably, it is not necessary that the population of units, $\Udom$, is partitioned in a biclique decomposition. This is crucial because partitioning both $\Udom$ and $\Zdom_\Fset$ may not be possible, or could 
lead to low-powered tests.


\section{Biclique-based randomization tests}
\label{sec:test}

\subsection{Main method and test validity}

We can now describe our proposed conditional randomization test for $\Hf$ in Equation~\eqref{eq:H0},
which is the key methodological contribution of this paper. Throughout, let $\Cset$ be some fixed biclique decomposition of the null exposure graph $\GfF$. 
Consider the following procedure:
\begin{procedure} \label{proc:test}
For observed assignment $\Zobs \sim P(\Zobs)$:
\begin{enumerate}
	\item Find the unique biclique, $C=(U, \Zset)\in\Cset$, such that $\Zobs\in\Zset$. Consider 
	a test statistic $t(z, y; C)$ restricted to biclique $C$.
\item Calculate the observed value of the test statistic, 
	$\Tobs = t(\Zobs, \Yobs; C)$. 
\item Define the randomization distribution as
$r(Z) \propto \Ind{Z\in\Zset}\cdot \pr(Z).$
	\item Define the randomization p-value as follows:
	\begin{equation}
	\label{eq:pval}
	\mathrm{pval}(\Zobs, \Yobs; C) = 
  \EX_{Z\sim r}\left[\Ind{ t(Z, \Yobs; C) > \Tobs }\right].
	\end{equation}
\end{enumerate}
\end{procedure}
\noindent We describe an algorithm for choosing a unique biclique (Step 1) in Section \ref{sec:decomposition}.
The following theorem shows that this procedure is valid; the proof is in Appendix~\ref{appendix:proofs}.
%
\begin{theorem}\label{thm:main}
\ThmOne
\end{theorem}
Theorem~\ref{thm:main}~follows from Theorem~\ref{thm1} by recognizing that
 Procedure~\ref{proc:test} describes a conditional randomization test in which the conditioning
event is a biclique $C \in \Cset$, and the conditioning mechanism is
defined as $P\{C=(\Uset, \Zset) | Z\} = \Ind{C \in \Cset} \Ind{Z \in \Zset}$. The proof first
verifies that any test statistic restricted on $C$ is imputable: this
follows from the construction of the biclique and Proposition~\ref{prop:clique}. It then shows
that the randomization distribution $r(Z)$ defined in Step~3 of the procedure is the correct conditional distribution $P(Z | C)$ implied by the 
design and the conditioning mechanism.

The computational tractability of
the randomization distribution,
$r(Z) \propto \Ind{Z \in \Zset} P(Z)$, 
is immediate provided that we can compute $P(Z)$ and enumerate the assignments
$\Zset$ in any biclique $C \in \Cset$. This last condition, however, may be prohibitive
if the support of the design $P(Z)$ is too large 
 since biclique enumeration is NP-hard. 
Fortunately, a small modification of our test can address this issue. The idea is to add a 
preliminary step in Procedure~\ref{proc:test} that subsamples assignments to limit the size of $\Zdom$. 
We describe this modification in Appendix~\ref{appendix:proofs} and show that the procedure is still valid.


Finally, while Procedure~\ref{proc:test} automates the construction of a conditioning mechanism, it
still allows flexibility in the choice of the test statistic. 
For instance, to test $\Hab$, the simplest choice is for the test statistic to denote the difference-in-means between outcomes of units in $C$ exposed to $\aa$ and $\bb$.
However, we may improve power by using test statistics motivated by a model of interference, such as a network regression models with spillovers, which may fit the data better than standard linear regression. See
also~\citet[Section 5.3]{athey2018exact} for an excellent related discussion. We turn to power in Section \ref{sec:power_main}.

\subsection{Comparison to related work}\label{sec:comparison}
In Section \ref{sec:workhorse}, we discussed how our method, along with
those of \cite{aronow2012general} and \cite{athey2018exact}, could be described
within the general framework of \cite{basse2019}. We can also describe these approaches using the framework in this paper, in which each method corresponds to a different approach for selecting bicliques from the null exposure graph.

%
The method of~\citet{basse2019}, for instance, can be viewed as implicitly
considering bicliques of the null exposure graph with possibly overlapping assignments.  In other words, assignments in the bicliques form a covering---not a partition---of
$\Zdom$, such that an assignment may belong to more than one biclique. The
conditioning mechanism is then uniform on the set of all bicliques containing the
observed assignment. This approach works in their particular setting and
results in powerful tests because the conditioning is guided by the observed
assignment, $\Zobs$.  However, the drawback of this approach is that there is
no general way to construct good biclique coverings, instead
requiring case-by-case derivations. Specifically, the covering that is implied
by the conditioning mechanism constructed in \cite{basse2019} works only in
two-stage experiments with clustered interference. 

The approach of \cite{athey2018exact} applies to more general settings, but it may lead to underpowered tests. To understand their approach, we introduce some additional notation.
For any $U \subseteq \Udom$, let  $C(U; z)$ denote the largest biclique of the null exposure
graph that contains only units from $U$ and also contains assignment $z$. Then, the conditioning
mechanism implicitly considered by \citet{athey2018exact} is of the form:
\begin{equation}\label{eq:athey}
  P(C | Z) = P(U) \Ind{C = C(U; Z)},
\end{equation}
where $P(U)$ is specified by the analyst. In other words, the approach of \citet{athey2018exact}
implicitly suggests first to sample units from one side of the bipartite null exposure graph, 
and then to calculate the induced biclique that contains the observed assignment. This approach is
generally applicable, but the random choice of $U$ may lead to underpowered or even
ill-defined tests. This is the case when, for example, $C(U; Z)$ is empty
due to a poor initial choice of $U$.

Our method combines the benefits of both approaches. Unlike the biclique covering
strategy implicit in \cite{basse2019}, our approach is not problem-specific and automates the
construction of conditioning mechanisms.
Also, in contrast to the method in \cite{athey2018exact}, our proposed approach gives concrete guidance on how to properly condition the test to achieve higher power. We discuss these advantages in the context of clustered and spatial interference in Sections \ref{sec:clust} and \ref{sec:spatial}.


\section{Power for biclique-based randomization tests}
\label{sec:power_main}

While the proposed randomization test is guaranteed to be exact, the power of the test depends on many factors, including the size of the study,  the interference structure, and the test statistic. 
In this section, we explore this problem via a theoretical analysis, turning to simulations in later sections. We then propose a \emph{biclique decomposition algorithm} that improves power by finding larger cliques. In Section \ref{sec:design_cliques}, we show how we can further leverage knowledge of the design to improve power.


\subsection{Theoretical results for statistical power}\label{sec:power}
In this section, we analyze the power of the main biclique randomization test of Procedure~\ref{proc:test}. 
%
As mentioned earlier, power results for Fisher randomization tests is challenging because power ultimately
depends on the full schedule of potential outcomes---very few results are available, even
in the simple, non-conditional case~\citep{rosenbaum2010design}. A standard approach in power analyses for
randomization tests is therefore to consider a simplified
model of the problem for which formal results can be obtained~\citep[Chapter 15]{lehmann2006testing}. These results then serve
as a foundation for useful heuristics.
In this spirit, we consider a simplified model solely for the purpose of assessing power under certain conditions --- if these conditions do not hold, the underlying randomization test is still guaranteed to be exact. 

\begin{theorem}\label{thm:power}
\ThmPower
\end{theorem}

The proof of Theorem~\ref{thm:power} is in Appendix~\ref{app:powertheory}, and can be summarized as follows. We start with the random variable $t(Z, Y(Z); C)$ ---i.e., the test statistic--- as induced
by the joint distribution of $(Z, Y(Z), C)$ conditional on $C$.
The distribution of the test statistic is denoted by $\hat F_{1, n, m}$.
Under the null, the distribution of
$t(Z, Y(Z); C)$ is the same as that of $t(Z, \Yobs; C)$, and is denoted by $\hat F_{0, n, m}$.

The first condition, $(A1)$, connects these two distributions to marginal and continuous  distributions denoted, respectively, by $F_{0, n}$ and $F_{1,n}$. In particular, this condition claims that the null (or alternative) distribution consists of $m$ i.i.d. samples from the corresponding 
marginal. The main restriction underlying this condition is that 
the actual biclique $C$ we condition on does not matter for any biclique decomposition, $\Cset$.
It is perhaps the least realistic condition, as we
expect the test statistic values to be correlated within a biclique $C$, and across 
different bicliques within the same decomposition $\Cset$. 
Note, however, that (A1) is only useful to leverage results from empirical process theory, 
and thus derive asymptotic rates for the testing power as in Equation~\eqref{eq:phi_power}.
The other two conditions are more typical in randomization test analysis.
For instance, the second
condition, $(A2)$, posits that the number of units, $n$, affects the null distribution only as a scale parameter.
Condition $(A3)$ posits a constant spillover effect, and thus 
relates the null distribution of the test statistic with its distribution under the alternative. 
%
As mentioned earlier, these conditions describe a simplified model --- neither can
be strictly true in our randomization-based framework, but they can be useful approximations
for large $m$ and will allow us to derive helpful power heuristics.
Indeed, the empirical studies
in Section~\ref{sec:clustered_sim} suggest that even under potential violations of these conditions,
the theoretical predictions from Theorem~\ref{thm:power} remain valid.

Under these conditions, Theorem~\ref{thm:power} shows that power is increasing in the size
of the biclique. Specifically, the number of focal units ($n$) in the biclique controls the
``sensitivity'' of the test, that is, how quickly the test achieves maximum power (as a
function of the treatment effect). The number of focal assignments ($m$) in the biclique
instead affects the maximum power of the test. See Section~\ref{appendix:confirm} in
the Appendix for a simulation study that confirms these results. These insights provide
useful guidance for tuning and diagnosing clique decomposition algorithms.

Finally, although Equation~\eqref{eq:power} shows the relationship between the power of the 
biclique test and the biclique size, it is not clear how this relates to properties the null exposure graph.
In Appendix \ref{app:implications}, we leverage an important result from extremal graph theory to show that the density of the null exposure graph relates to the existence of bicliques of certain size.


\newcommand{\EE}{\mathrm{E}}
\newcommand{\UU}{\mathrm{U}}
\newcommand{\ZZ}{\mathrm{Z}}
\subsection{Biclique decomposition algorithm and power considerations}
\label{sec:decomposition}

Given Theorem \ref{thm:power}, an important step for improving power is to condition on large bicliques in Step 1 of Procedure~\ref{proc:test}.
We now present an algorithm for decomposing the null exposure graph into bicliques, which in turn enables finding larger bicliques.
%
For a given biclique, $C$, in the null exposure graph, let $\EE(C)$, $\UU(C)$, and $\ZZ(C)$ denote the set of edges, the set of units, and the set of assignments in the biclique, respectively. Thus, $|\mathrm{E}(C)| = |\mathrm{U}(C)| |\mathrm{Z}(C)|$ since the biclique is bipartite. 
 Also, let $C\in G$ indicate that $C$ includes nodes and edges from null exposure graph $G$.

Our proposed biclique decomposition algorithm can be described as follows: 
\begin{enumerate}
\item Start with an empty biclique set: $\Cset = \{\}$, 
and the original null exposure graph $G=\GfF$ that 
corresponds to the null hypothesis, $\Hf$.
\item Solve the ``largest biclique problem'':
\begin{equation}
\label{eq:max_clique}
C^* = \arg\max_{C \in G}~|E(C)|.
\end{equation}
\item Remove biclique edges: $E(G)\gets  E(G) \setminus E(C^*)$.
\item Remove biclique assignments: $Z(G)\gets  Z(G) \setminus Z(C^*)$.
\item Update biclique set, $\Cset \gets \Cset\cup \{C^*\}$, 
and repeat from Step 2 if  $|E(G)| > 0$.
\end{enumerate}

\noindent The output of this procedure is a biclique decomposition, $\Cset$.

The main computational challenge is in Equation~\eqref{eq:max_clique}, 
where we calculate bicliques with the
largest possible number of edges in the remaining null exposure graph. 
 This is a variant of the maximal edge biclique problem, and is computationally challenging.  In fact,~\citet{peeters2003maximum} show that finding such bicliques is NP-hard. See \citet{zhang2014finding} for a review.
Fortunately, our randomization test remains valid even when the solution to~\eqref{eq:max_clique} is approximate, so we do not need to solve Equation~\eqref{eq:max_clique} exactly. In this paper, we use the ``binary inclusion-maximal biclustering'' ({\tt Bimax}) method~\citep{prelic2006systematic} for such approximation of \eqref{eq:max_clique}.\footnote{This is a fast divide-and-conquer method to find sub-blocks of ones in a binary matrix, known as biclusters, mainly used in the bioinformatics and gene expression literature. 
The algorithm is implemented in the R package {\tt biclust} and can incorporate constraints on the solution space.  The constraints are placed on the number of units and assignment nodes so that {\tt Bimax} returns all bicliques $C^*$ where, for example, $|\mathrm{U}(C^*)| \geq n_0, |\mathrm{Z}(C^*)| \geq n_1$ for specified, integer-valued 
parameters $n_0$ and $n_1$.  
In practice, to approximate the solution of Equation~\eqref{eq:max_clique}, and to produce balanced bicliques, we suggest choosing $n_0$ and $n_1$ so that they are similar to each other in magnitude, and $n_0\cdot n_1$ is large.
}
In the following section, we show that using this algorithm within the proposed method results in a powerful test.

Finally, we note that these power considerations are also useful for optimal experimental design. 
Specifically, one can parameterize the design, simulate outcomes ---according to an outcome model, as mentioned above--- and then estimate the power of the biclique test for a treatment effect of interest. The design with highest power can then be chosen for the experiment.
We illustrate this idea in Appendix~\ref{app:design} in the context of the Medell{\'i}n data.
In particular, we consider a design space with two parameters: one parameter controls the treatment probability at the ``center'' of the city; and the other controls the treatment probability at the outskirts. The resulting optimization problem appears to be convex, suggesting that this approach may be more useful for experimental design under interference. We leave this for future work.

%

\section{Application to clustered interference: A simulation study}\label{sec:clust}

In this section, we illustrate the proposed biclique test in settings with clustered interference, continuing our running Example~\ref{example:clustered}.
In these settings, we assume that interactions between units occur within, but not between, well-defined clusters of units, such as households or classrooms. 

\subsection{Problem setup and comparison of available methods}
\label{sec:clust_study}
Following Example~\ref{example:clustered}, we have $N$ units divided equally into $K$ clusters.
In the experiment, we first treat $K/2$ clusters at random and then randomly treat one unit in each treated cluster.  
Here, the exposure function is defined as $f_i(z) = z_i + \sum_{j\in[i]} z_j$, where $[i]$ denotes all units in $i$'s cluster.
Thus, each unit $i$ has three 
possible exposure levels, namely $0, 1$ and $2$, which we describe as ``control,'' ``spillover'', and ``treated'' exposures, respectively. 
Overloading notation for simplicity, we denote the three potential outcomes for each unit as 
$Y_i(\text{``control''}), Y_i(\text{``spillover''})$, 
and $Y_i(\text{``treated''})$, respectively.

We focus on the following spillover effect hypothesis,
\begin{equation}\label{eq:H0cluster}
H_0	: ~~ Y_i(\text{``spillover''}) = Y_i(\text{``control''}) +\tau,~\text{for all}~i,
\end{equation}
and vary $\tau$. We assess validity when $\tau=0$ and power when $\tau\neq 0$. 
The null hypothesis in Equation~\eqref{eq:H0cluster} is therefore 
an instance of $\Hab$ in Equation~\eqref{eq:H0_contrast}, with $\aa=1$ and $\bb=0$.
Following our discussion in Section~\ref{sec:comparison}, we compare three methods for testing this null:
\begin{enumerate}[(i)]
	\item The biclique test proposed in this paper~(Procedure~\ref{proc:test});
	\item The method of~\citet{basse2019}, which samples one focal unit per cluster among those 
	who are not treated~(``conditional focals'');
	\item The method of~\citet{athey2018exact}, which samples one focal unit per cluster at random~(``random focals'').
\end{enumerate}

In their current form with one focal unit per cluster, methods (ii) and (iii) can be implemented as permutation tests. 
We could select multiple focals per cluster for these methods, but the resulting test is difficult 
to implement via a permutation test.
Notably, the biclique method avoids such implementation issues~(see also Section~\ref{sec:clustered_sim}).
Figure~\ref{fig2_comparison} illustrates how these tests differ in a hypothetical example with six units arranged in two clusters, namely, 
 $\{1, 2, 3\}$  and $\{4, 5, 6\}$, and where 
unit 6 is assigned to treatment~(colored green).
For unit 6, we observe the ``treated'' potential outcome, $Y_6(\text{``treated''})$; thus, unit 6 is not connected to any other assignment in the null exposure graph, since we have no information about 
$Y_6(\text{``control''})$ or $Y_6(\text{``spillover''})$ under the null.
For units 4 and 5, we observe the ``spillover'' potential outcomes; and 
for units 1, 2, and 3 in the first cluster, we observe their ``control'' outcomes.
Under the null, we can impute the ``control'' potential outcomes for units 4 and 5 and the ``spillover'' potential outcomes for units 1, 2, and 3.

\begin{figure}[t!]
\centering
\hspace{1mm}
\begin{tikzpicture}[thick,
  every node/.style={draw,circle},
  fsnode/.style={fill=darkgray},
  tnode/.style={fill=mygreen},
  ssnode/.style={fill=darkgray},
  enode/.style={fill=none,draw=none,minimum height=5cm,inner sep=1cm},
  every fit/.style={ellipse,draw,inner sep=-1pt,text width=1.8cm}, ,shorten >= 2pt,shorten <= 2pt
]

\begin{scope}[start chain=going below,node distance=1mm]
\foreach \i in {1}
  \node[fsnode,on chain] (f\i) [label=left: \i] {};
\end{scope}

\begin{scope}[yshift=-0.5cm,start chain=going below,node distance=1mm]
\foreach \i in {2,...,3}
  \node[fsnode,on chain] (f\i) [label=left: {*\i}] {};
\end{scope}

\begin{scope}[yshift=-2cm, start chain=going below,node distance=1mm]
\foreach \i in {4,5}
  \node[fsnode,on chain] (f\i) [label=left: *\i] {};
\end{scope}

\begin{scope}[yshift=-3cm, start chain=going below,node distance=1mm]
\foreach \i in {6}
  \node[tnode,on chain] (f\i) [label=left: \i] {};
\end{scope}

\begin{scope}[xshift=2.2cm,yshift=-0cm,start chain=going below,node distance=1mm]
\foreach \i in {1, 2,...,8}
  \node[ssnode,on chain] (s\i) [label=right: \i] {};
\end{scope}

\begin{scope}[xshift=2.2cm,yshift=-3.8cm,start chain=going below,node distance=1mm]
\foreach \i in {9}
  \node[enode,on chain] (s\i) {$\vdots$};
\end{scope}

\node [white,fit=(f1) (f3), label={[label distance=-0cm]90:Units} ]  {};

\node [white,fit=(s1) (s8), label={[label distance=-1.1cm]90:Assignments} ] {};

\draw[ultra thick,gray] (f2) -- (s1);
\draw[ultra thick,gray] (f2) -- (s2);
\draw[ultra thick,gray] (f2) -- (s3);
\draw[ultra thick,gray] (f2) -- (s4);
\draw[ultra thick,gray] (f2) -- (s5);
\draw[ultra thick,gray] (f2) -- (s6);
\draw[ultra thick,gray] (f2) -- (s7);
\draw[ultra thick,gray] (f2) -- (s8);
\draw[ultra thick,gray] (f3) -- (s1);
\draw[ultra thick,gray] (f3) -- (s2);
\draw[ultra thick,gray] (f3) -- (s3);
\draw[ultra thick,gray] (f3) -- (s4);
\draw[ultra thick,gray] (f3) -- (s5);
\draw[ultra thick,gray] (f3) -- (s6);
\draw[ultra thick,gray] (f3) -- (s7);
\draw[ultra thick,gray] (f3) -- (s8);
\draw[ultra thick,gray] (f4) -- (s1);
\draw[ultra thick,gray] (f4) -- (s2);
\draw[ultra thick,gray] (f4) -- (s3);
\draw[ultra thick,gray] (f4) -- (s4);
\draw[ultra thick,gray] (f4) -- (s5);
\draw[ultra thick,gray] (f4) -- (s6);
\draw[ultra thick,gray] (f4) -- (s7);
\draw[ultra thick,gray] (f4) -- (s8);
\draw[ultra thick,gray] (f5) -- (s1);
\draw[ultra thick,gray] (f5) -- (s2);
\draw[ultra thick,gray] (f5) -- (s3);
\draw[ultra thick,gray] (f5) -- (s4);
\draw[ultra thick,gray] (f5) -- (s5);
\draw[ultra thick,gray] (f5) -- (s6);
\draw[ultra thick,gray] (f5) -- (s7);
\draw[ultra thick,gray] (f5) -- (s8);
\end{tikzpicture} \hspace{-8mm}	
\begin{tikzpicture}[thick,
  every node/.style={draw,circle},
  fsnode/.style={fill=darkgray},
  tnode/.style={fill=mygreen},
  ssnode/.style={fill=darkgray},
  enode/.style={fill=none,draw=none,minimum height=5cm,inner sep=1cm},
  every fit/.style={ellipse,draw,inner sep=-1pt,text width=1.8cm}, ,shorten >= 2pt,shorten <= 2pt
]

\begin{scope}[start chain=going below,node distance=1mm]
\foreach \i in {1}
  \node[fsnode,on chain] (f\i) [label=left: \i] {};
\end{scope}

\begin{scope}[yshift=-0.5cm, start chain=going below,node distance=1mm]
\foreach \i in {2}
  \node[fsnode,on chain] (f\i) [label=left: *\i] {};
\end{scope}

\begin{scope}[yshift=-1cm, start chain=going below,node distance=1mm]
\foreach \i in {3}
  \node[fsnode,on chain] (f\i) [label=left: \i] {};
\end{scope}

\begin{scope}[yshift=-2cm, start chain=going below,node distance=1mm]
\foreach \i in {4}
  \node[fsnode,on chain] (f\i) [label=left: *\i] {};
\end{scope}

\begin{scope}[yshift=-2.5cm, start chain=going below,node distance=1mm]
\foreach \i in {5}
  \node[fsnode,on chain] (f\i) [label=left: \i] {};
\end{scope}

\begin{scope}[yshift=-3cm, start chain=going below,node distance=1mm]
\foreach \i in {6}
  \node[tnode,on chain] (f\i) [label=left: \i] {};
\end{scope}

\begin{scope}[xshift=2.2cm,yshift=-0cm,start chain=going below,node distance=1mm]
\foreach \i in {1, 2,...,8}
  \node[ssnode,on chain] (s\i) [label=right: \i] {};
\end{scope}

\begin{scope}[xshift=2.2cm,yshift=-3.8cm,start chain=going below,node distance=1mm]
\foreach \i in {9}
  \node[enode,on chain] (s\i) {$\vdots$};
\end{scope}

\node [white,fit=(f1) (f3), label={[label distance=-0cm]90:Units} ]  {};

\node [white,fit=(s1) (s8), label={[label distance=-1.1cm]90:Assignments} ] {};

\draw[ultra thick,gray] (f2) -- (s1);
\draw[ultra thick,gray] (f2) -- (s2);
\draw[ultra thick,gray] (f2) -- (s3);
\draw[ultra thick,gray] (f2) -- (s4);
\draw[ultra thick,gray] (f2) -- (s5);
\draw[ultra thick,gray] (f2) -- (s6);
\draw[ultra thick,gray] (f2) -- (s7);
\draw[ultra thick,gray] (f2) -- (s8);
\draw[ultra thick,gray] (f4) -- (s1);
\draw[ultra thick,gray] (f4) -- (s2);
\draw[ultra thick,gray] (f4) -- (s3);
\draw[ultra thick,gray] (f4) -- (s4);
\draw[ultra thick,gray] (f4) -- (s5);
\draw[ultra thick,gray] (f4) -- (s6);
\draw[ultra thick,gray] (f4) -- (s7);
\draw[ultra thick,gray] (f4) -- (s8);
\end{tikzpicture} \hspace{-8mm}
\begin{tikzpicture}[thick,
  every node/.style={draw,circle},
  fsnode/.style={fill=darkgray},
  tnode/.style={fill=mygreen},
  ssnode/.style={fill=darkgray},
  enode/.style={fill=none,draw=none,minimum height=5cm,inner sep=1cm},
  every fit/.style={ellipse,draw,inner sep=-1pt,text width=1.8cm}, ,shorten >= 2pt,shorten <= 2pt
]

\begin{scope}[start chain=going below,node distance=1mm]
\foreach \i in {1}
  \node[fsnode,on chain] (f\i) [label=left: \i] {};
\end{scope}

\begin{scope}[yshift=-0.5cm, start chain=going below,node distance=1mm]
\foreach \i in {2}
  \node[fsnode,on chain] (f\i) [label=left: *\i] {};
\end{scope}

\begin{scope}[yshift=-1cm, start chain=going below,node distance=1mm]
\foreach \i in {3}
  \node[fsnode,on chain] (f\i) [label=left: \i] {};
\end{scope}

\begin{scope}[yshift=-2cm, start chain=going below,node distance=1mm]
\foreach \i in {4,5}
  \node[fsnode,on chain] (f\i) [label=left: \i] {};
\end{scope}

\begin{scope}[yshift=-3cm, start chain=going below,node distance=1mm]
\foreach \i in {6}
  \node[tnode,on chain] (f\i) [label=left: *\i] {};
\end{scope}

\begin{scope}[xshift=2.2cm,yshift=-0cm,start chain=going below,node distance=1mm]
\foreach \i in {1, 2,...,8}
  \node[ssnode,on chain] (s\i) [label=right: \i] {};
\end{scope}

\begin{scope}[xshift=2.2cm,yshift=-3.8cm,start chain=going below,node distance=1mm]
\foreach \i in {9}
  \node[enode,on chain] (s\i) {$\vdots$};
\end{scope}

\node [white,fit=(f1) (f3), label={[label distance=-0cm]90:Units} ]  {};

\node [white,fit=(s1) (s8), label={[label distance=-1.1cm]90:Assignments} ] {};

\draw[ultra thick,gray] (f2) -- (s1);
\draw[ultra thick,gray] (f2) -- (s2);
\draw[ultra thick,gray] (f2) -- (s3);
\draw[ultra thick,gray] (f2) -- (s4);
\draw[ultra thick,gray] (f2) -- (s5);
\draw[ultra thick,gray] (f2) -- (s6);
\draw[ultra thick,gray] (f2) -- (s7);
\draw[ultra thick,gray] (f2) -- (s8);
\end{tikzpicture}

\vspace{-20mm}
\caption{Example bicliques used by the three methods of Section~\ref{sec:clust_study}. Interference is clustered with six units divided equally into two clusters.
$\Zobs$ is~Assignment 1 and unit 6 is treated~(colored green). 
The units of each clique~(focal units) are marked by an asterisk, ``*''. An edge denotes that the 
unit is exposed to either $f_i=1$ or $f_i=0$, the 
exposures in $H_0$ of Equation~\eqref{eq:H0cluster}.\\
{\textit{Left:}} conditioning event of the biclique test: two focals per cluster; 
{\textit{Middle:}} conditioning event of~``conditional focals''~\citep{basse2019}: one untreated focal per cluster is randomly chosen~(units 2 and 4);~\textit{Right:}  conditioning event of~``random focals''~\citep{athey2018exact}:
only one focal per cluster is randomly chosen~(units 2 and 6). 
However, unit 6 is effectively removed since unit 6 is treated, 
and there are no edges between this unit and 
any assignment.
}
\label{fig2_comparison}
\end{figure}
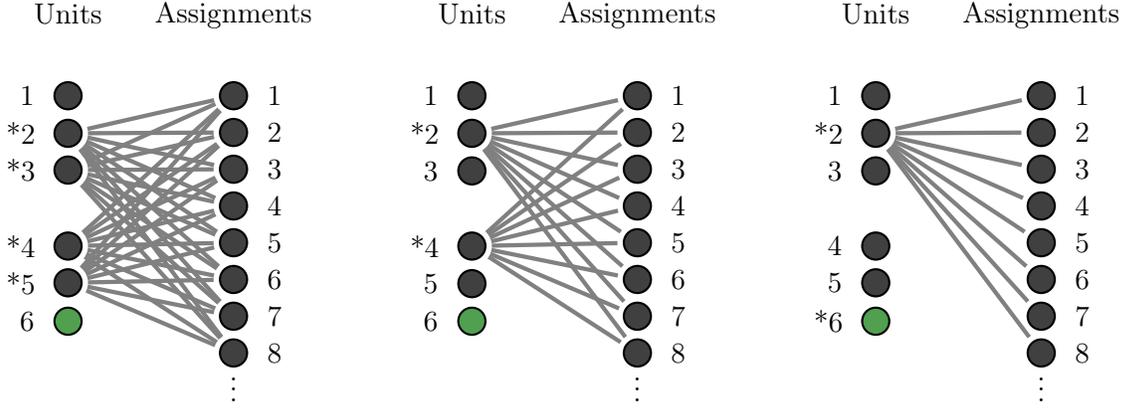

Figure~\ref{fig2_comparison} also depicts the conditioning events for every method, with the selected focal units in each biclique marked with an asterisk.
The leftmost subfigure shows that the biclique test in Procedure~\ref{proc:test} conditions on the biclique 
that includes $\{2, 3, 4, 5\}$ as focal units. Unit 6 is not included 
since it does not add any edges in the objective function of Equation~\eqref{eq:max_clique}.
This biclique is balanced and dense in the sense that there are two units in each exposure of interest, and it contains many edges.
%
The middle subfigure shows the ``conditional focals'' method of~\citet{basse2019}.
For this test, we see that, by construction, only one focal unit is selected per cluster.  This leads to a biclique that is less dense than the biclique from our proposed test.
The rightmost subfigure shows the ``random  focals'' method of~\citet{athey2018exact}. 
This differs from ``conditional focals'' because it may select treated units as focal units.
For illustration, if the method selects unit 6, which is assigned to treatment, this unit is effectively dropped from the conditioning biclique, 
thus reducing power. 


\subsection{Simulation study: Two-stage experiment}\label{sec:clustered_sim}

We follow \citet{basse2018analyzing} to generate data for the 
setting with clustered interference:
\begin{align}
		y_{i, 0} &\sim N(\mu_{0},\sigma_\mu^2),\nonumber\\
		\tau_i^P & \sim N(\tau^P, \sigma_\tau^2),\nonumber\\
		 \tau_i^S & \sim N(\tau^S, \sigma_\tau^2),\nonumber\\
		y_{i, 2} &=  y_{i, 0} + \tau_i^P,\nonumber\\
		y_{i, 1} &=  y_{i, 0} + \tau_i^S.\nonumber
\end{align}

In the experiment, $\lfloor K/2 \rfloor$ clusters are assigned to treatment, and within treated clusters, one unit is randomly assigned to treatment (as in Example \ref{example:clustered}). We sample $Y_i(\text{``control''}) \sim N(y_{i, 0},\sigma_y^2)$, $Y_i(\text{``spillover''}) \sim N(y_{i, 1},\sigma_y^2)$, 
 and  $Y_i(\text{``treated''}) \sim N(y_{i, 2},\sigma_y^2)$, 
in order to generate the individual potential outcomes.
As such, $\tau_i^P$  and $\tau_i^S$ correspond to idiosyncratic primary and spillover effects, respectively.
We consider the following specifications:  $N=300$, $\mu_{0}=2$, $\sigma_{\mu}=\sigma_{\tau}=0.1$, $\sigma_y=0.5$, $\tau^S=0.7$, $\tau^P=1.5$, and $K \in \{20, 30, 75\}$.  In the simulation we vary $K$ to see how 
different methods perform with small and large-sized clusters. The different cluster sizes, $N/K$, are therefore contained in $\{15, 10, 4\}$. %
%
%
For the simulations, we generate 5,000 different assignments and construct the null exposure graph for the biclique method.  
For each cluster size and fixed $\tau$, we generate 2,000 data sets from the DGP given above;  $\tau$ is varied among 300 equally spaced values from 0 to 1.%


The power plots are shown in Figure~\ref{powerpanel}. 
We see that the biclique method performs substantially better with larger cluster sizes. 
Recalling the conclusions from Figure \ref{fig2_comparison}, this is explained by noting that the biclique method can select multiple focal units per cluster in the biclique decomposition.  In contrast, the methods of~\citet{basse2019} and \citet{athey2018exact}, as defined here, both use one focal unit per cluster.
 %
To confirm, we calculated the average number of focal units per cluster for each method and data configuration with 15 units per cluster ($N=300,K=20$).  The biclique method had 5.24 focal units per cluster compared to 1 for method (ii) and 0.97 for method (iii). 
In contrast, the biclique method underperforms in the right-most plot of Figure~\ref{powerpanel}.  This is because 
smaller cluster sizes imply a sparser null exposure graph. 
In such settings, the biclique method has trouble finding good biclique decompositions, and can drop entire clusters in its conditioning. By contrast,
alternative methods are fixed to sample exactly one focal from every cluster.
%


\begin{figure}[t!]
\centerline{
\includegraphics[scale=0.41]{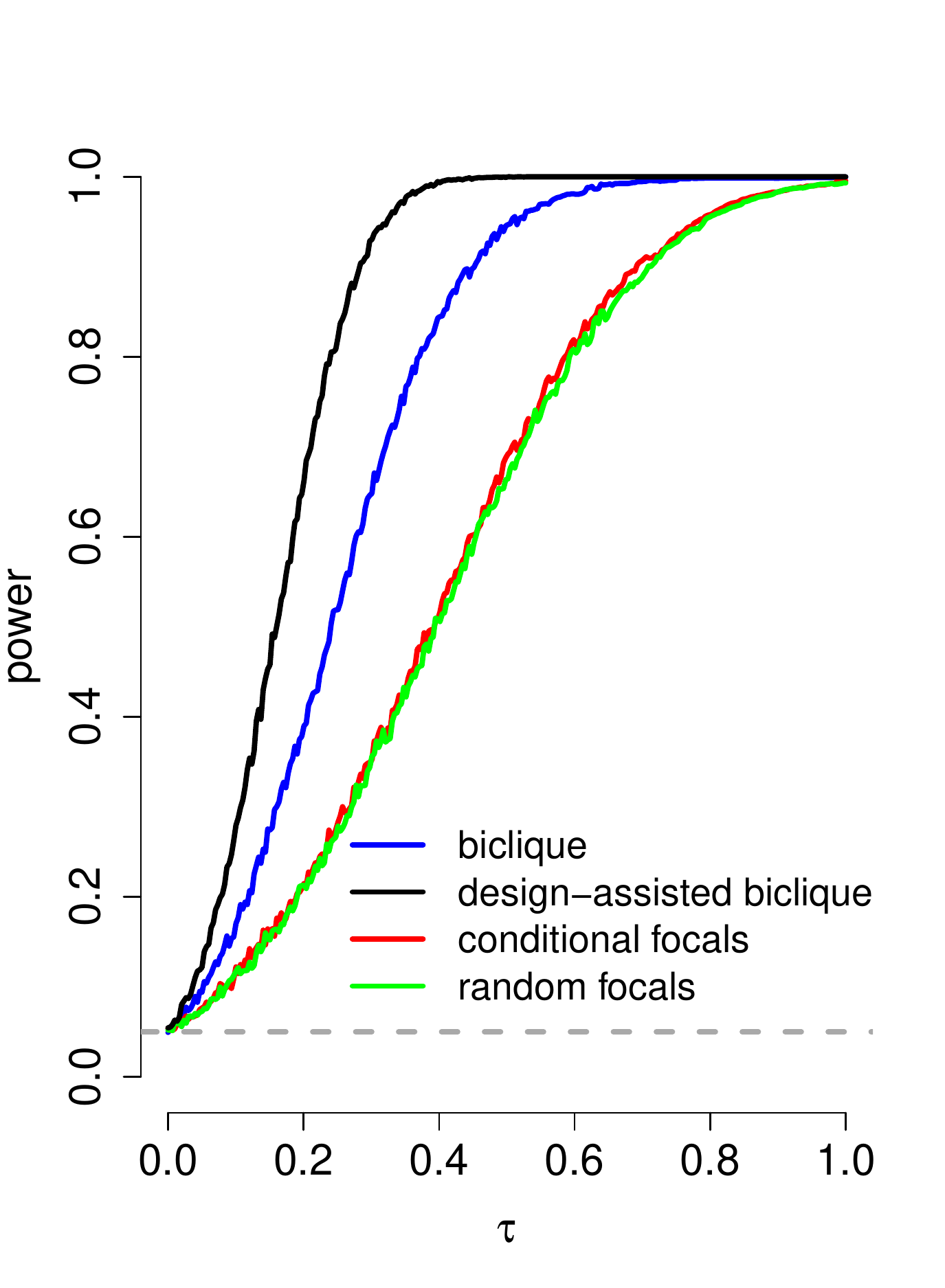}
\includegraphics[scale=0.41]{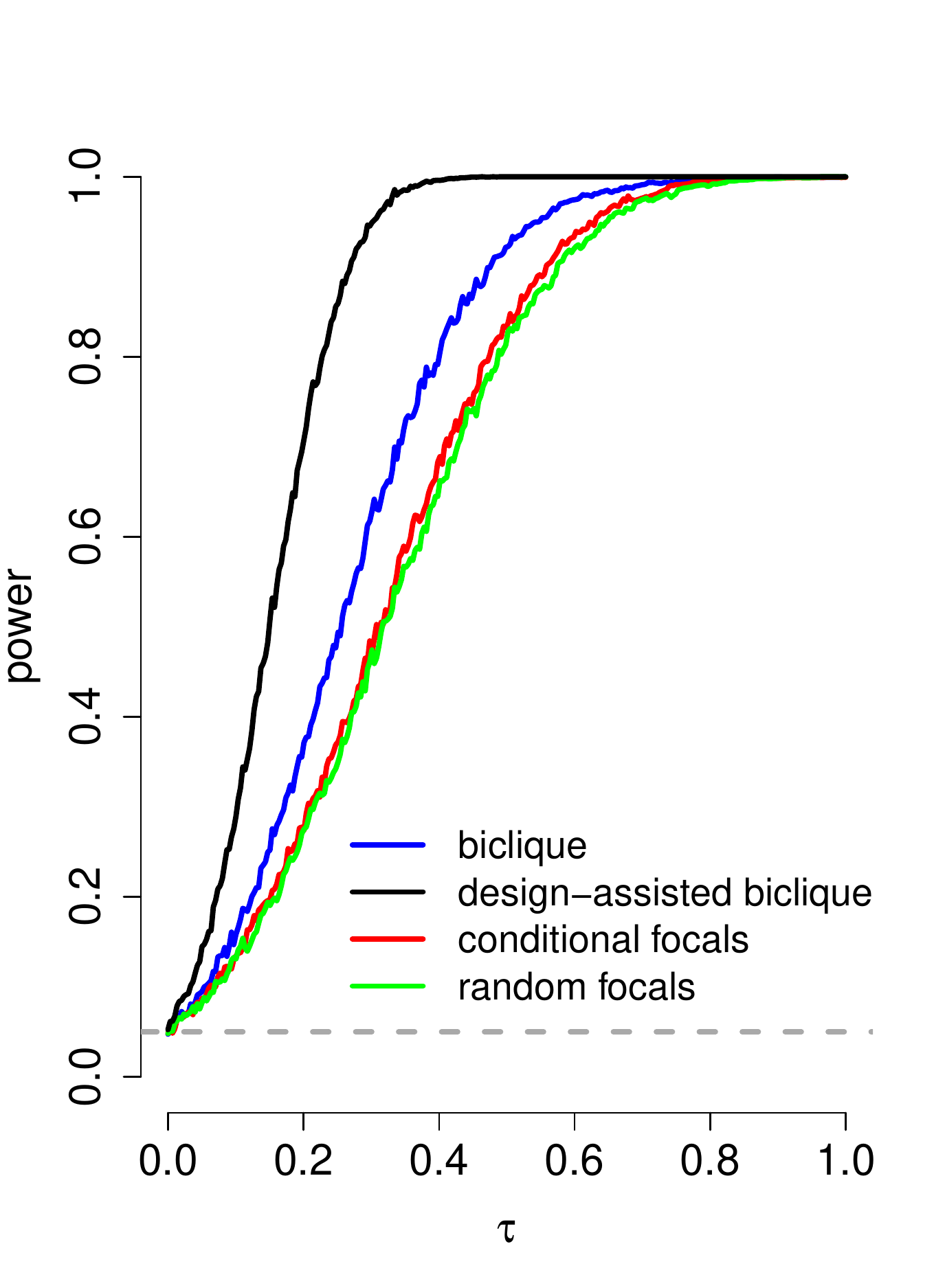}
\includegraphics[scale=0.41]{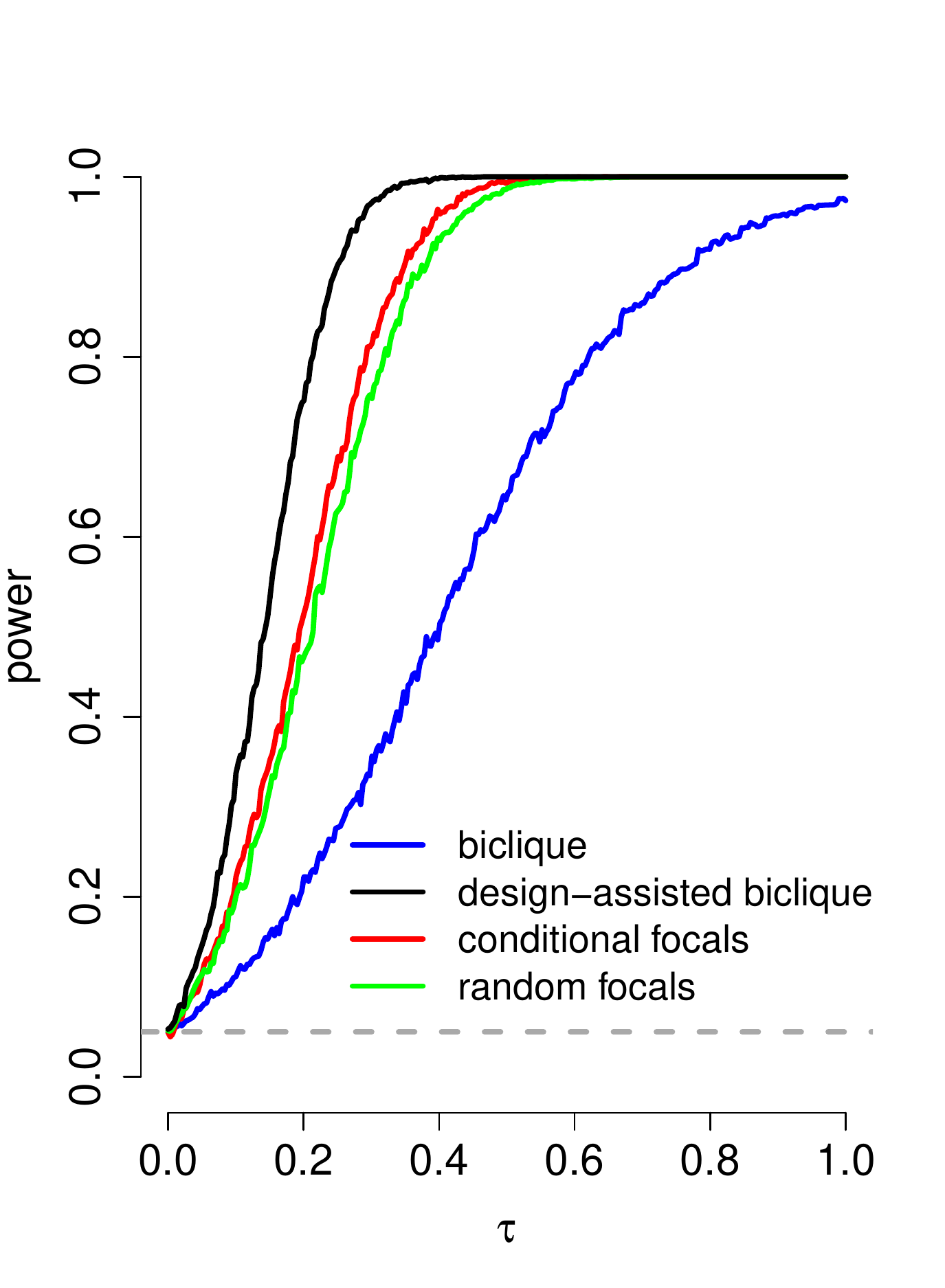}
}
\caption{Power plots for three data configurations; the design-assisted biclique procedure is described in Section \ref{sec:design_cliques}.
The results include the results shown in Figure 3;
{\em left:} $N=300,K=20$; {\em middle:} $N=300,K=30$;
and 
{\em right:} $N=300,K=75$.}
\label{powerpanel}
\end{figure}

To further investigate the power of the biclique test, 
we gather data on the bicliques generated by the biclique decomposition.
We alter the optimization constraints in 
the biclique decomposition algorithm we employ~(see Section~\ref{sec:decomposition} for implementation details), and fix $N=300,K=20$ and $\tau=0.3$.
Figure~\ref{power_cliq1} displays the results.  
We see that the average number of units and assignments in the decomposition are negatively correlated, which 
suggests a power trade-off.
As the number of (focal) units per biclique increases, power also increases until some threshold. 
Beyond this threshold, more units in the biclique come at the expense of fewer treatment assignments, 
which on aggregate decreases power.
This highlights that the automated biclique test does not incorporate specific information about the clustered interference structure. In Section~\ref{sec:design_cliques}, we show that we can modify the standard procedure to improve the test's performance in this case.

Not evident in Figures~\ref{powerpanel}  and~\ref{power_cliq1} is the relative ease of implementation for each method.
The biclique method is straightforward and essentially automated for a well-defined biclique decomposition algorithm.
On the other hand, methods (ii) and (iii) are generally hard to implement.
For example,~\citet{basse2019} show that their test can be implemented as a permutation test
  only under the conditioning mechanism described in (ii). Selecting more units per cluster to increase power will generally 
  break this property. Furthermore, the test described in (iii) is a permutation test 
  only when all clusters have equal size. It is not clear how to implement a valid permutation test based on the method of~\citet{athey2018exact} with unequal cluster sizes.
  
Overall, these results suggest that our proposed biclique test 
can deliver reasonably powered randomization tests that are comparable 
to methods specifically designed for the interference structure at hand.
In settings where the null exposure graph is dense, the biclique 
test can even outperform such special-purpose methods, since it is able to  discover a more flexible conditioning.
  Most importantly, the biclique test operates automatically under any interference structure, whereas other methods require user input and specification.

\begin{figure}[t!]
\centerline{\includegraphics[scale=0.6]{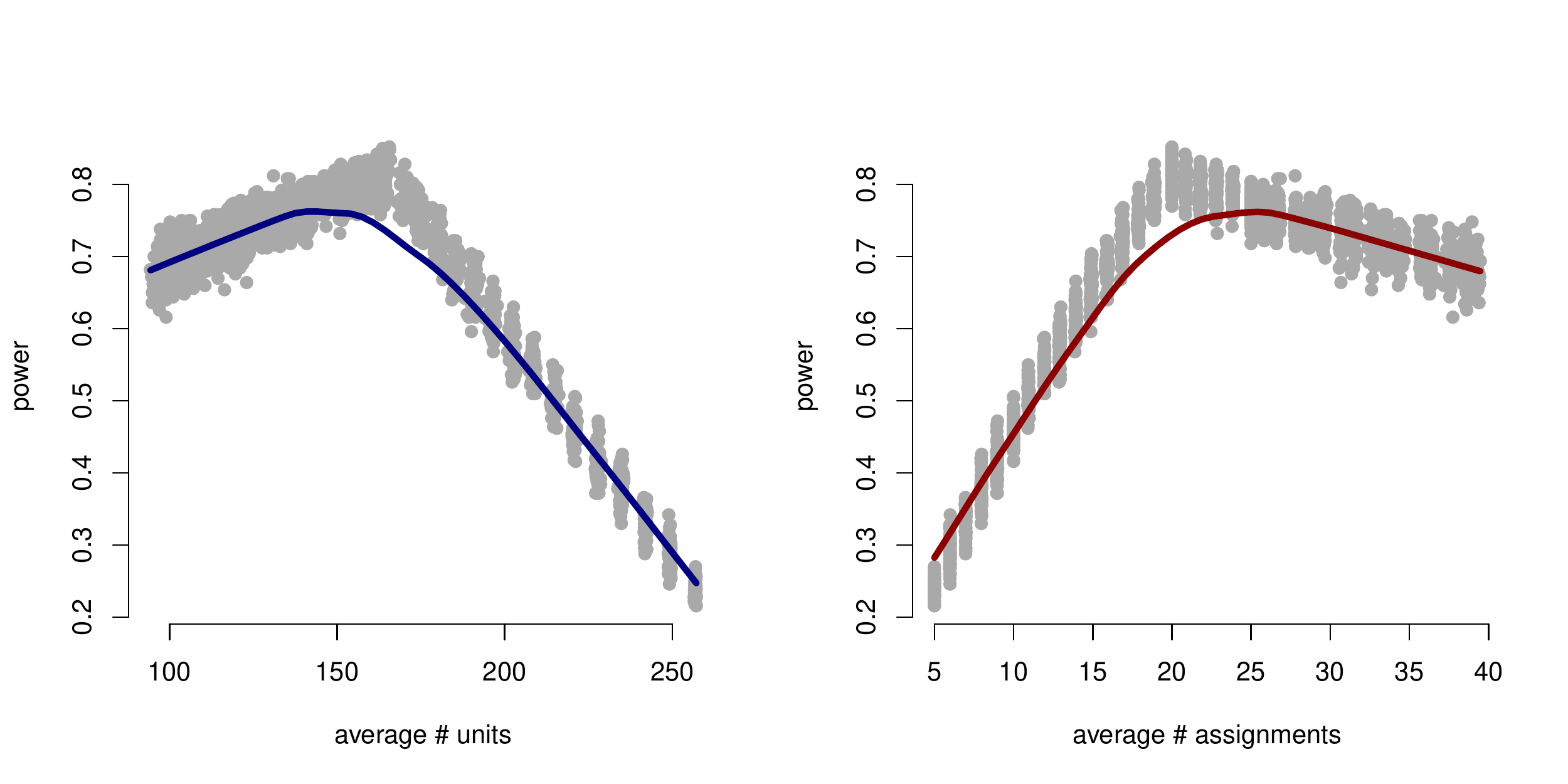}}
\caption{Power values for varying biclique characteristics. The data structure is fixed at $\{N=300,K=20\}$ and $\tau=0.3$.  Each individual gray dot corresponds to a biclique decomposition of the null exposure graph which we condition on for the biclique method. The left graph shows power as a function of the average number of units (focals), and the right graph shows power as a function of the average number of assignments. A figure combining this information into a single plot is shown in Appendix~\ref{appendix:clustered}.}
\label{power_cliq1}
\end{figure}

\subsection{Improving power: Design-assisted biclique tests}\label{sec:design_cliques}

We now discuss how to leverage knowledge of the experimental design to improve the power of the biclique test.
We describe the approach in detail and prove its validity in Appendix \ref{sec:design_test_appendix}.  In the following, we give a brief overview and illustrate the performance gains. 

The main idea relies on understanding the structure of the null exposure graph in the clustered interference setting. 
Recall that, within each cluster, every unit $i$ is connected to every cluster assignment, except for the assignment where $i$ is treated.
Therefore, the subgraph of the null exposure graph corresponding to cluster $k$ (i.e., containing only the units in cluster $k$ and all possible ``cluster subvectors'' of the population assignment vector) 
can be decomposed into two large bicliques.  These two bicliques are constructed to each contain half of cluster $k$'s units (see Appendix \ref{sec:design_test_appendix} for details).  This is useful because the biclique $C$ that the test conditions on 
can then be constructed by joining together the large bicliques from each individual cluster. As we will see, this modified approach largely addresses the power issues for large $K$ in the simulation of the previous section. 

Figure~\ref{powerpanel} 
shows the simulations with the ``design-assisted'' biclique test described above.
We see that the design-assisted test dominates all other tests. 
The main benefit of the design-assisted test is that it uses half of the units in the cluster as focals,  
whereas the tests of \citet{basse2019} and \citet{athey2018exact} only use one unit.
This also explains why the performance gap between these tests narrows for larger $K$.
For instance, when $K=75$, the design-assisted biclique test uses roughly two units per cluster, which still yields power improvement compared to using just one.


The caveat of this analysis is that it may not always be easy to leverage the design in constructing a powerful biclique test because it requires understanding the relation
between the design and the structure of the null exposure graph. 
However, the analysis and empirical results of this section suggest that when such combination is possible, 
the design-assisted biclique test dominates alternative approaches.


\newcommand{\pc}[1]{``\text{pure control}"}
\newcommand{\sr}[1]{``\text{spillover}_{#1}"}
\section{Application to spatial interference: Crime in Medell\'in}\label{sec:spatial}
In this section, we illustrate our method in a spatial interference setting in which interactions occur between ``neighboring'' units, but without the simpler structure of clustered interference.
%
%
We focus on re-analyzing a large-scale experiment
in Medell{\'i}n, Colombia studying the impact of ``hotspot policing'' on crime~\citep{collazos2019hot}.

\subsection{Problem setup and comparison of available methods}
Following \citet{collazos2019hot}, the units are $N=37,055$ street segments, 967 of which were identified as \emph{hotspots} using geo-located police data and further consultation with police. Of these hotspots, 384 were randomly assigned to treatment, a six-month increase in daily police presence, via a completely randomized design over a domain $\Zdom$ of roughly $10,000$ possible assignments.  
%
The outcome of interest is a \emph{crime index}, a weighted sum of the crime counts on each street segment.\footnote{As discussed in \citet{collazos2019hot}, the index weights are chosen based on the length of sentence for a given crime.  They are:  0.550 for homicides, 0.112 for assaults, 0.221 for car and motorbike theft, and 0.116 for personal robbery. Crime data is matched to street segment within 40-meter buffers.  In other words, if a crime happened in an alley, it will be matched to the closest street segment within a 40-meter radius. If there is overlap, it is matched to the street segment closest in terms of Euclidean distance.}
%

To define the spillover hypothesis, we need to define the exposure function, $f$. 
Following~\citet{collazos2019hot}, we use geographic distance as a measure for spillover exposure. Let $d(i, j)$ be the distance (in meters) between unit $i$ and unit $j$. 
We then define:
\begin{align}\label{eq:f}
f_i(z) = 
\begin{cases}
``\text{pure control}",~\text{if}~z_i=0~\text{and}~\sum_{j\neq i} 
\Ind{d(i, j) \le 500} z_j = 0;\\
\sr{r},~~~~~\text{if}~z_i=0~\text{and}~\sum_{j\neq i} 
\Ind{d(i, j) \le r} z_j > 0; \\
``\text{other}".
\end{cases}
\end{align}
This defines a ``pure control'' as a control unit 
that has no treated units closer than $500$ meters, ensuring significant spatial separation from treated streets. A control unit is assigned to ``spillover'' when there are treated units closer than the specified distance.

Our goal is to test contrast hypotheses of the form $H_0^{\aa, \bb_r}$, with 
$\aa = ``\text{pure control}"$ and $\bb_r = \sr{r}$, where $r$ is a free variable 
in the set:
$$
r = \{75, 100, 125, 150, 175, 225, 275, 325, 375, 425\}.
$$
Each hypothesis for a given $r$ will have its own null exposure graph and biclique decomposition.
Following~\citet{collazos2019hot}, we refer to ``short-range spillovers'' as the set 
of hypotheses $H_0^{\aa, \bb_r}$ with $r \le 125$m.

To illustrate, Figure~\ref{map2} shows the induced exposures under the  randomization realized in the experiment (left), and an alternative randomization that was not realized (right), for $r=125$m. The figure highlights short-range spillover units (among all available units) 
in light blue and light green for the observed treatment and randomization, respectively.
Figure~\ref{map2} also depicts important geographic features of Medell\'in, with many hotspots near the city center, and where dark areas or holes correspond to physical barriers (e.g., mountains) or major infrastructure (e.g., airports).
Even though only 967 street segments can receive the active treatment, every street segment can potentially 
receive spillovers, and so we use the entire street network in our analysis.
A key challenge is that street segments near the city center have a much higher probability of being exposed to crime spillovers than segments on the outskirts of the city.

\begin{figure}[!t]
\centerline{\includegraphics[scale=0.5]{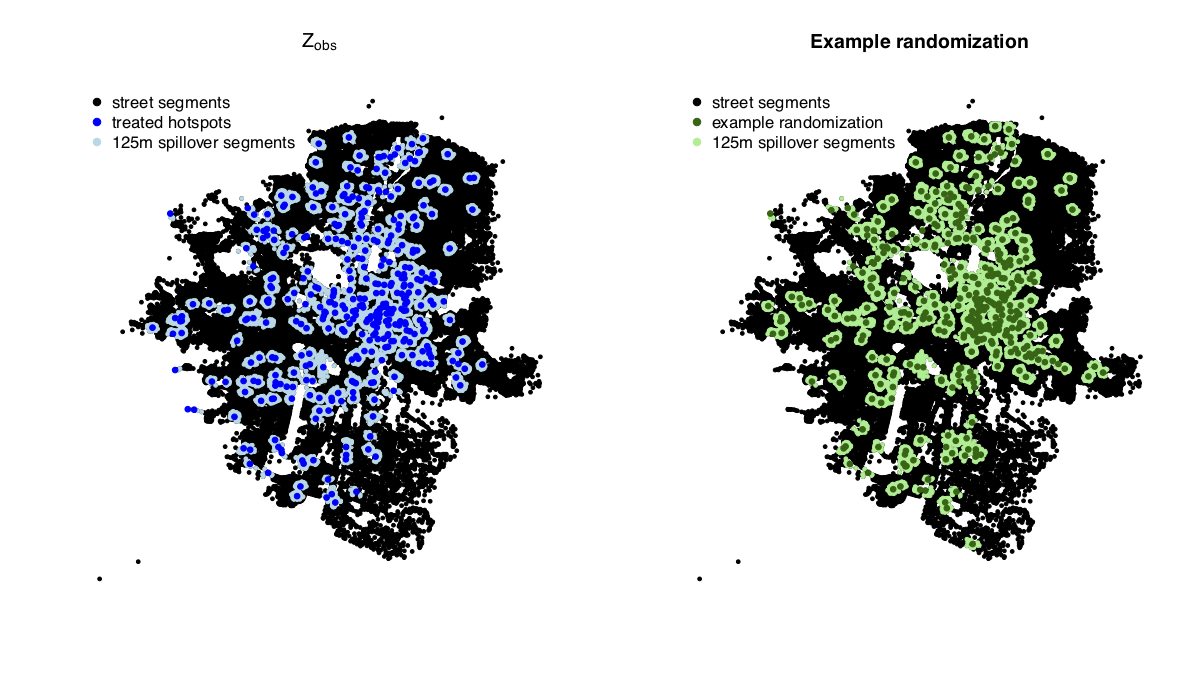}}
\caption{Street segments, hotspots, and treated hotspots for the data set.  The left figure is the observed assignment for the experiment.  The right figure is an example randomization of the assignment vector.  The dots cover different hotspots, but they are still within the 967 segments representing the hotspots.  Additionally, the light colored dots represent the 125m spillover units, i.e.: street segments that are within 125m of the treated hotspots for a given randomization.}
\label{map2}
\end{figure}

\textcolor{black}{Finally, unlike our analysis of clustered interference, we will only apply the biclique test in this case study. In principle, we could adapt the test of~\citet{athey2018exact} to this setting, but the implementation of their procedure would be underpowered due to the spatial structure.
Specifically, we will see that (under the short-range spillover hypothesis) focal units are concentrated either at the center or outskirts of the city; see Section~\ref{sec:FRT_med} and Figure~\ref{nulls2}. It is unlikely that this ``center-outskirts'' pattern could be generated through a random selection of focals.  Furthermore, the $\epsilon$-net method of focal selection of~\citet[Section 5.4.2]{athey2018exact} would not work well because it would generate patterns of focal units that are spatially uniform.
Similarly, it is not clear how to apply the approach of~\citet{basse2019}, which is more narrowly tailored to the clustered interference setting.}

\subsection{Spillover hypotheses: A simulation study}\label{sec:spillover_sim}
We now assess our proposed method via a simulation study calibrated to the actual Medell\'in street network. For these simulations, each biclique decomposition is constrained to include bicliques with at least 100 focal units and 1000 assignments.  
We explore the effect of radius $r$ on test power using the following simple model for the outcomes, $Y$:
\begin{align}\label{eq:Y}
		Y_i(``\text{pure control}") & \sim \text{Gamma}(\alpha, \beta),\nonumber\\
		Y_i(``\text{spillover}_r") & = Y_i(``\text{pure control}") + \tau_r.
\end{align}
The shape and rate parameters~($\alpha, \beta$, respectively) are selected to match the mean and variance of the observed outcome in the actual experiment, the crime index.
The parameter $\tau_r$ determines an additive spillover effect at radius $r$. 
We set $\tau_r \propto 1/r^2$ to allow for heterogeneity of the spillover effect with respect to radius.\footnote{This is in line with existing literature~\citep{thomas2013quantifying,barr1990crime,verbitsky2012causal},
which has pointed out that ``the likelihood that an offender will target an opportunity will be inversely related to the distance it is located from their routine activity spaces''~\citep{eck1993threat, johnson2014crime}.}
In the simulation, we sample assignments according to the true design for every value of $r$ and outcomes according to Equation~\eqref{eq:Y}. We then 
test   the null hypothesis $H_0^{\aa, \bb_r}$ on spillovers at distance $r$, defined at the beginning of this section. 
The test statistic is the simple difference in means between focal units exposed to $\aa$ and $\bb_r$.

The results are shown in Figure~\ref{powervradius}.
In the left subfigure, we fix the additive treatment effect at zero~($\tau_r=0$) to assess validity.  Our rejection level is 0.05, and so we confirm the validity of the biclique method since all power values gather around the 5\% rejection rate.  
In the right subfigure, we consider nonzero additive spillovers effects~($\tau_r > 0$) that are calibrated based on the spillover radius. 
We observe that the power curve is generally concave and nonmonotonic since it increases until some radius and then decreases.
At first, this may be counterintuitive since as $r$ increases the spillover effect, $\tau_r$, decreases (by definition), which should make it harder to detect.
However, as shown in Appendix~\ref{appendix:clustered}, the number of focal units 
also increases sharply with respect to $r$. 
The net effect is an increase of power.
At the same time,  as we discussed in Section~\ref{sec:clustered_sim},
an increase in the number of focals per biclique generally results in a decrease in the number of 
assignments per clique~(see also Figure~\ref{app:power_cliq2} in Appendix~\ref{appendix:clustered}), and, eventually, in decreased testing power.
Under our assumed outcome model, the maximum power is achieved approximately at a 275m spillover radius. 

This kind of analysis gives a useful estimate for the power profile~(Figure~\ref{powervradius}, right) of our proposed biclique test for a given biclique decomposition algorithm. In practice, we could compare between the 
power profiles of different biclique decomposition algorithms, and choose the most favorable algorithm to apply on the real data. See Section~\ref{sec:power}
for additional discussion on testing power.

\begin{figure}[t!]
\centerline{\includegraphics[scale=0.6]{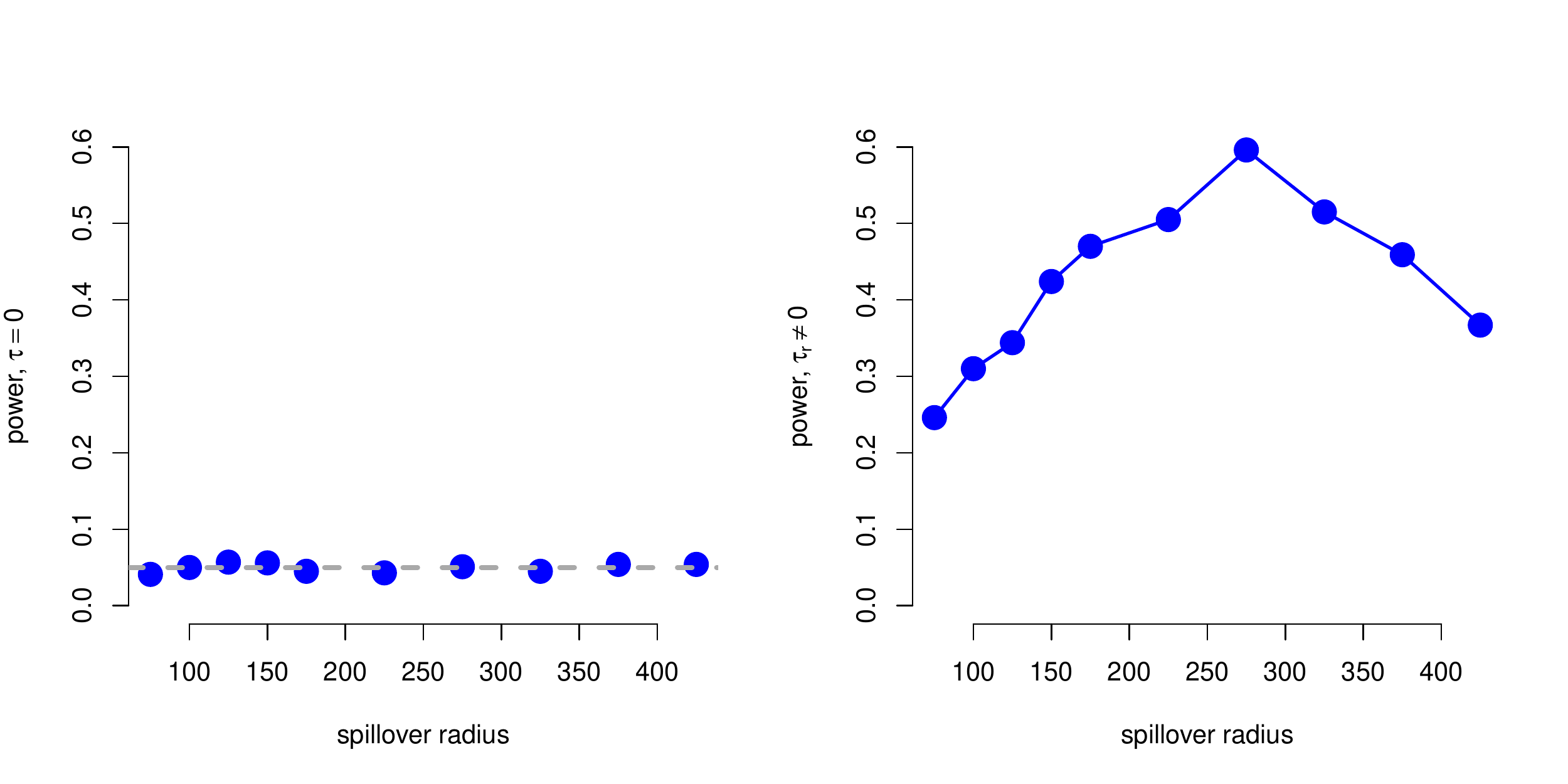}}
\caption{The left figure shows the power analysis across 1000 simulations when $\tau_r=0$.  We reject the null at the 0.05 level.  Therefore, the left figure confirms the validity of the biclique method.  The right figure displays the power analysis for nonzero $\tau_r \propto 1/r^2$.}
\label{powervradius}
\end{figure}

\subsection{Spillover hypotheses: biclique test on real data}\label{sec:FRT_med}

In this section, we demonstrate our proposed biclique test using the actual outcome data. We focus on the spillover hypothesis $H_0^{\{\aa, \bb_r\}}$ with radius $r=125$m, following \citet{collazos2019hot} who define this 
type of exposure as ``short-range spillover''. Results for all radius values 
in the previous simulation are included in Appendix~\ref{appendix:spatial}. We therefore test whether 
there is a difference between  outcomes of pure control units and units who 
receive short-range spillovers:
\begin{align}\label{eq:H0_med}
H_0: Y_i(z) = Y_i(z'),~\text{for all}~z,z'~\text{such that}~ f_i(z), f_i(z')\in
\{``\text{pure control}", ``\text{spillover}_{r}" \}, 
\end{align}
where $r=125$m, and the exposures are defined in Equation~\eqref{eq:f}.
As we discuss in Section \ref{sec:H0}, rejecting $H_0^{\{\aa, \bb_r\}}$ does not necessarily imply that the treatment exposures are different; instead it is possible that any
of the singleton hypotheses $H_0^{\{\aa\}}$ and $H_0^{\{\bb_r\}}$ does not hold.

The first step of the biclique test is to construct the null exposure graph.  This graph has 37,055 nodes on one side (number of streets/units), and 10,000 nodes on the other~(number of  assignments).
Figure~\ref{bicilique_visual2} visualizes this graph.
The left subfigure shows a block of the null exposure graph that is reasonably discernible, where a dot corresponds to a unit-assignment pair, say $(i, z)$. If the dot has white color, then there is no edge between 
unit $i$ and assignment $z$ in the graph. The light blue and navy blue colors indicate whether
$i$ receives short-range spillover or pure control exposure under $z$, respectively.
We note that the null hypothesis is not sharp for units exposed to ``white'', therefore we need to condition on a biclique where the white components are effectively removed. This is analogous to choosing a biclique where all units are either exposed to ``navy blue'' or ``light blue''.
The right side of Figure~\ref{bicilique_visual2} shows such a clique~(zoomed-in to have same size as the left side).\footnote{
Figure~\ref{bicilique_visual2} displays additional important information about the test.~For example, the colorings on the right side of Figure~\ref{bicilique_visual2} reveal that many units in the biclique are always ``pure control'' (navy blue columns), and a handful of units always receive ``short-range spillovers'' (light blue columns).  
Conceptually, more variation in exposures across units (i.e., more random colorings in 
right side of Figure~\ref{bicilique_visual2}) leads to more power.
Currently, our biclique decomposition algorithm cannot guarantee such variation; we leave this problem open for future work.
}

\begin{figure}[!t]
\centerline{\includegraphics[scale=0.64]{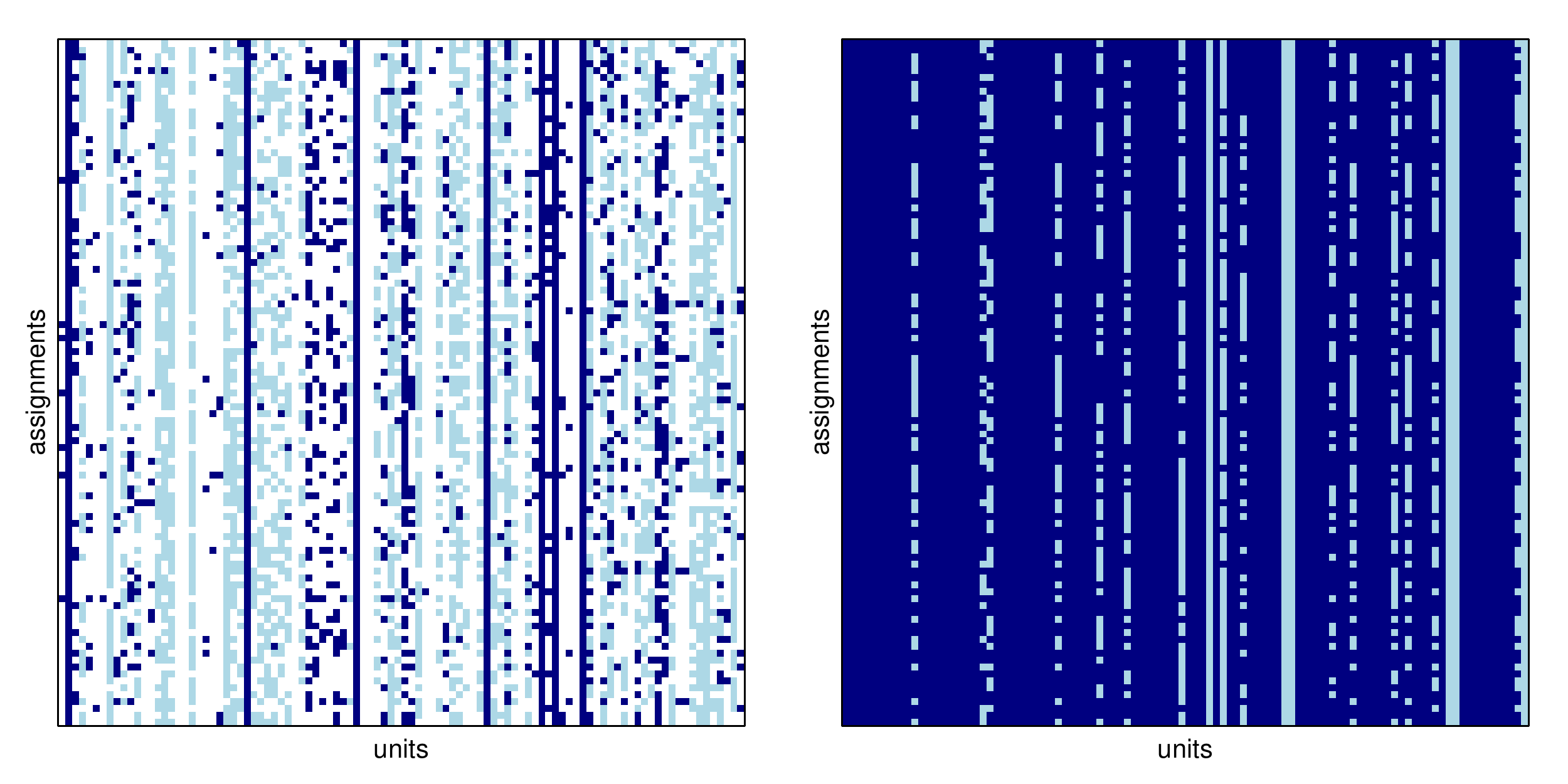}}
\caption{The left figure visually depicts the null exposure graph.  The vertical axis corresponds to the assignments from the randomization procedure, and the horizontal axis displays the units. Light blue denotes an untreated unit that is a spillover and close to a treated hotspot, and navy denotes pure control.  The right figure is a (zoomed-in) biclique of the null exposure graph containing the observed assignment.  There is no white in the biclique since the biclique only contains units exposed to spillover or pure control. To conserve space, we only display the first 100 assignments and units for both.}
\label{bicilique_visual2}
\end{figure}

We now apply our proposed test in Procedure~\ref{proc:test} to the spillover hypothesis of Equation~\eqref{eq:H0_med}. The left side of Figure~\ref{nulls2} displays hotspots, treated hotspots, and the focal units identified by the biclique decomposition.  As mentioned earlier, most focal units are at the center or outskirts of the city
due to the particular spatial structure of spillovers.
The right side of Figure~\ref{nulls2} displays the randomization distribution of the test statistic measuring the difference in means between crime index values on short-range spillover units and pure control units:
\begin{equation}\label{eq:Tc}
t(z, y; C) = \frac{1}{N_\bb} \sum_{i\in C} \Ind{f_i(z)=\bb} Y_i
-  \frac{1}{N_\aa} \sum_{i\in C} \Ind{f_i(z)=\aa} Y_i,
\end{equation}
where $C$ denotes the biclique we condition on; 
$i\in C$ denotes that unit $i$ is a node in the clique; 
$\aa=$ ``pure control'', $\bb=$ ``short-range spillover'';
$N_\aa = \sum_{i=1}^N \Ind{f_i(z)=\aa}$ and 
$N_\bb = \sum_{i=1}^N \Ind{f_i(z)=\bb}$ are the 
exposure counts.
Thus, positive values of the test statistic indicate that 
the average crime outcome for  units exposed to ``short-range spillovers'' is larger than the average outcome for units that are exposed to ``pure control''.

\begin{figure}[!t]
\centerline{
\includegraphics[scale=0.55]{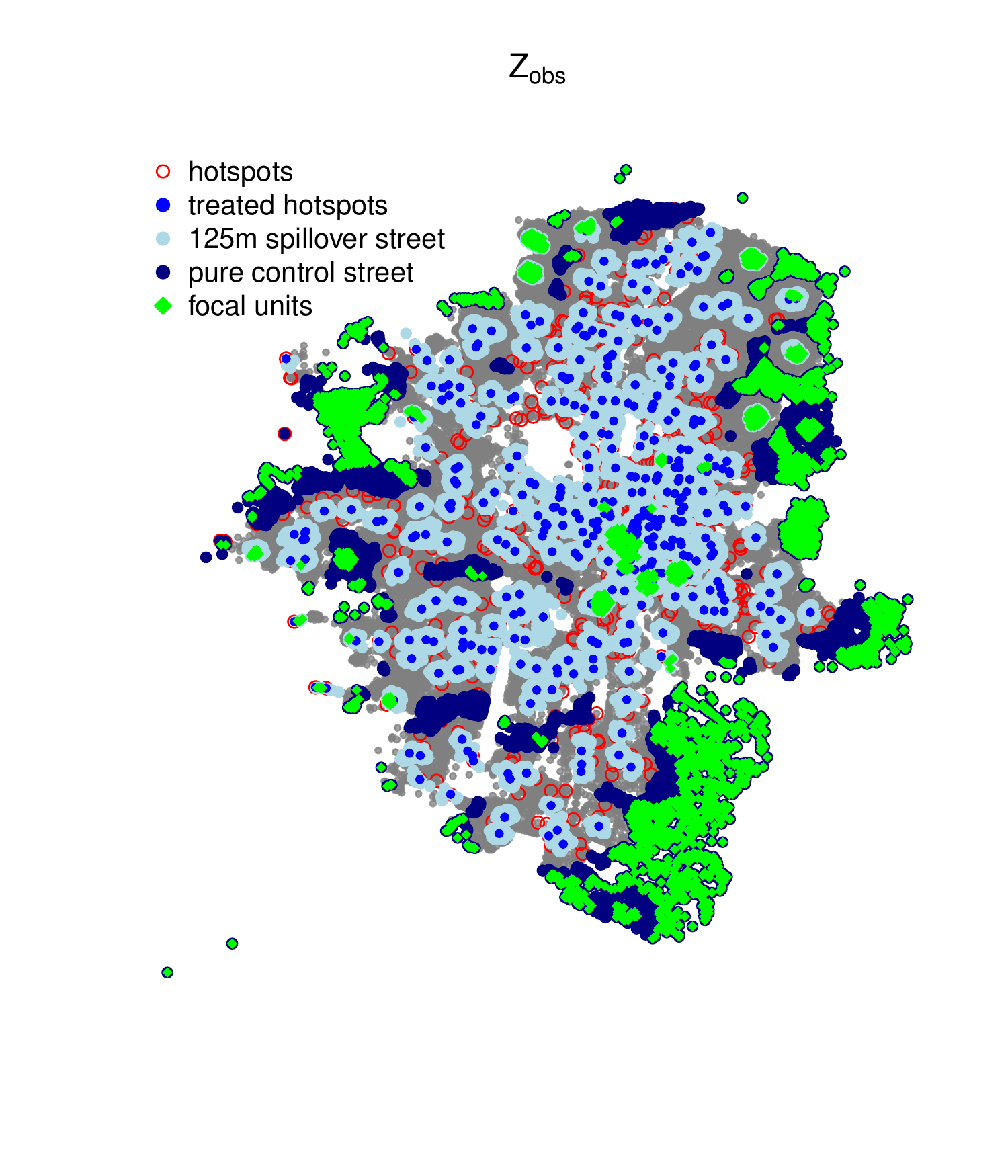} 
\includegraphics[scale=0.65]{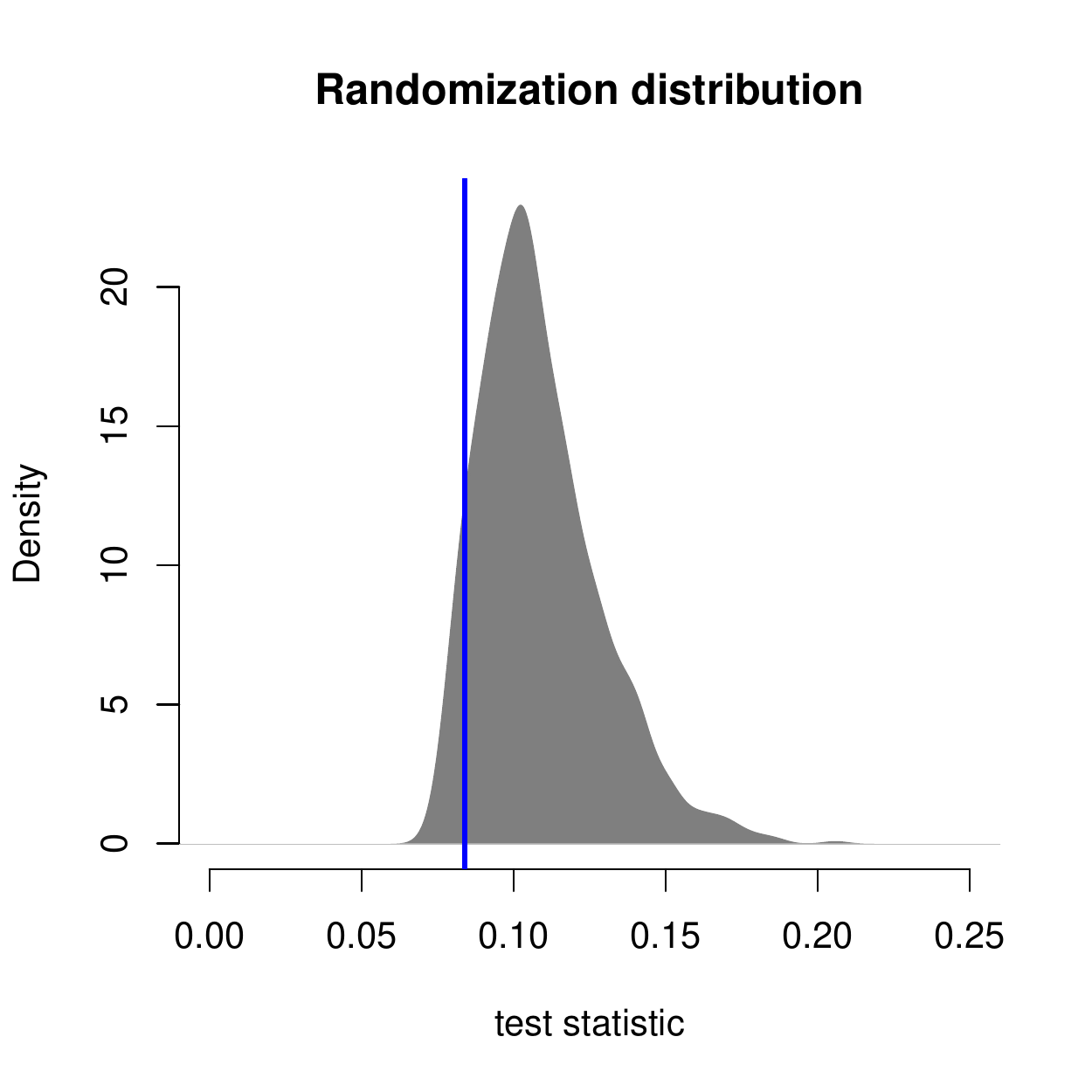}}
\caption{{\em Left:}~Representation of unit exposures for $\Zobs$.  Also shown in green are the focal units from the biclique represented on the right in Figure \ref{bicilique_visual2}. 
{\em Right:}~Test of the $r=125$m spillover radius hypothesis, where the test statistic is a difference in average outcomes between ``short-range spillover'' units and pure control units, defined in Equation~\eqref{eq:Tc}.  
Shown is the distribution of the test statistic under the null, and the blue line is the observed test statistic.  The p-value of the observed test statistic is approximately 0.077.}
\label{nulls2}
\end{figure}

\begin{figure}[h!]
\centering
	\includegraphics[scale=0.6]{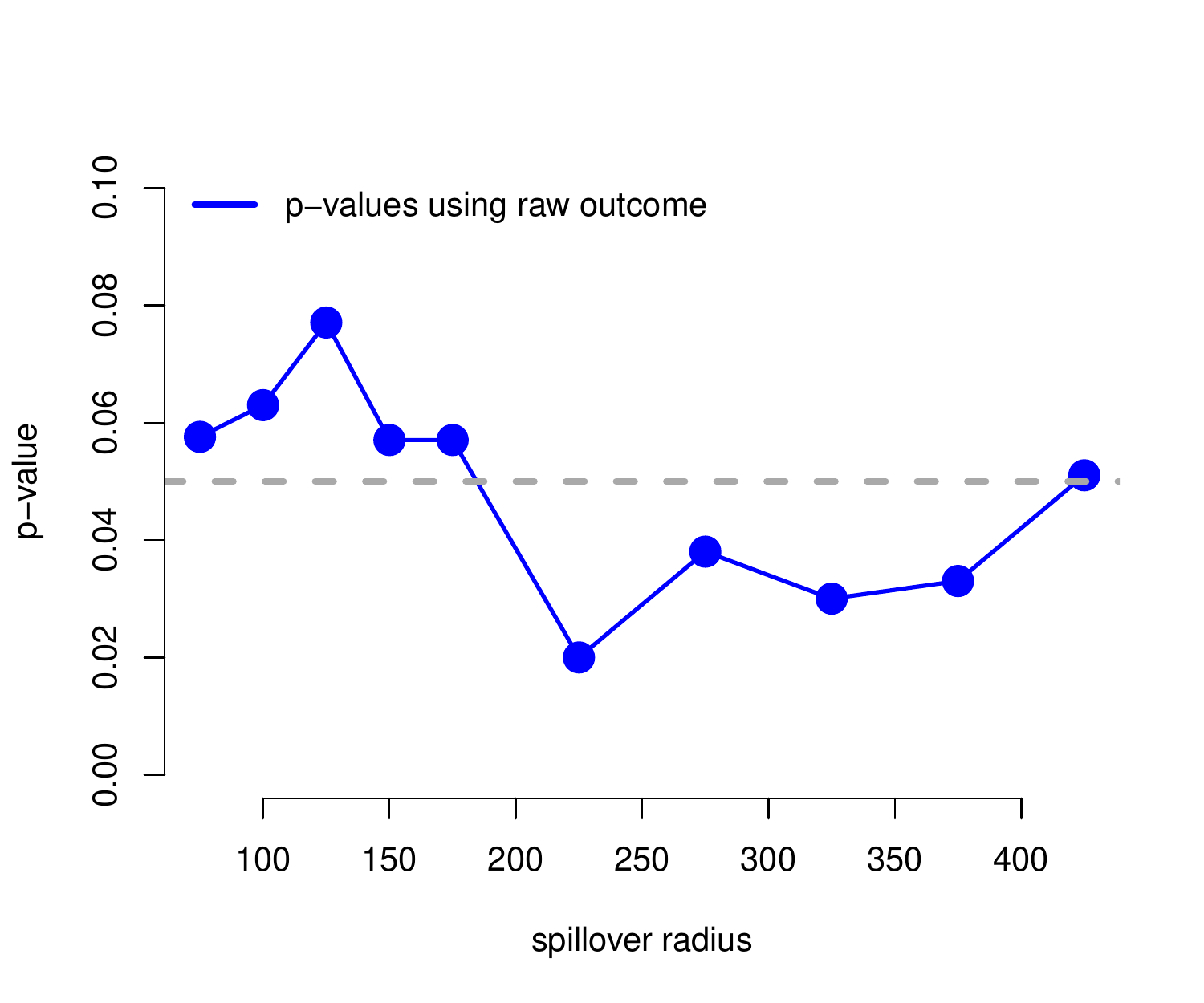}
	\caption{P-values (left vertical axis) for biclique tests with varying spillover radii (horizontal axis).  The blue line shows p-values for tests using the raw crime index.
	}
	\label{pval_fig1}
\end{figure}

%
The randomization values of the test statistic are computed conditional on the biclique depicted on the right of Figure~\ref{nulls2}.  
The observed test statistic is only larger than 7.7\% of all the randomization values, 
which is not significant at the 5\% level.
Assuming an additive short-range spillover effect, inversion of the randomization test
gives us $[-0.51, 0.03]$ as the 95\% confidence interval, indicating
that negative values for the spillover effect are more plausible under the additivity assumption. This observation suggests a decrease in crime of street segments surrounding an area with increased law enforcement/community policing, which is consistent with the literature~\citep{collazos2019hot, verbitsky2012causal}. 

Several caveats are in order, however. First, as we discuss above, it is possible that we reject the null hypothesis because the underlying exposures themselves are not correctly specified; that is, any of the singleton hypotheses $H_0^{\{\aa\}}$ and $H_0^{\{\bb_r\}}$ does not hold.
Another caveat is that, even though our randomization test is always valid, the power of the test may be affected by the inherent differences between 
the focal units. Specifically, due to the particular spatial arrangement in our application, the focal units that receive spillovers 
are mostly downtown streets, whereas units that are pure controls are mostly on 
the outskirts of the city~(see Figure~\ref{nulls2}, left). 
The observed differences in crime outcomes of these focal units
could depend, say, on differences in demographics between these two city areas---failing to account for such differences could reduce the power of the randomization test.
We discuss this issue in more detail, along with potential solutions using covariate adjustment, in Section~\ref{sec:discussion_het}.

Finally, we conduct biclique randomization tests, as described earlier, while varying the 
spillover distance, $r$.
The results are shown in Figure \ref{pval_fig1}, 
which also includes the results for the 125m-spillover 
presented in Figure~\ref{nulls2}.  For outcomes, we consider the raw crime index. 
We see that the p-values for the raw crime index are all small for varying radii; see the flat blue line in Figure \ref{pval_fig1}.
This suggests that some form of spillovers exists, where the distance 
does not seem to matter.
Alternatively, the results could indicate that the spillover effects may be 
heterogeneous with respect to distance. 

\subsection{Covariate adjustment for heterogeneous focals}\label{sec:discussion_het}



The difference between short-range spillovers and pure control affects the power of the 
spillover hypothesis tests. The concern becomes evident in Figure~\ref{nulls2}, where we see that the focal units that receive spillovers are mostly downtown streets, whereas units that are pure controls are mostly on the outskirts of the city.

One straightforward way to address such possible heterogeneity in the focal units is to adjust for known covariates. For example, we could regress outcomes 
on observed covariates,
and then perform our proposed biclique test on the residuals~\citep{rosenbaum2010design}.
To illustrate, we used this approach for the test of Section~\ref{sec:FRT_med} 
and adjusted the outcome by  distance from important societal center points, such as
school, police station, courthouse, church, park, as well as ``comuna'', which represent a neighborhood or district.
The results are shown in Figure~\ref{nulls3}. 
The new p-value is 0.13, and
suggests that outcomes under short-range spillovers are not statistically different from outcomes
under pure control. 
In Appendix~\ref{appendix:spatial}, however, we include a randomization analysis 
for many different radii,
which suggests that spillovers may exist at distances larger than 125m.
This result hints at the insufficiency of geographic distance to fully capture the intensity of spillovers. In future work, we could incorporate additional information in the distance function,  such as socioeconomic differences between street segments.  An advantage of the biclique-based testing methodology is the ability to arbitrarily define the distance function (and thus exposures) of interest.

More broadly, this analysis suggests that  adjusting for heterogeneity 
may be important in practice. The regression-based approach could be extended to adapt related randomization-based approaches that account for treatment effect heterogeneity~\citep{ding2016randomization}. Another approach could be 
to balance the focal units while incorporating covariates. 
We leave these ideas open for future work.

\begin{figure}[t!]
\centerline{
\includegraphics[scale=0.64]{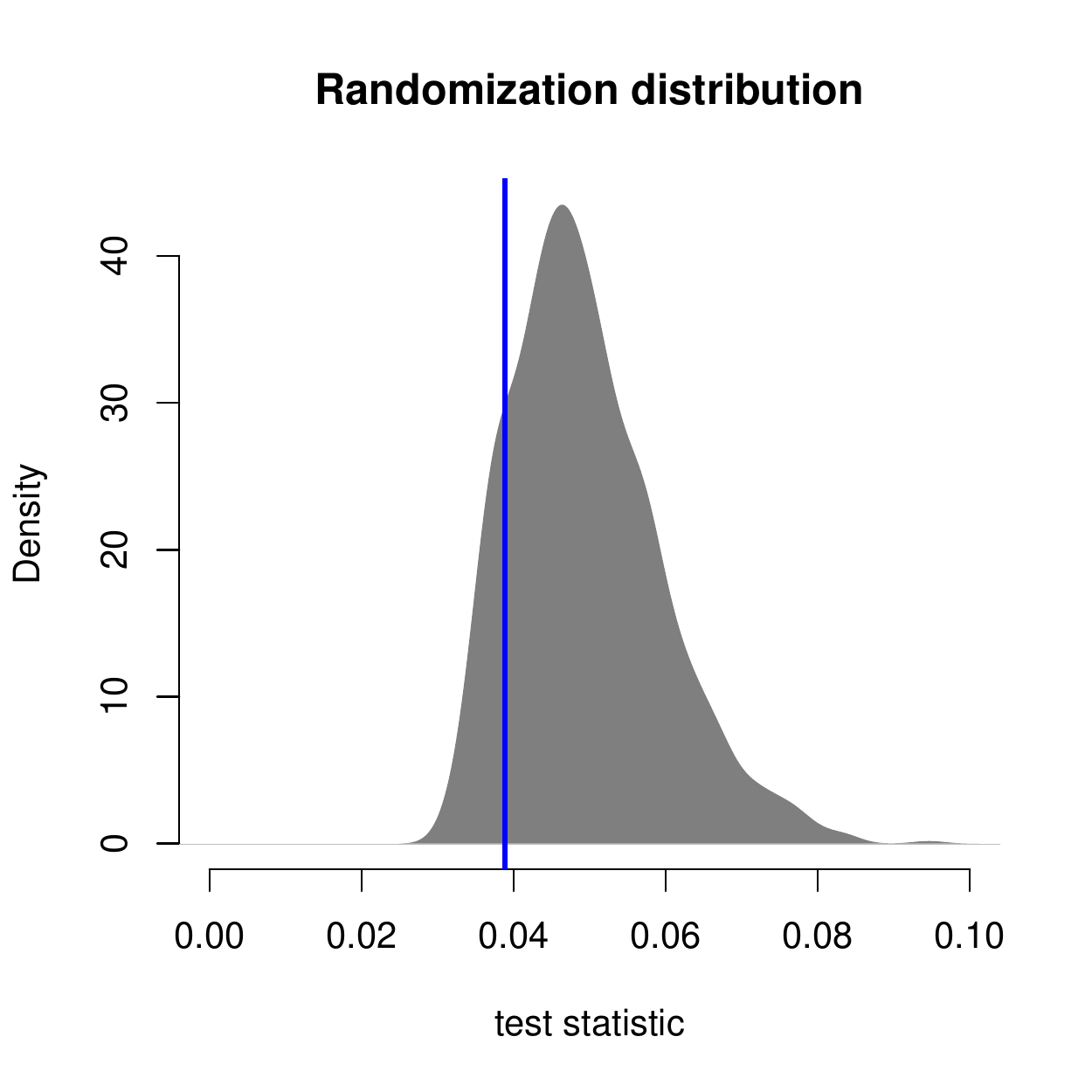}}
\caption{Randomization test for 125m spillover radius hypothesis, 
applied on the residuals of a regression of outcomes on known covariates.
%
 The p-value of the observed test statistic is approximately 0.13.}
\label{nulls3}
\end{figure}


%

%
%
\section{Extension to complex null hypotheses}
\label{sec:composite}

Thus far, we have focused on testing null hypotheses that contrast two exposures, as in Equation~\eqref{eq:H0}. 
We now turn to testing more complex null hypotheses that can be written as intersections of singleton hypotheses, such as testing null hypotheses that restrict interference between units. 
The main difficulty in testing such hypotheses through our framework is that they are based on multiple sets of exposures, rather than just one as above.
Thus, constructing the null exposure graph requires more care.


Here, we extend our biclique method to test composite intersection hypotheses. Let
 $\Inter = \{\Fset_j\subset \Edom, j=1, \ldots, J\}$ be a set of non-overlapping exposure sets, such that $\Fset_j \cap \Fset_{j'} = \emptyset$ when $j\neq j'$.
Our goal is to develop a biclique test for the  intersection hypothesis defined as follows:
\begin{equation}\label{eq:H0_inter}
\Hinter = \bigcap_{j=1}^J H_0^{\Fset_j}.
\end{equation}
As discussed in Section~\ref{sec:H0}, when $\Inter$ is a partition of $\Edom$ into singletons, such that $\bigcup_j \Fset_j = \Edom$ and $|\Fset_j| = 1$, then the intersection hypothesis in~\eqref{eq:H0_inter} is equivalent to the exclusion-restriction hypothesis, $\Hex$, of Equation~\eqref{eq:H0_ex}.
%
%
While $\Hex$ is a special case of $\Hinter$, we focus on it here because of its connection to standard hypotheses on the ``extent of interference''~(see Example~\ref{example_extent}).
We describe two methods to test $\Hex$. The first method ``patches together'' the null exposure graphs that correspond to the individual hypotheses $H_0^{\Fset_j}$.
The second method extends the approach of~\citet{athey2018exact} to leverage information from the multiple exposure graphs that correspond to $\Hex$.
Finally, in Appendix~\ref{appendix:Hex}, we extend both of these methods 
to test $\Hinter$ in its more general form~\eqref{eq:H0_inter}.

\subsection{Multi-null exposure graph}

Our first approach to test $\Hex$ relies on the key observation that, for any unit $i$,
$\Hex$ essentially splits $\Zdom$ into equivalence classes, such that any two assignments 
within the same equivalence class induce the same potential outcome for unit $i$.
In general, these equivalence classes are allowed to differ across units. In practice, however, these classes typically
 correspond to the same exposure level, and thus have the same interpretation across units; e.g., the exposure levels 
``pure control'', ``spillover$_r$'', and ``other'' in~\eqref{eq:f} for any unit $i$ can also be viewed as equivalence classes 
of $\Zdom$.
This differs from $\Hab$, which posits equal potential outcomes only within the two specified classes, $\aa$ and $\bb$, rather than within every class. 
To test $\Hex$, we will therefore need to expand the definition of the null exposure graph.

In particular, we define the {\em multi-null exposure graph} of $\Hex$ with respect to $z\in\Zdom, \Zset\subseteq \Zdom$ as follows:
\begin{equation}\label{eq:def_Gz}
G(z; \Zset) = (V, E),~\text{with}~V=\Udom~\text{and}~E = \{(i, z') : i\in\Udom, z'\in\Zset,~f_i(z') = f_i(z)\}.
\end{equation}
In words, $G(z; \Zset)$ encodes which potential outcomes would be imputable 
if $z$ is realized in the experiment, just like the simpler null exposure graph above. 
We need this new definition because $\Hex$ is comprised of many individual null hypotheses, $H_0^{\Fset_j}$.
As such, for any given pair $(i, z')$, the missing potential outcome $Y_i(z')$
could be imputed via a different $H_0^{\Fset_j}$ depending on which assignment $\Zobs$ is realized 
in the experiment.
 The idea for a test of $\Hex$ is then to partition $\Zdom$ through a ``patchwork'' of bicliques from multi-null exposure graphs. This idea is formalized below.
\begin{procedure}\label{proc:multi}
To generate a biclique decomposition for testing $\Hex$, we do the following:
\begin{enumerate}
\item Initialize $\Cset \gets \emptyset$, and $\Zdom_0 \gets \Zdom$. 
\item While $\Zdom_0 \neq \emptyset$:
\begin{enumerate}[(a)]
\item Sample $Z$ uniformly from $\Zdom_0$.
\item Obtain a non-empty biclique $C = (U, \mathcal{Z})$ within $G(Z; \Zdom_0)$; e.g., as in~Equation~\eqref{eq:max_clique}.
\item Update $\Cset \gets \Cset\cup \{C\}$ and $\Zdom_0 \gets \Zdom_0 \setminus \Zset$.
\end{enumerate}
\end{enumerate}
\end{procedure}
The main output of this procedure is a biclique decomposition, i.e., a collection $\Cset$ of bicliques that fully partitions the treatment assignment space~(see Definition~\ref{def:decomposition}).
However, it is different from the earlier definition since the decomposition here is not applied to a single 
null exposure graph, but is instead comprised of bicliques from various null exposure graphs, $G(z;\Zset)$.
To test $\Hex$ we then simply condition the test of Procedure~\ref{proc:test} on the biclique of $\Cset$ that contains $\Zobs$, as
in the original test of Procedure~\ref{proc:test}.
We prove the validity of this test in the following theorem.

\begin{theorem}\label{thm2}
\ThmTwo
\end{theorem}

Theorem~\ref{thm2} shows that our proposed method can in fact test more complex hypotheses than the contrast hypothesis of Equation~\eqref{eq:H0}. 
As with contrast hypotheses, the power of this test mainly depends on the sizes of the conditioning bicliques in Step 2(b) of Procedure~\ref{proc:multi}; 
see also the power analysis in Section~\ref{sec:power}. 
However, Procedure~\ref{proc:multi} offers no guarantees for power because, like the decomposition algorithm of Section~\ref{sec:decomposition}, 
it also employs a greedy approach in constructing bicliques. We leave for future work a more sophisticated procedure that could overcome this.

\subsection{Approach based on Athey et.~al.~(2018)}\label{sec:multi_athey}
Hypothesis $\Hex$ can also be tested using the approach of~\citet{athey2018exact}. 
As described in Section~\ref{sec:comparison}, this approach is equivalent to conditioning on a biclique $C = (U, \Zset)$, where 
we first sample $U$ using an arbitrary distribution $g(U)$ on $\Udom$, and then set
$\Zset = \{ z\in\Zdom: f_i(z) = f_i(\Zobs),~i\in\Uset \}$.
However, as also argued earlier, this approach may lack power because an arbitrary sampling of focal units $U$ 
may lead to small bicliques, which, as we discuss in Section~\ref{sec:power}, can be detrimental to power.

We therefore propose to adjust $g(U)$ to give more weight to focal units that allow for larger bicliques.
Specifically, 
we first estimate the average degree of every unit in $\Udom$ over a large sample  of multi-null exposure graphs.
We then sample units with larger weight to higher-degree units; these units are on average connected to more assignments, and so will lead to larger conditioning  bicliques.
The concrete procedure can be described as follows:
\begin{enumerate}
\item Initialize, $d_i \gets 0$ for all $i\in\Udom$.
\item For $r=1, \ldots, R$ replications:
\begin{enumerate}[(a)]
\item Sample $Z \sim \pr(Z)$ according to the design.
\item Set $d_i \gets d_i + \mathrm{deg}_i(G(Z; \Zdom))$, where $\mathrm{deg}_i(G)$ denotes the degree of node $i$ in graph $G$.
\end{enumerate}
\item Define $g(U)$ as sampling over $\Udom$ weighted by $(d_1,\ldots, d_n)$, 
and apply the approach of~\citet{athey2018exact}.
\end{enumerate}
This test is valid using the same justification as for the original test of~\citet{athey2018exact} since we still randomly sample the focal units independently of $\Zobs$. However, our approach may increase power as the sampled focal units will generally 
lead to larger  bicliques.  We leave to future work a more detailed power comparison with the original and ``$\epsilon$-nets'' methods of~\citet{athey2018exact}.

\CB

\section{Concluding remarks}
In this paper,  we extend the classical Fisher randomization test to settings with general interference.
Our main contribution is the concept of the null exposure graph, which 
represents the null hypothesis as a bipartite graph between units and assignments.
Conditional on a  biclique in this graph, the null hypothesis is sharp, and the corresponding (conditional) randomization test is valid. 
We showed the benefits of this approach in both clustered and spatial interference settings. Furthermore, our approach inherits the properties and theory 
of conditional randomization tests. This includes covariate adjustment methods 
\citep{rosenbaum2002covariance}, as well as recent developments on studentization 
of test statistics \citep{wu2018randomization}.
%


There are a number of promising directions for future work. 
First, we can explore how to combine information across biclique tests that condition on different bicliques. For instance, we could follow \citet{geyer2005fuzzy} and leverage the distribution of $p$-values across tests to improve power, as discussed in \citet{basse2019randomization}. We could also adapt recent proposals on multiple randomization tests from \citet{zhang2021multiple}.
Second,
 we could investigate how much data ``is thrown away'' 
 by conditioning, and suggest biclique decompositions of the null exposure graph that minimize data loss. 
 Third, building on our results for power,
 we could further develop the problem of optimal design  for a given set of spillover hypotheses.
 Finally, it would be interesting 
 to know under which conditions our proposed tests can be implemented more efficiently.

 \appendix

\setcounter{theorem}{0}
\setcounter{figure}{10}

\section{Null hypothesis $\Hf$}\label{appendix:H0}
An alternative way to view the null hypothesis, $\Hf$, in Equation (\ref{eq:H0}), is to think of it as a composite hypotheiss.
In particular, $\Hf$ is equivalent to testing a composite hypothesis, denoted as $\{H_{0, k}\}$, where $H_{0,k}$ is any simple hypothesis of the form $H_{0, k}:~ \mathbf{Y} = \mathbf{Y}_{0,k}$.  Here, $\mathbf{Y}_{0,k}$ is any $N\times 2^N$ potential outcomes matrix with the 
following constraints: (\textit{i}) it has $Y^{obs}$ at the column corresponding to $Z^{obs}$, (\textit{ii}) when  unit $i$ under $z$ is exposed, i.e., $f_i(z)\in\{a, b\}$, its outcome in $\mathbf{Y}_{0,k}$ is fixed 
to $Y^{obs}_i$. 
This reveals that testing under interference entails a problem of identification.
Specifically, the challenges here are that the simple hypotheses are not testable since not all units are exposed under all assignments, and 
the number of $\mathbf{Y}_{0,k}$ is intractably large.  

%

\section{Proofs}\label{appendix:proofs}

\subsection{Main results}

\addtocounter{theorem}{1}
\begin{theorem}
\ThmOne
\end{theorem}
\begin{proof} 
Let $\Zobs$ be the observed assignment and $C$ be the biclique that 
contains $\Zobs$. Let $\Zset(C)$ denote the set of assignments in $C$, 
In the formalism of~\citet{basse2019}, $C$ is the 
conditioning event of our test.
We will use \citep[Theorem 1]{basse2019} to prove the validity of our proposed test.
This requires to show that the following two  conditions hold.  For notational simplicity, below we implicitly condition on $\Cset$ and also assume 
that $\Ind{C\in\Cset}=1$ whenever we use ``$C$'' to denote a conditioning biclique. \\

\noindent\underline{1. Imputability of test statistic}:
The potential outcomes are imputable within the biclique $C$ under $\Hf$, since, by definition, 
the biclique units are exposed to some exposure in $\Fset$ for any assignment in the biclique.Our test statistic is using only outcomes from units and assignments within the biclique, and so the condition 
in Equation~(4) of \citet[Theorem 1]{basse2019} holds. \\

\noindent\underline{2. Correct randomization distribution.}:
It remains to show that $r(Z) = p( C| Z)$; i.e., that the randomization distribution, $r(Z)$, of our test 
coincides with the actual conditional distribution, $P(Z | C) \propto P(C | Z) \pr(Z)$, 
induced by the conditioning mechanism $p(C|Z)$~\citep[Section 3.2]{basse2019}, 
and the design $\pr(Z)$.
The conditioning mechanism of our procedure is equal to:
\begin{equation}\label{eq0}
p(C | \Zobs=z) = \Ind{z \in  \Zset(C)},
\end{equation}
since the test simply conditions on the biclique that contains the assignment.
The marginal probability of conditioning on biclique $C$ is therefore equal to:
\begin{equation}\label{eq1}
p(C) = \sum_{z'} p(C | z') \pr(z') =  \sum_{z'}  \Ind{z' \in  \Zset(C)} \pr(z').
\end{equation}
The randomization distribution defined in Step 3 of the testing procedure of Section~\ref{sec:test} is equal to:
\begin{align}\label{eq:rz}
r(z)  &= \Ind{z \in \Zset(C)} \frac{\pr(z)}{\sum_{z'} \Ind{z' \in \Zset(C)} \pr(z')} \nonumber&\commentEq{By definition, in Step 3 of the biclique test}\\
& = \Ind{z \in \Zset(C)} \frac{\pr(z)}{p(C)}
\nonumber&\commentEq{by Equation~\eqref{eq1}}\\ 
& = P(C | Z) \frac{\pr(z)}{p(C)}
\nonumber&\commentEq{by Equation~\eqref{eq0}}\\ 
& = P(Z | C).
\end{align}
The conditional validity of our test now follows from \citep[Theorem 1]{basse2019}.
\end{proof}

\addtocounter{theorem}{1}
\begin{theorem}
\ThmTwo
\end{theorem}
\begin{proof}
The main difference with Theorem~\ref{thm:main} is that the test now also conditions on $A = \{A_i : i=1, \ldots, N\}$, the exposure sets of each unit.
In this setting, there is no single biclique decomposition. 
Instead, for every possible value of $A$ there corresponds 
one biclique decomposition, say, $D(A)$.
Thus, Equation~\eqref{eq0} is updated to:
\begin{equation}\label{eq0_2}
p(C | \Zobs=z) = \Ind{C \in  D(A_z)} \Ind{z \in  \Zset(C)},
\end{equation}
where $A_z$ denotes the value of $A$, which is uniquely determined by the observed assignment $z$.
We will show that the equality
$$
\Ind{C \in  D(A_z)} \Ind{z \in  \Zset(C)} = \Ind{z \in  \Zset(C)}
$$
is guaranteed by construction of our null exposure graph. 
For that, it suffices to show that $\Ind{C \in  D(A_z)}=0$ implies that 
$ \Ind{z \in  \Zset(C)} = 0$.
We prove by contradiction. 
Suppose that $\Ind{C \in  D(A_z)}=0$ and $\Ind{z \in  \Zset(C)} = 1$
for some biclique $C$ and observed assignment $z$.
Note that the construction of the null exposure graph
implies that all units in the null exposure graph receive exposures contained in $A_z$.
Thus, $ \Ind{z \in  \Zset(C)} = 1$ implies that all units in $C$ receive exposures contained in $A_z$. However, since the exposures in $\Fset_{[i]}$ are all distinct, 
$\Ind{C \in  D(A_z)}=0$ implies that there exists at least one unit in $C$ 
that receives an exposure that is not in $A_z$. This is a contradiction.
We conclude that:
$$
p(C | \Zobs=z) = \Ind{z \in  \Zset(C)},
$$
and so the rest of the proof of Theorem~\ref{thm:main} follows.

\end{proof}


\subsection{Implementation for arbitrary designs}
\label{sec:arbitrary-designs}

We now address a practical implementation issue.  Although our method works for arbitrary designs, it can become
computationally intractable if the support, $\Zdom = \{z : \pr(z) > 0\}$, is too large~(on the order of hundreds of thousands of nodes), since biclique enumeration is NP-hard. Fortunately, a small modification of our test can address this issue. The idea is to add a step at the beginning of Procedure~\ref{proc:test} that subsamples assignments to limit the size of $\Zdom$. 
We now show that the following procedure is still valid:
\begin{enumerate}
\item Draw $\Zobs \sim P(\Zobs)$.
\item Draw $M-1$ assignments uniformly at random from 
$\Zdom \setminus \{\Zobs\}$ and let
  $\Zdom_M$ be the set of size $M$ formed as the union of $\{\Zobs\}$ and the set of size $M-1$ just constructed.
\item Run our biclique test in Procedure~\ref{proc:test}, using the null exposure graph
of the new support set, $\Zdom_M$.
\end{enumerate}

The key for the proof is to consider the conditioning event as the pair $(C, \Zdom_M)$, where $\Zdom_M$ is the sampled support set, 
and $C$ is the biclique we condition on. The original test in Procedure~\ref{proc:test} conditions only on a biclique, and $\Zdom_M = \Zdom$.
We obtain:
\begin{align}\label{eq:arb}
  P(Z | C, \Zdom_M)
  &\propto P(C, \Zdom_M | Z) P(Z) \nonumber\\
  &\propto P(C | \Zdom_M, Z) P(\Zdom_M\mid Z) P(Z) \nonumber\\
  &\propto \Ind{Z \in C} \Ind{C \in \Zdom_M} P(Z),
\end{align}
where we used the fact that $P(\Zdom_M | Z)$  is independent of $Z$ 
by construction of $\Zdom_M$; also, 
$ P(C | \Zdom_M, Z)  =  \Ind{Z \in C} \Ind{C \in \Zdom_M} $ because 
we simply condition on  the biclique in $\Zdom_M$ that assignment $Z$ is contained in.
Equation~\eqref{eq:arb} ensures validity of this test if we simply make sure to 
make a biclique decomposition on $\Zdom_M$, so that  $\Ind{C \in \Zdom_M} =1$. 
The remainder terms in the randomization distribution, 
$\Ind{Z \in C} P(Z)$, correspond to the sampling distribution of Procedure~\ref{proc:test}~(Step 3).

\section{More on clustered interference}\label{appendix:clustered}
Here, we continue our discussion on testing power in Section~\ref{sec:clustered_sim}. In Figure \ref{app:power_cliq2}, we reexamine how biclique characteristics affect the testing power.  The data shown are the same as in Figure \ref{power_cliq1}, now displayed on a single plot with color denoting power.  Each dot corresponds to a different biclique decomposition of the null exposure graph, and we compute the average number of assignments and units within each decomposition.  Requiring more assignments in a biclique will make including more units a challenge.  Intuitively,  this results from the graphical nature of bicliques -- 
they are complete subgraphs, and including more left nodes will dampen the size of the right node set.  This inverse (nonlinear) relationship can be seen on the plot.  Note also that there is a balancing of power for different sized cliques.  The highest powered tests come from bicliques with $\sim 150$ units and $\sim 25$ assignments.  In practice, this is a tradeoff that can be navigated.

\begin{figure}[H]
\centerline{\includegraphics[scale=0.7]{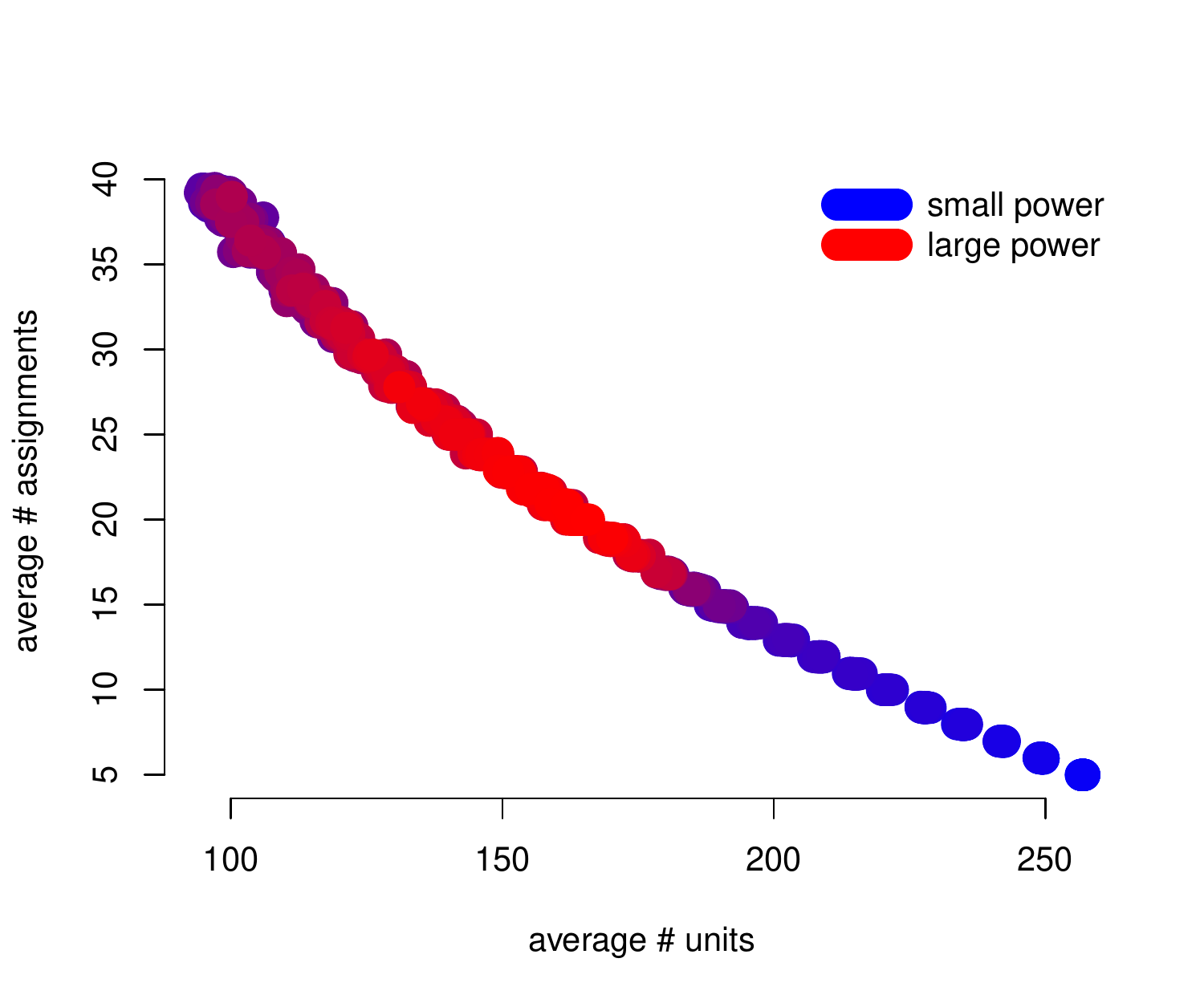}}
\caption{Average number of biclique assignments and units versus the power.  Each dot corresponds to biclique decomposition of the null exposure graph, and the color denotes the power value from the simulation.  Red (blue) correspond to large (small) power values.}
\label{app:power_cliq2}
\end{figure}

\section{Design-assisted biclique test}
\label{sec:design_test_appendix}

We now define the 
``design-assisted'' biclique test.
\begin{procedure}\label{proc:test2}
For the two-stage experimental design in the cluster interference setting of Section~\ref{sec:clustered_sim}, 
define the following procedure.
\begin{enumerate}
\item For each cluster $k=1, \ldots, K$,
\begin{enumerate}[(a)]
\item Let $U_k\subseteq \Udom$ denote the set of units in cluster $k$, and $\Zdom_k\subseteq \Zdom$ the set of 
possible treatment assignments of cluster $k$~($N/K$-length binary vector).

\item Take $U_{k,1}$ to be a random half of $U_k$ and $U_{k, 2} = U_k \setminus U_{k, 1}$. Define $\mathbb{Z}_{k, j} = \{ z \in \Zdom_k : z_i = 1~\&~i \in U_{k, j}\}$ for $j=1, 2$, as the set of assignments for which the cluster is treated, 
and the single unit that is treated is from $U_{k, j}$. 
\item Define the following two bicliques in cluster $k$: $C_{k, 1} = (U_{k, 1}, \mathbb{Z}_{k, 2})$ 
and $C_{k, 2} = (U_{k, 2}, \mathbb{Z}_{k, 1})$. 
\item Let $\mathbf{0}_k$ denote the ``all control'' assignment in cluster $k$. We attach this assignment 
to one of the two bicliques from step 1(c) at random.
That is, with probability $1/2$ we set $C_{k, 1} = (U_{k, 1}, \mathbb{Z}_{k, 2} \cup \{ \mathbf{0}_k \})$, 
otherwise we set $C_{k, 2} = (U_{k, 2}, \mathbb{Z}_{k, 1} \cup \{ \mathbf{0}_k \})$.
\end{enumerate}
\item Finally, let $Z$ be the global assignment vector and let $Z_{[k]}$ denote the 
subvector corresponding to cluster $k$. Let $C_k(Z)$ be the unique biclique from 
$\{C_{k, 1}, C_{k, 2}\}$ such that $Z_{[k]} \in c(k, Z)$.
We define the following conditioning mechanism:
$$
P(C | Z) \propto \Ind{ C = \bigcup_{k=1}^K  C_k(Z)}.
$$
\item Run the randomization test of Procedure~\ref{proc:test} conditional on $C \sim p(C | \Zobs)$ as usual.
\end{enumerate}
\end{procedure}

The above procedure is using the structure of the two-stage experimental design to define 
a better conditioning mechanism, $P(C|Z)$.
The structure of the design is leveraged mainly in step 1(c).
Specifically, note that only one unit is being treated within each cluster, and that we are testing the null hypothesis in~\eqref{eq:H0cluster}, which allows imputation of potential outcomes for some unit $i$ for any cluster assignment for which $i$ is not treated.
Thus, $C_{k, 1}, C_{k, 2}$, as defined, are indeed bicliques. 
The rest of the Procedure~\ref{proc:test2} makes sure to utilize these two large bicliques per cluster in the joint conditioning mechanism, which leads to conditioning the test on larger bicliques. Theorem~\ref{thm:power} implies that this can yield more powerful tests. 
The following theorem shows that this procedure is valid. 
\begin{theorem}\label{thm:design_aid}
\ThmDesignAid
\end{theorem}

The proof of Theorem~\ref{thm:design_aid} is straightforward and relies on the fact that, by construction,
the conditioning mechanism relies on a partitioning of the assignment space, such that 
$P(C | Z) \propto \Ind{Z \in C}$ as in Procedure~\ref{proc:test}.


\begin{proof}
By definition of the two-stage design, there is a constant proportion, say $p$, of clusters being treated out of $K$ in total. 
As defined in the main text, let $C_k(Z)$ denote the unique cluster from set $\{C_{k,1}, C_{k, 2}\}$ in cluster. By construction, when $Z_{[k]}=\mathbf{0}_k$ --- i.e., when the cluster is in control ---then 
$P(C_k(Z) = C_{k, j}) = 1/2$ for $j=1, 2$. And when $Z_{[k]}\neq \mathbf{0}_k$ --- i.e., when the cluster is treated ---then 
$P(C_k(Z) = C_{k, j}) = \Ind{Z_{[k]} \in C_{k, j}}$, for $j=1, 2$.

Now, for some biclique $C$ let $C = \bigcup_{k=1}^K C_k$ be its unique decomposition into cluster bicliques. Following the definition of the conditioning mechanism in Procedure~\ref{proc:test2}.
\begin{align}
P(C | Z) & = \prod_{k=1: Z_{[k]} = \mathbf{0}_k}^K P(\mathbf{0}_k \in C_k) 
 \prod_{k=1: Z_{[k]} \neq \mathbf{0}_k}^K P(Z_{[k]} \in C_k) 
 & \comm{By definition of conditioning mechanism}\nonumber\\
& = \prod_{k=1: Z_{[k]} = \mathbf{0}_k}^K \Ind{\mathbf{0}_k \in C_k} (1/2)
\prod_{k=1 : Z_{[k]} \neq \mathbf{0}_k}^K  \Ind{Z_{[k]} \in C_k} 
& \comm{From analysis above} \nonumber\\
& = 2^{-(1-p)K} \prod_{k=1}^K  \Ind{Z_{[k]} \in C_k}
=  2^{-(1-p)K} \cdot \Ind{Z \in C}. & \comm{Always $(1-p) K$ clusters in control}\nonumber
\end{align}

Recall that in our original test of Procedure~\ref{proc:test}, $p(C | Z) \propto \Ind{Z \in C}$. 
We can now follow the proof of Theorem~\ref{thm:main}. 
The additional term 
$2^{-(1-p)K}$ in the revised procedure does not affect the proof because it is fixed and drops from the calculations in the final derivations for $r(Z)$ in Equation~\eqref{eq:rz}.
\end{proof}

\section{Power analysis}\label{app:powertheory}

\subsection{Proof of Theorem~\ref{thm:power}}
Here, we prove the main theorem on power analysis of the biclique test.
\addtocounter{theorem}{-3}
\begin{theorem}
\ThmPowerApp
\end{theorem}
\begin{proof}
\begin{lemma}\label{L2}
Under Assumptions (A2) and (A3), it holds:
$$
1 - \Fn\Fo^{-1}(1-\alpha) = 1 - F\big(F^{-1}(1-\alpha) - \tau/\sigma_n\big).
$$
\end{lemma}
\begin{proof}
First, note that $y = \Fo(x)$ implies that $y = F(x/\sigma_n)$ from (A2) and 
so $x = F^{-1}(y) \sigma_n$; thus, $\Fo^{-1}(y) = \sigma_n F^{-1}(y)$.
Now, for the first part:
\begin{align}
1 - \Fn\Fo^{-1}(1-\alpha) & =  1 - \Fo\big(\Fo^{-1}(1-\alpha) - \tau\big) & \comm{From (A3)}\nonumber\\
 & =  1 - F\big( \Fo^{-1}(1-\alpha)/\sigma_n - \tau/\sigma_n\big) & \comm{From (A2)}\nonumber\\
  & =  1 - F\big( F^{-1}(1-\alpha) - \tau/\sigma_n\big). & \comm{From result above}\nonumber
\end{align}
\end{proof}
\begin{lemma}\label{L3}
Let $\Fo^{-1}(1-\alpha)  \triangleq \qa$ and  
$\Fom^{-1}(1-\alpha)    \triangleq \Qa$.
Then, for small enough $\epsilon$ there exists $\lambda > 0$ such that
$$
P( | \Qa - \qa | > \epsilon) \le 2 \exp(- \lambda m \epsilon^2).
$$
\end{lemma}
\begin{proof}
Assumption~(A1) implies that $\Qa$ is the sample $(1-\alpha)$-quantile  
of $\qa$, and so
$$
P( | \Qa - \qa | > \epsilon) \le 2 \exp(- m d_\epsilon^2),
$$
where $d_\epsilon = \min\{ \Fo(\qa + \epsilon) - \Fo(\qa), \Fo(\qa) - \Fo(\qa - \epsilon)\}$.
See, for example, \citep[Theorem 2]{xia2019non}.
Since $\Fo$ is continuous  and bounded, we can write 
$\Fo(\qa + \epsilon) = \Fo(\qa) + f_{0, n}(\qa) \epsilon + O(\epsilon^2)$, 
and $\Fo(\qa - \epsilon) = \Fo(\qa) - f_{0, n}(\qa) \epsilon + O(\epsilon^2)$, where 
$f_{0, n}$ is the pdf that corresponds to $F_{0, n}$.
Therefore, $d_\epsilon/\epsilon = f_{0, n}(\qa) + O(\epsilon)$. 
So,  for small enough $\epsilon$,  there exists $\lambda > 0$ such that 
$d_\epsilon \ge \lambda \epsilon$.
\end{proof}

\begin{lemma}\label{L4}
Let $f_m(x) = \lambda m^{-x} + 2 e^{-\lambda m^{1-2x}}$.
Then, for large enough $m$ the minimum of $f_m(x)$ is in $(1/2, 1 + O(\log^{-1} m)]$. Specifically, for any large enough $m$,
$$
\min_{x\in\mathbb{R}^+} f_m(x)  = \min_{x\in(1/2,~1 +  O(\log^{-1} m)]} f_m(x).
$$
\end{lemma}
\begin{proof}
Let $g_m(x) = \lambda m^{-x} = \lambda e^{-x \log m}$, then we can write $f_m(x)$ as 
$f_m(x) = g_m(x) + 2 e^{-(m/\lambda) g_m^2(x)}$.
By differentiation, the minimum of $f_m(x)$ is at
\begin{align}\label{eq2829}
g_m'(x) - 4 e^{-(m/\lambda) g_m^2(x)}(m/\lambda) g_m'(x) g_m(x) & = 0 \nonumber\\
 e^{(m/\lambda) g_m^2(x)} & = 4(m/\lambda) g_m(x) \nonumber\\
(m/\lambda) g_m^2(x) & = \log \big(  4(m/\lambda) g_m(x) \big) \nonumber\\
(m \lambda ) e^{-2x \log m} & = \log(4 m) + \log e^{-x\log m} \nonumber\\
 (m \lambda) e^{-2y} & = \log(4  m) - y~\quad~\comm{we set $y = x \log m$} \nonumber\\
y+  m \lambda e^{-2y} & = \log(4  m).
\end{align}
It is straightforward to see from~\eqref{eq2829}  that for large enough $m$ the solution to this equation satisfies
$y = \theta \log(4 m)$ for some $\theta\in[1/2, 1]$. 
Indeed, when $\theta=0.5$, the RHS in~\eqref{eq2829} is larger than the LHS  for large enough $m$; and when $\theta=1.0$,  the LHS is larger.
Thus, $x = \theta \log(4 m)/\log m
= \theta (1 + O(1/\log m)) \in (1/2, 1+O(\log^{-1} m))$.
\end{proof}

\begin{lemma}\label{L5}
Suppose that (A1) holds, and fix any arbitrarily small constant $\delta>0$, with $\delta<0.5$. Then, for large enough $m$,
\begin{align}
\EX(\hat F_{1, n, m}(z) - F_{1, n}(z) )= O(m^{-0.5 + \delta}),~\text{for any $z\in\mathbb{R}$}.
\end{align}
\end{lemma}
\begin{proof}
Let $\Delta = \hat F_{1, n, m}(z) - F_{1, n}(z)$. Then,
\begin{align}
\hat F_{1, n, m}(z) - F_{1, n}(z) & \le 
| \Delta | \cdot \mathbb{I}\{ |\Delta | \le \epsilon\} + |\Delta| \cdot \mathbb{I}\{ |\Delta | > \epsilon\}
\quad~\commentEq{for any $\epsilon>0$}
\nn\\
& \le | \Delta | \cdot \mathbb{I}\{ |\Delta | \le \epsilon\} + \mathbb{I}\{ |\Delta | > \epsilon\}
\quad~\commentEq{since $|\Delta| \le 1$, by definition} 
\nn\\
& \le \epsilon + \mathbb{I}\{ |\Delta | > \epsilon\} \nn\\
\EX(\hat F_{1, n, m}(z) - F_{1, n}(z) ) & \le \epsilon + P(|\Delta| > \epsilon)\nn\\
& \le 
\epsilon + P(\sup_z |\hat F_{1, n, m}(z) - F_{1, n}(z)| > \epsilon) \nn\\
& \le \epsilon + 2 e^{-2m\epsilon^2}
\quad\commentEq{from (A1) and the DKW inequality}
\nn\\
& \le m^{-0.5+\delta} + e^{-2m^{2\delta}}
\quad\commentEq{set $\epsilon=m^{-0.5+\delta}$}
\nn\\
& = O(m^{-0.5+\delta}).\commentEq{for large enough $m$}.\nn
\end{align}
\end{proof}

\begin{lemma}\label{L6}
Suppose that Assumptions (A2) and (A3) hold. Then, 
for some $r\in(0.5, 1+O(\log^{-1} m))$,
$$
\EX(\Fn\big(\qa\big) - F_{1, n}\big(\Qa\big)) \ge - O(m^{-r}).
$$
\end{lemma}
\begin{proof}
Let $\Delta_{n,m} = \Fn\big(\qa\big) - F_{1, n}\big(\Qa\big)$. Then,
\begin{align}
\Delta_{n,m}  & \ge - \mathbb{I}\{|\Qa - \qa| \le \epsilon_m\} |\Fn\big(\qa\big) - \Fn\big(\Qa\big)| -  \mathbb{I}\{|\Qa - \qa| > \epsilon_m\}  & \comm{for any $\epsilon_m>0$ } \nonumber\\
& \ge- \mu \mathbb{I}\{|\Qa - \qa| \le \epsilon_m\} | \Qa - \qa| -  \mathbb{I}\{|\Qa - \qa| > \epsilon_m\} & \hspace{-50px}\comm{$\Fn$ is $\mu$-Lipschitz}
\nonumber\\
& \ge - \mu \epsilon_m -  \mathbb{I}\{|\Qa - \qa| > \epsilon_m\} & 
\nonumber\\
\EX(\Delta_{n,m})& \ge - \mu \epsilon_m - P(|\Qa - \qa| > \epsilon_m) &   \nonumber\\
& \ge  - \mu  \epsilon_m - 2 e^{-m \mu \epsilon_m^2} &  
\comm{From Lemma~\ref{L3}} \nonumber\\
& \ge  - \min_{x\in\mathbb{R}^+} \{ \mu m^{-x} + 2 e^{-\mu m^{1-2x}}\} &  
\comm{Tight bound along $\epsilon_m=m^{-x}, x\in\mathbb{R}^+$} \nonumber\\
& \ge  - O(m^{-r}). &  \comm{For some $r\in(1/2, 1)$ by Lemma~\ref{L4}}.\nn
%
\end{align}
\end{proof}

Now, we put all the pieces together. 
%
Let
$$
\Phi_{n, m} =: \EX \big(\Ind{\mathrm{pval}(Z, \Yobs ;  C) \le \alpha} \mid |C|=(n, m), \Qa\big).
$$
From the law of iterated expectation:
\begin{equation}\label{eq:Phi1}
\EX(\Phi_{n,m} \mid  |C|=(n, m)) = \phi_{n,m}.
\end{equation}
Moreover, by Assumption~(A1) we have:
\begin{equation}\label{eq:Phi2}
\Phi_{n, m} = 1 - \hat F_{1, n, m}(\Qa).
\end{equation}
It follows that
\begin{align}
\Phi_{n,m} - (1-\Fn(\qa)) & = \left(\Fn\big(\qa\big) - F_{1, n}\big(\Qa\big)\right)
+ \left(F_{1, n}\big(\Qa\big) - \hat F_{1, n, m}\big(\Qa\big) \right)  &  \commentEq{by Eq.~\eqref{eq:Phi2}}
\nn\\
\EX(\Phi_{n,m}) & \ge 1-\Fn(\qa)  -O(m^{-r})  -O(m^{-0.5+\delta})
& \commentEq{From Lemma~\ref{L5} and Lemma~\ref{L6}}
\nn\\
\phi_{n, m} & \ge 1-\Fn(\qa) -O(m^{-0.5+\delta}) &  
\commentEq{Since $r > 0.5$; and Eq.~\eqref{eq:Phi1}} \nonumber\\
& =  1 - F\big(F^{-1}(1-\alpha) - \tau/\sigma_n\big) - O(m^{-0.5+\delta}). & \comm{From Lemma~\ref{L2}}\nonumber
%
\end{align}
If we also assume that $\sup_{x\in\mathbb{R}} |F(x) - 1/(1+e^{-b x})| \le \epsilon$ for fixed $b, \epsilon>0$, and 
$\sigma_n = O(1/\sqrt{n})$, then
\begin{align}
\phi_{n, m} & \ge  1- \frac{e^{b(F^{-1}(1-\alpha) - \tau c\sqrt{n})}}{1 + e^{b(F^{-1}(1-\alpha) - \tau c \sqrt{n})}} - O(m^{-0.5+\delta}) - \epsilon & \comm{For some $c>0$}\nonumber \\
\phi_{n, m} & \ge  \frac{1}{1 + A e^{-a \tau \sqrt{n}}} - O(m^{-0.5+\delta}) - \epsilon.
& \comm{Where $A = e^{b F^{-1}(1-\alpha)}$ and $a = b c$}\nonumber \\
%
\end{align}

\end{proof}

\subsection{Empirical confirmation of Theorem~\ref{thm:power}}\label{appendix:confirm}
In this section, we perform a simple simulation study to illustrate the result of Theorem~\ref{thm:power}.
In the simulation we assume that the Assumptions (A1)-(A3) of Theorem~\ref{thm:power}.
Another simulation study where the assumptions do not hold is presented in Section~\ref{sec:clustered_sim}.

Thus, we only have to simulation the test statistic values. Let $T$ denote our test statistic. Then, 
we define $T | H_0 \sim N(0, 1/\sqrt{n})$ and $T | H_1 \sim N(\tau, 1/\sqrt{n})$, where 
$\tau$ is the treatment effect. 
This amounts to $\Fo(t) = \Phi(t \sqrt{n})$ and $\Fn = \Phi( (t-\tau)\sqrt{n})$. 
These definitions satisfy Assumptions (A1) and (A2) of Theorem~\ref{thm:power}.

In every simulation run, we repeat the following steps:
\begin{enumerate}
\item Sample $t_j \sim \Fo$, $j=1, \ldots, m$, i.i.d.
\item For a given $\tau$, sample $\tobs \sim \Fn$.
\item Define p-value: $\mathrm{pval} = \sum_{j=1}^m \Ind{t_j \ge \tobs}$.
\item Reject if $\mathrm{pval} \le 0.05$.
\end{enumerate}

In summary, under these assumptions, $n$ --- which is the number of focal units in the biclique test ---
controls the precision of the test statistic distributions.
On the hand, $m$ --- which is the number of focal assignments in the biclique test ---
controls the number of samples taken from $H_0$ in order to test the hypothesis; in other words, 
$m$ controls the support of the randomization distribution.

The results for various values of $(n, m)$ are shown in Figure~\ref{fig1} below.
We see that $n$ controls how fast the power function increases to its maximum (sensitivity).
On the other hand, $m$ controls the maximum power of the test, but not its sensitivity.
This shows that the power of the biclique test is affected by the biclique size in different ways.
\begin{figure}[h!]
\centering
\includegraphics[scale=0.38]{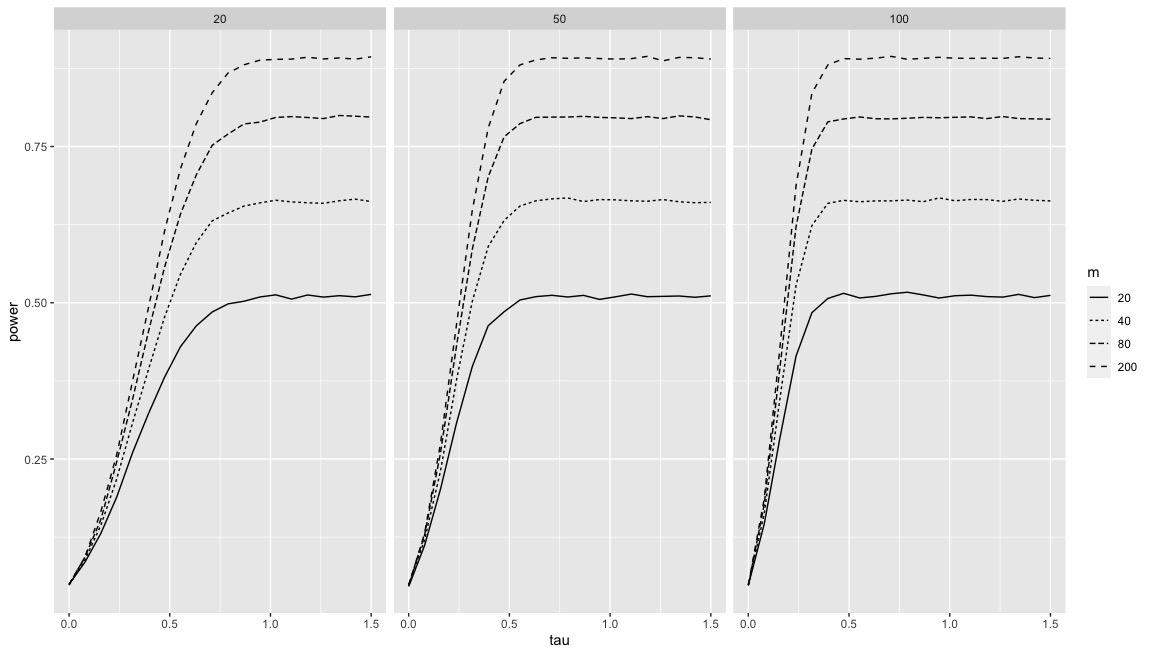}
\caption{\small Power simulation study of the biclique test. The three panels correspond to $n=20, 50, 100$. 
Different line types correspond to various $m$ (number of assignments). The x-axis corresponds to the 
strength of (spillover) treatment effect $\tau$.}
\label{fig1}
\end{figure}

\subsection{Implications for biclique-based randomization tests}\label{app:implications}

The goal of this section is to connect graph-theoretic properties of the null exposure graph to bounds on the size of the biclique drawn from the conditioning mechanism.  Since the structure of the null exposure graph is fixed once the null hypothesis and design are specified, the final goal will be to understand how the test can be conditioned and relate that back to testing power (Theorem~\ref{thm:power}).

We will focus on density of the null exposure graph.  There is classical work in extremal graph theory that relates the number of graph edges to existence of certain-sized bicliques \citep{turan1941external,zarankiewicz1951problem}.  Others have built upon this work for bipartite graphs, including \cite{kovari1954problem} and \cite{hylten1958combinatorical}.  We rely on analysis in the latter two references (especially extensions of what is known as the Kovári-Sós-Turán Theorem), formulated here in the context of the null exposure graph.  See Theorem 2.2 in Chapter 6, Section 2 of \cite{bollobas2004extremal} for the precise bound stated below in the Theorem.
\begin{theorem*}[\cite{bollobas2004extremal}]\label{thmbound}
	Consider the null exposure graph $\GfF=(V, E)$.  Denote the total number of units as $N = |\Udom|$ and assignments as $H = |\Zdom|$. Define integers $0 < n < N$ and $0 < h < H$.  If the following bound holds:
	\begin{equation}\label{NEbound}
		\begin{split}
	|E| \hspace{2mm} \geq \hspace{1mm} (n-1)H + (h-1)^{1/n}H^{1-1/n}(N-n+1), 
		\end{split}
	\end{equation}
	then, there exists a biclique of $\GfF$, $C = (U, \Zset)$, with $U\subseteq \Udom$, $\Zset\subseteq \Zdom$ and number of focal units  $n = |U|$ and focal assignments $h = |\Zset|$.   
\end{theorem*}

This result states that if there are ``enough edges'' in the null exposure graph (defined by the bound), then we are guaranteed existence of a certain-sized biclique. Moreover, we can translate this bound to one on the density of the null exposure graph.
\begin{definition}[null exposure graph density]\label{def:density}
	Given the null exposure graph $\GfF=(V, E)$, its density is defined as the number of edges divided by the total number of possible edges: $\text{d}(\GfF) = \frac{|E|}{NH}$.
\end{definition}
 Taking Definition \ref{def:density} and the above theorem together, we conclude that if the null exposure graph is ``dense enough,'' there exists bicliques of a desired size.  To give a concrete example of this bound, suppose we have an experiment like the clustered interference simulation of Section \ref{sec:clustered_sim} with 300 units.  We generate 100,000 treatment assignments and construct the null exposure graph as in Section \ref{sec:arbitrary-designs}.  In order to guarantee the existence of a biclique with 30 focal units and 500 focal assignments, we would need the null exposure graph density to be at least 85\% (calculated by the bound in (\ref{NEbound})).  If we desired a biclique with 100 focal units and 500 focal assignments, we would need the null exposure graph density to be at least 96\%.
 
The example above illustrates the important connection between density and existence of sufficiently large bicliques for the randomization test.  If the null exposure graph is not very dense, we are likely to condition on bicliques with a small number of focal units and be underpowered (Theorem~\ref{thm:power}).  On the other hand, if the null exposure graph has a large enough density, we will find bicliques with many focal units and the test will be powerful.  Indeed, for a fixed number of units $N$, assignments $H$, and focal assignments $h$, the bound in (\ref{NEbound}) is monotonically increasing in the number of focal units $n$.  In words, we need large density to find large bicliques and have large power.

We now return to the clustered interference setting of Section \ref{sec:clustered_sim}. The power results presented in Figure \ref{powerpanel} are better understood in the context of  the above 
theorem and the discussion presented in Section~\ref{sec:power}.  The cluster size defined by $N/K$ in the experiment determines the density of the null exposure graph.  As before, density is defined as the number of edges in the graph -- given by spillover and control exposures -- divided by the total number of possible edges.  The experiment in Section \ref{sec:clustered_sim} fixed the number of units $N$ and varied the cluster size $K$. The left panel of Figure \ref{app:densitybound} displays how density varies with cluster size.  We can see that the density is increasing in the number of units per cluster.

\begin{figure}[H]
\centerline{\includegraphics[scale=0.55]{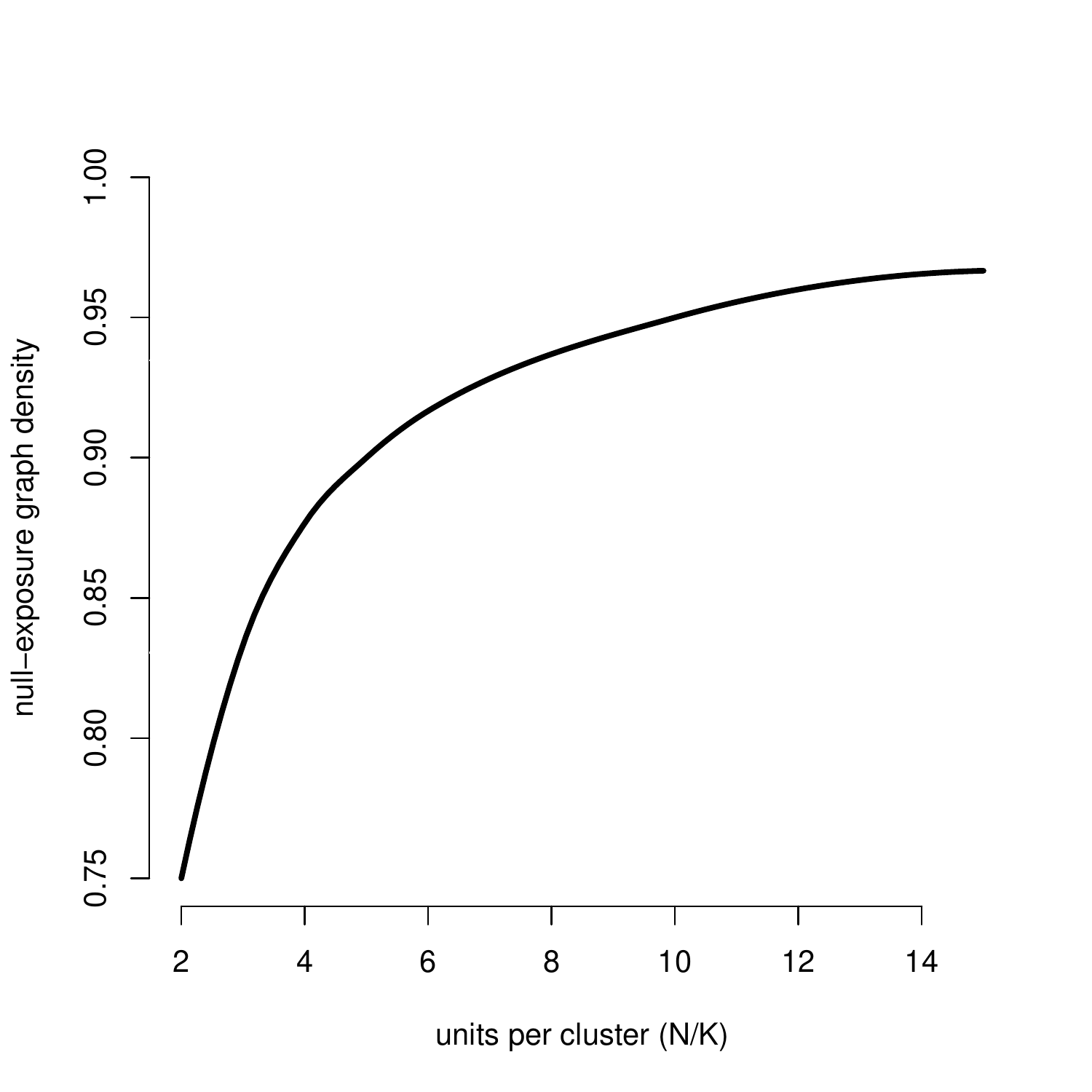}
\includegraphics[scale=0.55]{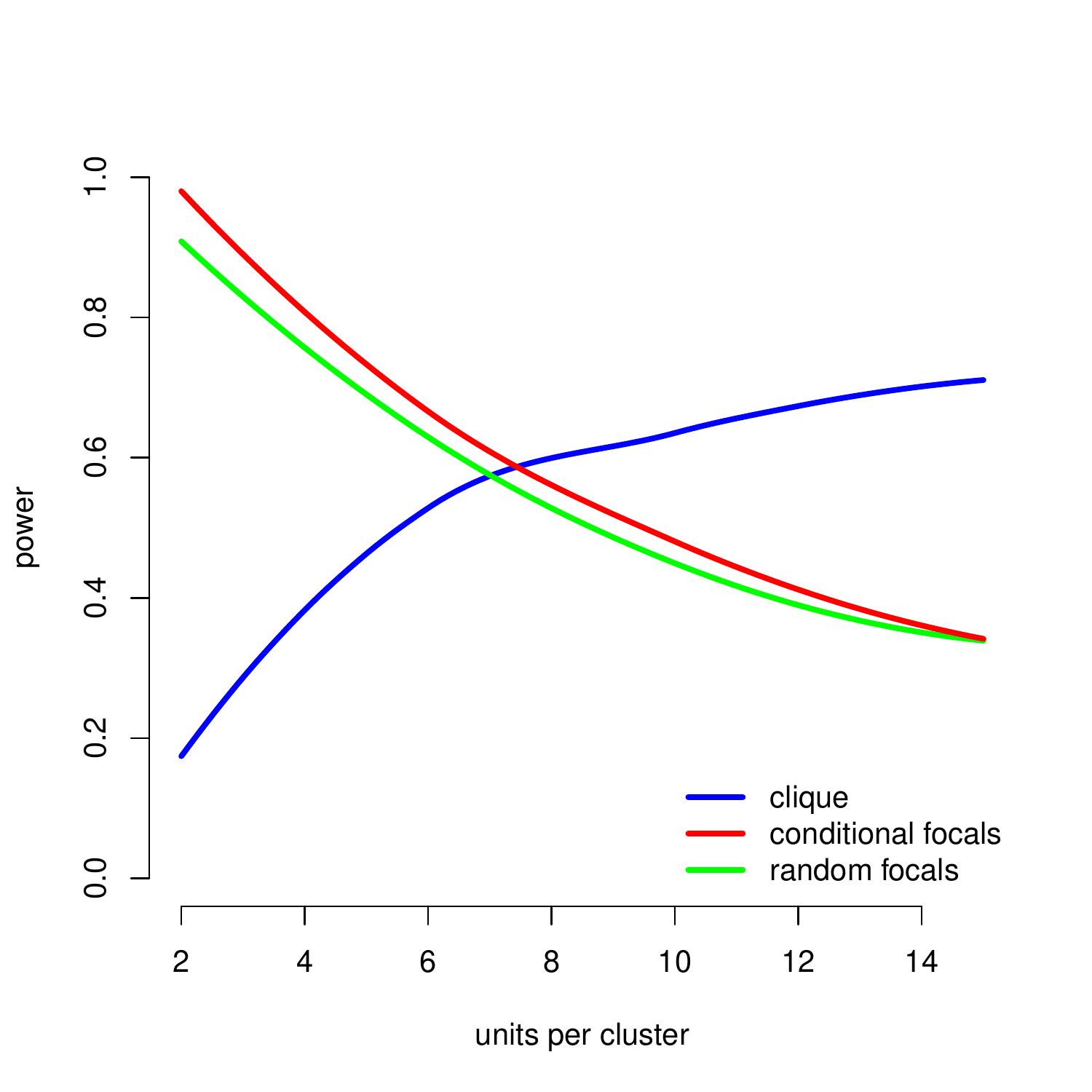}}
\caption{{\em Left:} The null exposure graph density as a function of cluster size for the clustered interference simulation study.~
{\em Right:}. The power of the biclique-based randomization test and competing methods for a fixed $\tau=0.3$ and varying cluster size. }
\label{app:densitybound}
\end{figure}

The right panel of Figure \ref{app:densitybound} shows how power varies with cluster size when $\tau=0.3$ for the biclique method and the methods of \cite{basse2019} (conditional focals) and \cite{athey2018exact} (random focals).  Notice how the density of the null exposure graph follows the same trend as the biclique method's power.  As explored theoretically above, increasing density will lead to increasing power. 

Finally, in light of this discussion, we are able to better understand the power tradeoffs between the three methods on the right in Figure \ref{app:densitybound}. The objective of the graph decomposition algorithm is biclique size, and \textit{not} explicitly power.  Therefore, if the null exposure graph is too sparse, the biclique-based randomization will be underpowered relative to permutation-based alternatives specifically designed for clustered interference.  This is shown on the right panel of Figure \ref{app:densitybound} where the blue line is below the red and green lines.  However, if the null exposure graph is dense enough, the biclique-based test is comparable to and can easily exceed the power of alternative methods.  For our clustered interference setting, the method starts outperforming at around 90\% null exposure graph density (Figure \ref{app:densitybound} -- left) and 6-8 units per cluster (Figure \ref{app:densitybound} -- right).

\subsection{Experimental design}\label{app:design}

In this section, we investigate our method's use for experimental design.  Since the biclique-based randomization test represents the first all-purpose approach for testing under general interference, including spatial intereference, we consider a spatial network like that of the Medell\'in example.  We generate 1000 points from a bivariate Gaussian with non-diagonal covariance to simulate the network (shown in Figure \ref{app:spatial_sim_map}).  Suppose we would like to design an experiment where the probability of treatment in the city center $p_0$ is different from the outskirts $p_1$.  As shown in Figure \ref{app:spatial_sim_map}, the city center is defined by the black circle.  There are 362 units within the city center and 638 units in the outskirts. The design question is: What are the optimal choices of $p_0,p_1$ for testing the spillover hypothesis defined in (\ref{eq:H0_med}) (with spillover radius equal to 0.1)?

\vspace{-8mm}
\begin{figure}[H]
\centerline{\includegraphics[scale=0.55]{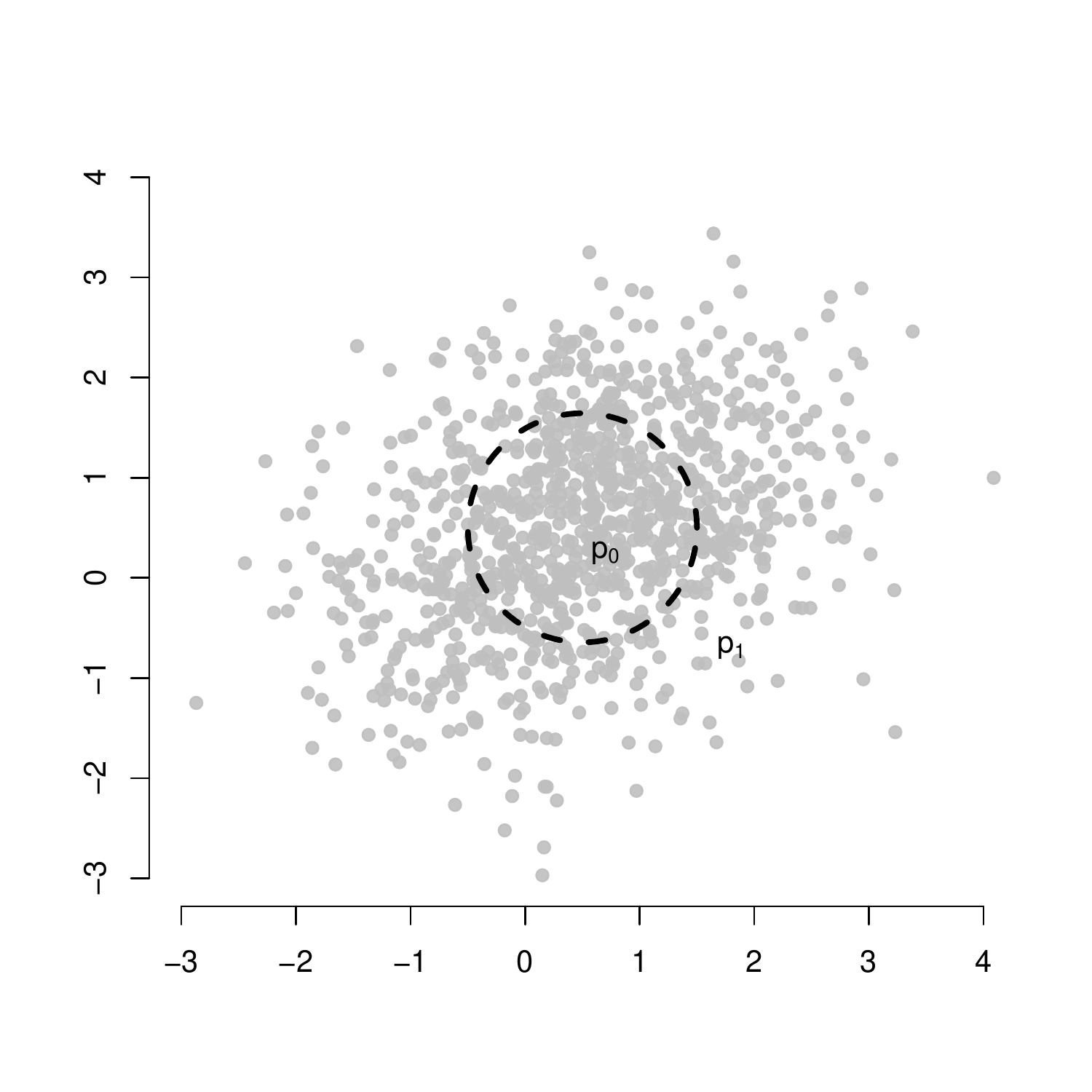}}
\caption{The simulated network for the experimental design study.  The black circle denotes the border between the city center and outskirts.  The treatment probabilities for the center and outskirts are given by $p_0$ and $p_1$, respectively.}
\label{app:spatial_sim_map}
\end{figure}

We answer this question by performing a power analysis for a grid of $p_0,p_1$.  Like the previous power analyses, we assume that the potential outcomes differ by an additive treatment effect $\tau=0.3$ and are normally distributed.  For 40,000 different combinations of $p_0,p_1$, we compute the power of the test and display the results in the left of Figure \ref{app:spatial_sim_metrics}.  The power values range from 0.25 to 0.73, with the largest values occurring for $p_1\sim0.068$ and $p_0\sim0.061$.  The power surface is relatively convex, suggesting that there is a specific range of probabilities lead to high power, while other combinations in the unit cube lead to underpowered tests.

The right side of Figure \ref{app:spatial_sim_metrics} shows the average number of focal units included in the biclique tests for a given design. Interestingly, this surface is nearly identical to the power surface \textit{except} for very small values of $p_0,p_1$.  As mentioned in the power discussion of Section~\ref{sec:power}, the number of focal units is positively related to power.  Therefore, the surfaces are close to each other.  In the bottom left corner of the surface, the treatment probabilities are very low.  Since the exposures in the hypothesis are defined as \textit{untreated} spillover and pure control units, treating a very small number of units will leave the majority as pure control units.  Moreover, this exposure status will rarely change among the different (sparse) treatment assignments.  This part of the design space illustrates the tradeoff between the number of focal units and exposure balance. Very small treatment probabilities lead to large bicliques with focal units mostly exposed to pure control.  In order to achieve optimal power, slightly smaller bicliques with a balance between spillover and pure control exposures are necessary.  The biclique method is able to navigate this tradeoff to find a strictly interior solution.

\begin{figure}[H]
\centerline{
\includegraphics[scale=0.55]{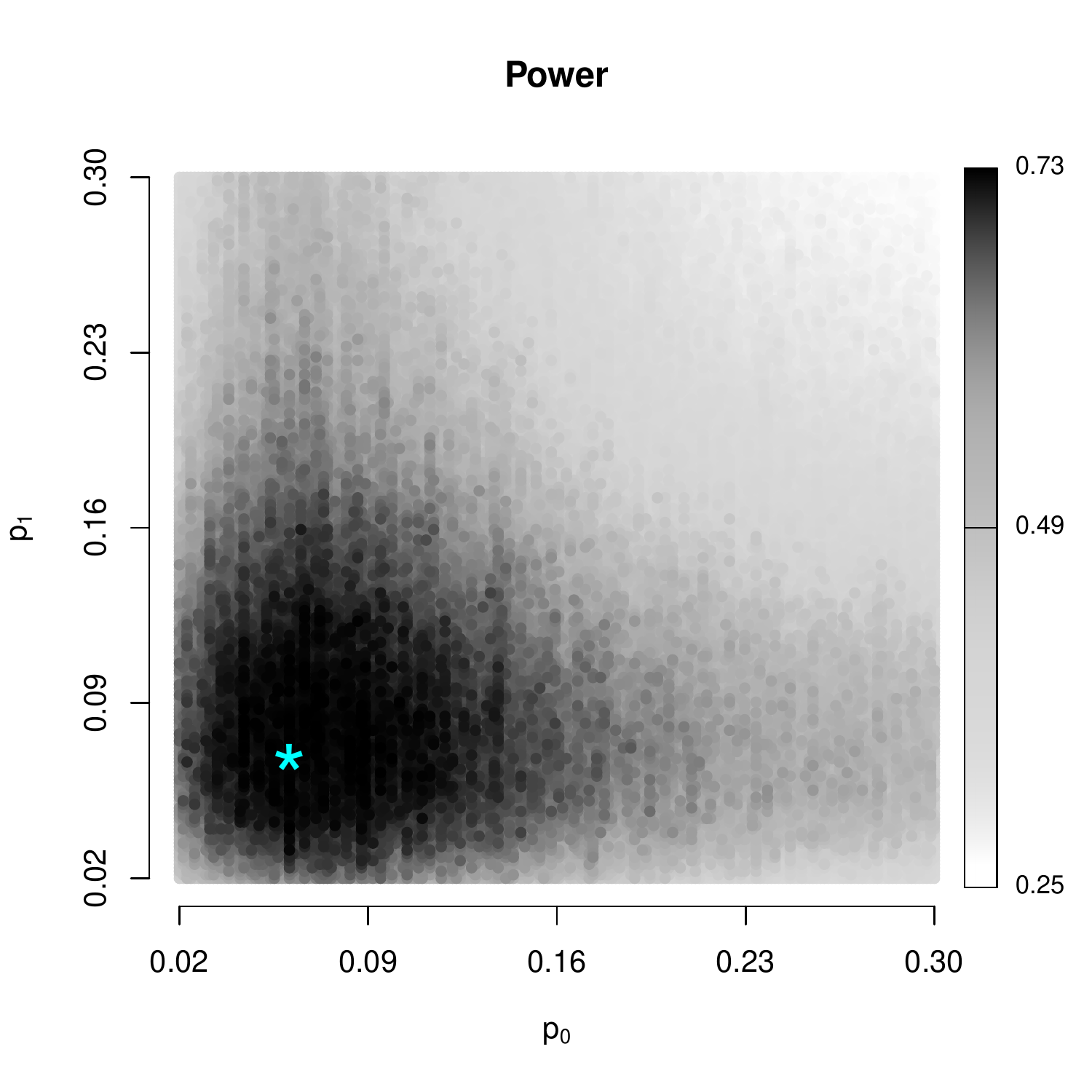} 
\includegraphics[scale=0.55]{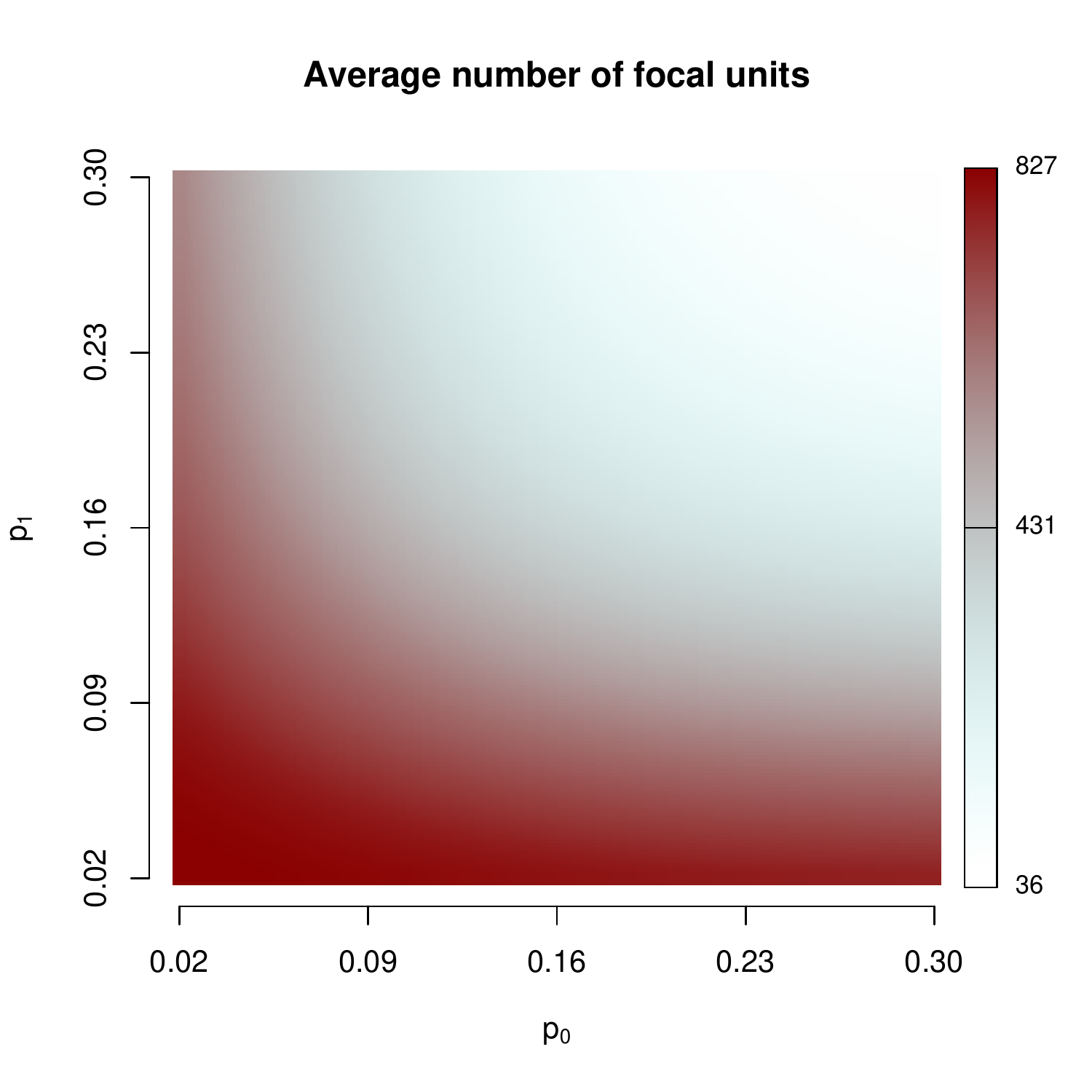}}
\caption{{\em Left:} The power of the test for different combinations of $p_0,p_1$.  Darker colors denote larger power values, while lighter colors denote smaller power values.~ 
{\em Right:} null exposure graph density for different combinations of $p_0,p_1$.  Darker colors denote larger density values, while lighter colors denote smaller density values.}
\label{app:spatial_sim_metrics}
\end{figure}

\section{More on spatial interference}\label{appendix:spatial}

Here, we show more information on how properties of the bicliques (such as biclique size) affect testing power in the context of spatial interference.
Figure \ref{app:numfocalvradius} displays the number of focals contained in the biclique for each hypothesis, $H_0^{\aa, \bb_r}$ as a function of $r$.  We see that for larger radii, there are more focals per biclique, on average, since more units are exposed to spillovers as the radius gets larger.  

\begin{figure}[H]
\centerline{\includegraphics[scale=0.7]{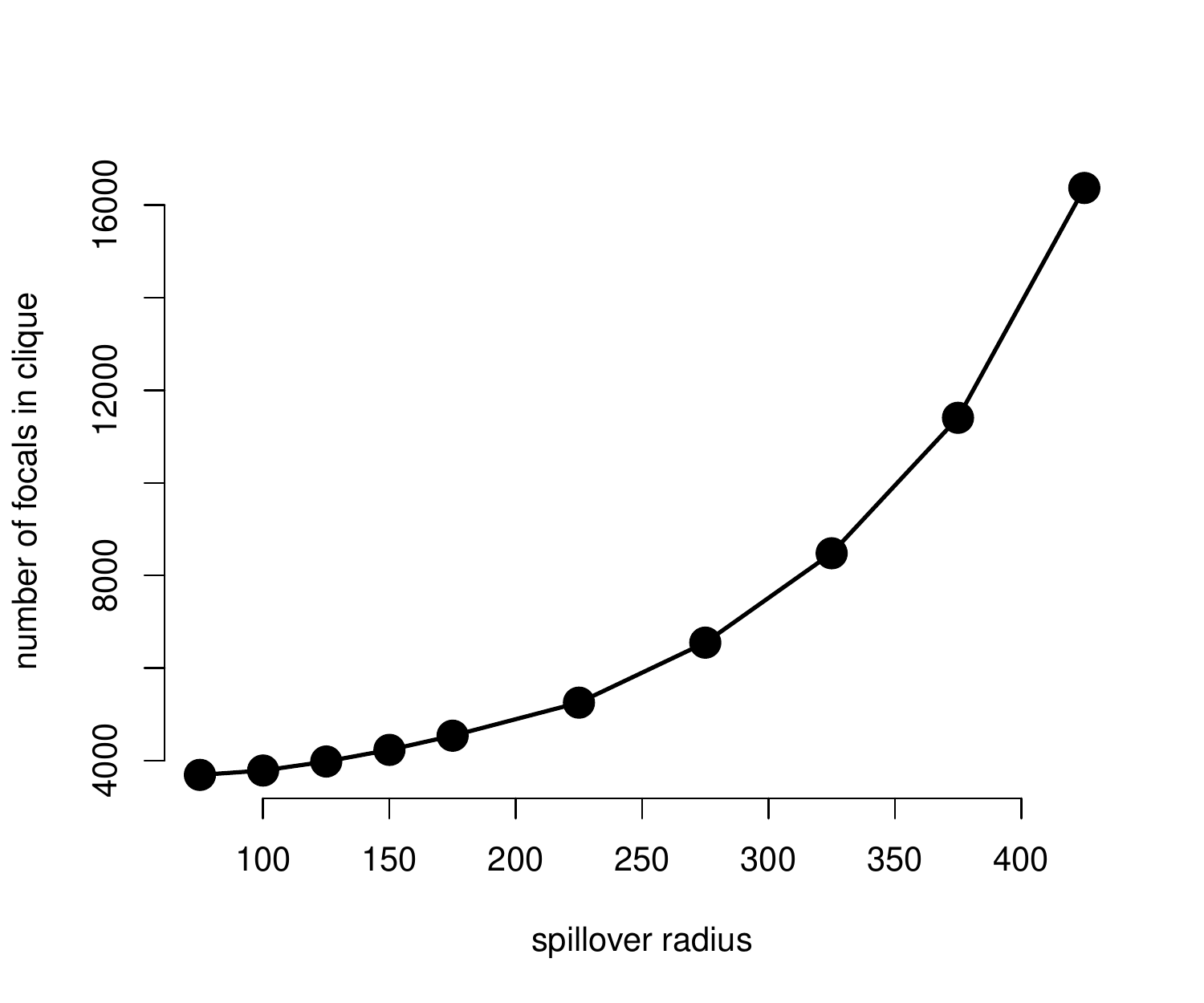}}	
\caption{Number of focals versus spillover radius for bicliques containing the observed assignment.  The radii considered are in the set $\{75,100,125,150,175,225,275,325,375,425\}$.  Notice that the number of focals increases nonlinearly as the spillover radius increases. }
\label{app:numfocalvradius}
\end{figure}

Next, we show an extended version of the randomization analysis of Figure~\ref{pval_fig1} shown in Figure~\ref{app:numfocalvradius}.
First, we show  the p-values of the biclique randomization test (left vertical axis) with respect to distance radius $r$, as described earlier~(``raw outcome" curve).
Second, we show a version of the test where we first regress the 
crime outcomes on known covariates,  including information about the neighborhood and social center points, and then perform the biclique test 
on the residuals~(``adjusted outcome" curve).
Finally, as a baseline, we also show regression coefficients from a 
simple OLS model that includes a binary variable indicating 
whether a unit receives spillovers at distance $r$ or not, and known covariates.

\begin{figure}[t!]
\centering
	\includegraphics[scale=0.7]{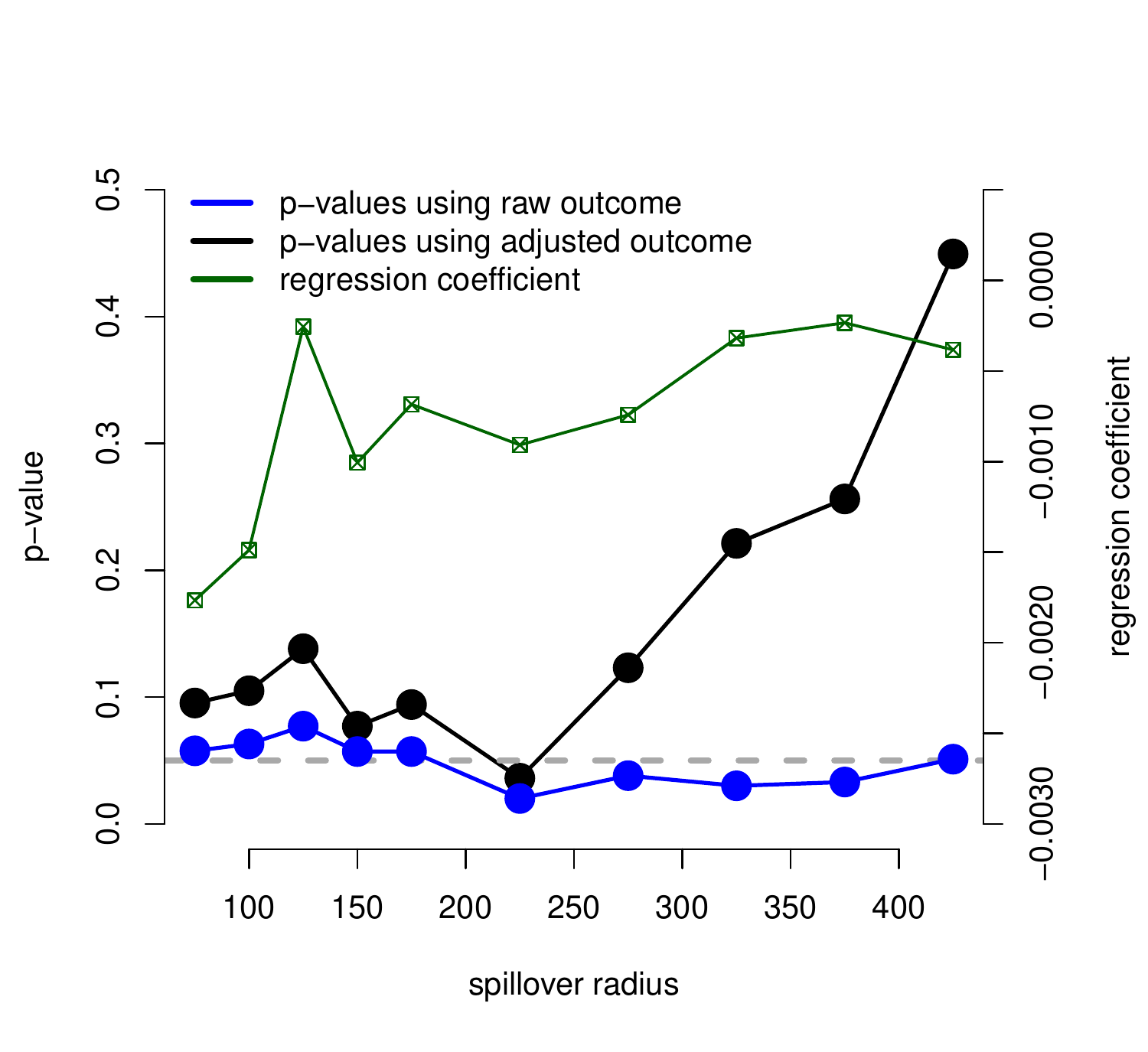}
	\caption{P-values (left vertical axis) for biclique tests with varying spillover radii (horizontal axis).  The blue line shows p-values for tests using the raw crime index, and the black line shows p-values for tests using the the crime index adjusted for known covariates.  The right vertical axis displays regression coefficients on the binary variable defined by spillover or pure control statuses. For each radii, we restrict the OLS estimation to observations such that they are either exposed to spillover or pure control.}
	\label{app:pval_reg_fig}
\end{figure}

\pagebreak

We see that the p-values for the raw outcome are all small for varying radii; see the flat blue line. This suggests that some form of spillovers exists, where the distance 
does not seem to matter.
However, the biclique test on the adjusted outcomes (black curve in Figure \ref{app:pval_reg_fig}) points to the other direction, as it does 
not indicate significance of spillover effects at any distance.
This result suggests  that the covariate distributions of 
``pure control" and ``spillover" units are very different, 
such that the significance of the raw outcome FRTs may be attributed to that difference.
The regression coefficient (green curve) agrees with the adjusted outcome results: no regression coefficient is significant at the 0.05 level, and there is a similar, though concave, trend for increasing radii.

Finally, we show here additional information on the Medell\'in policing experiment.  Figure~\ref{app:randdist_fig} shows the randomization distribution of the 
test statistics for various radii, $r$, 
and for both raw outcomes and adjusted outcomes.
We see that the tests with the adjusted outcomes are sharper than the 
tests with the raw outcomes, indicating heterogeneity in the 
spillover effects.

\begin{figure}[t!]
\centerline{
\includegraphics[scale=0.6]{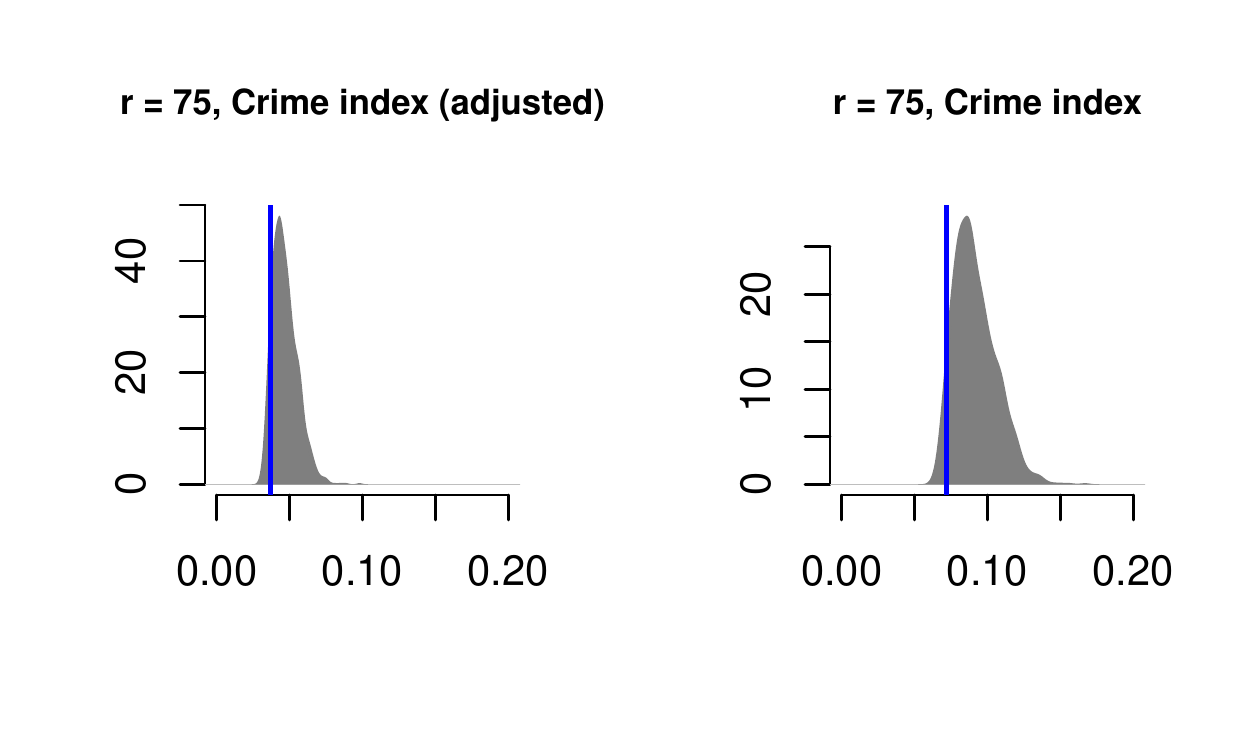} \hspace{5mm}\includegraphics[scale=0.6]{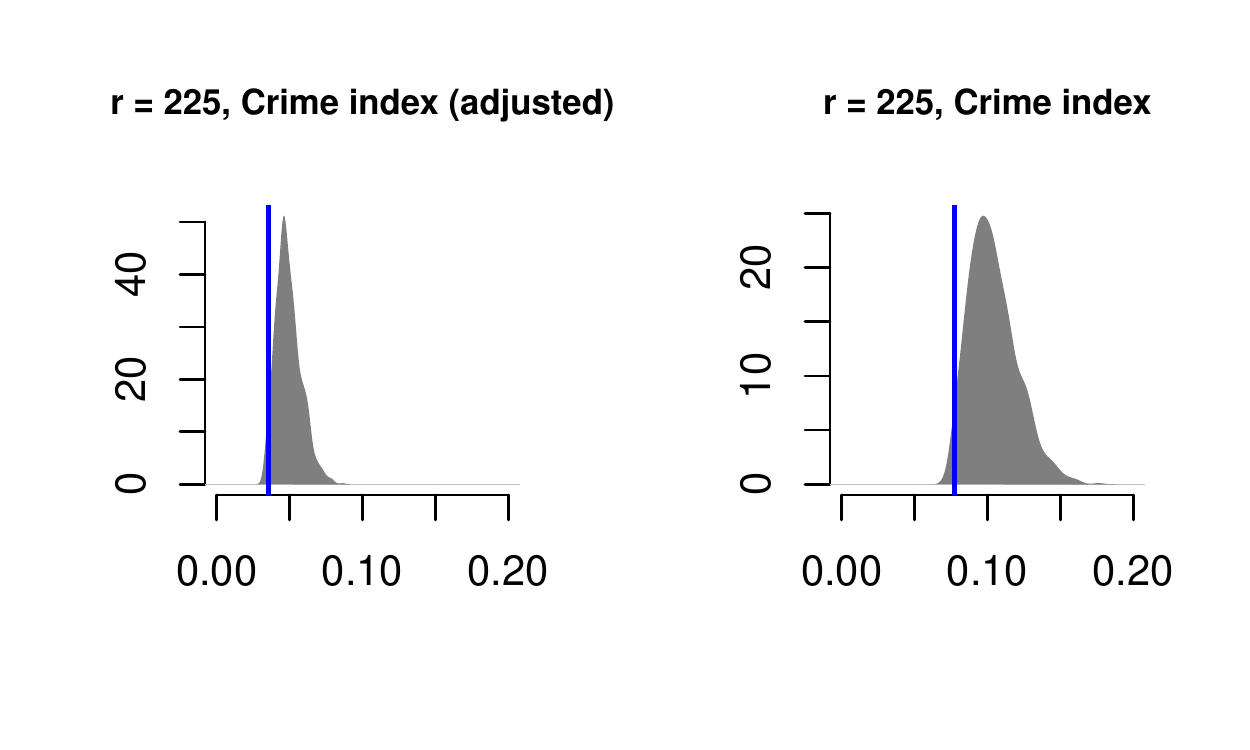}} \vspace{-10mm}
\centerline{	\includegraphics[scale=0.6]{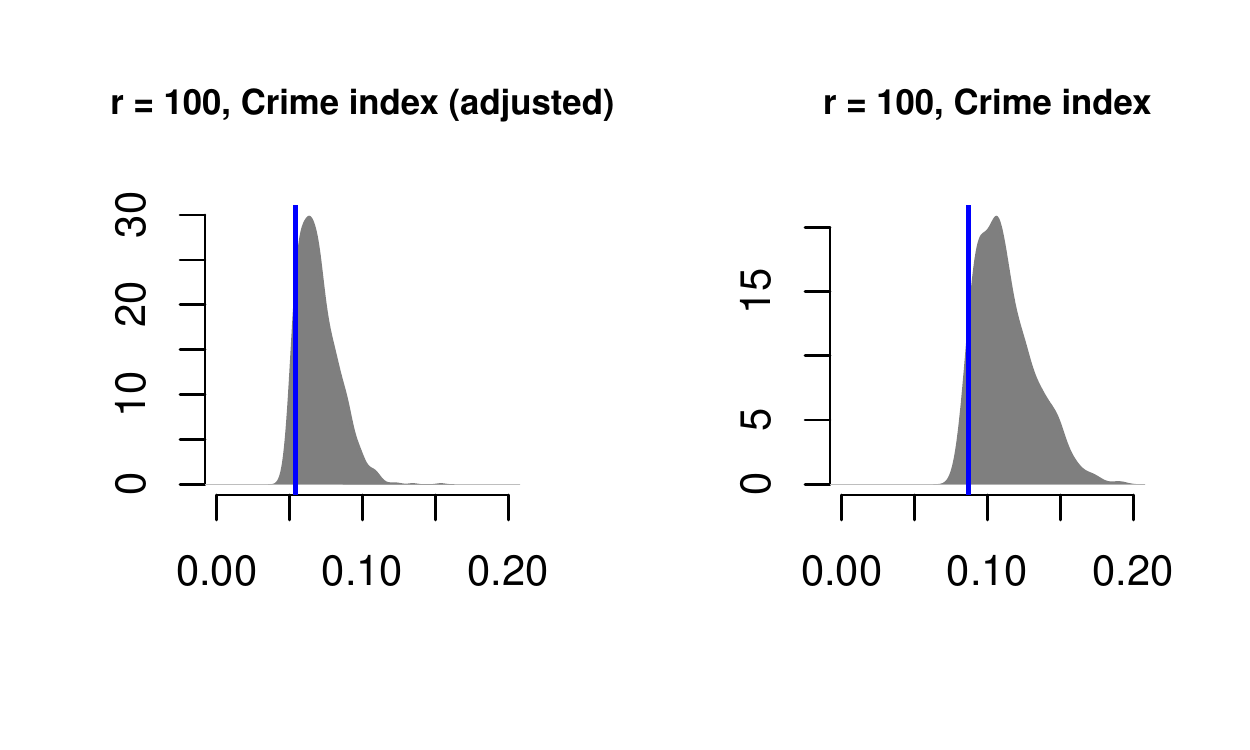} \hspace{5mm}\includegraphics[scale=0.6]{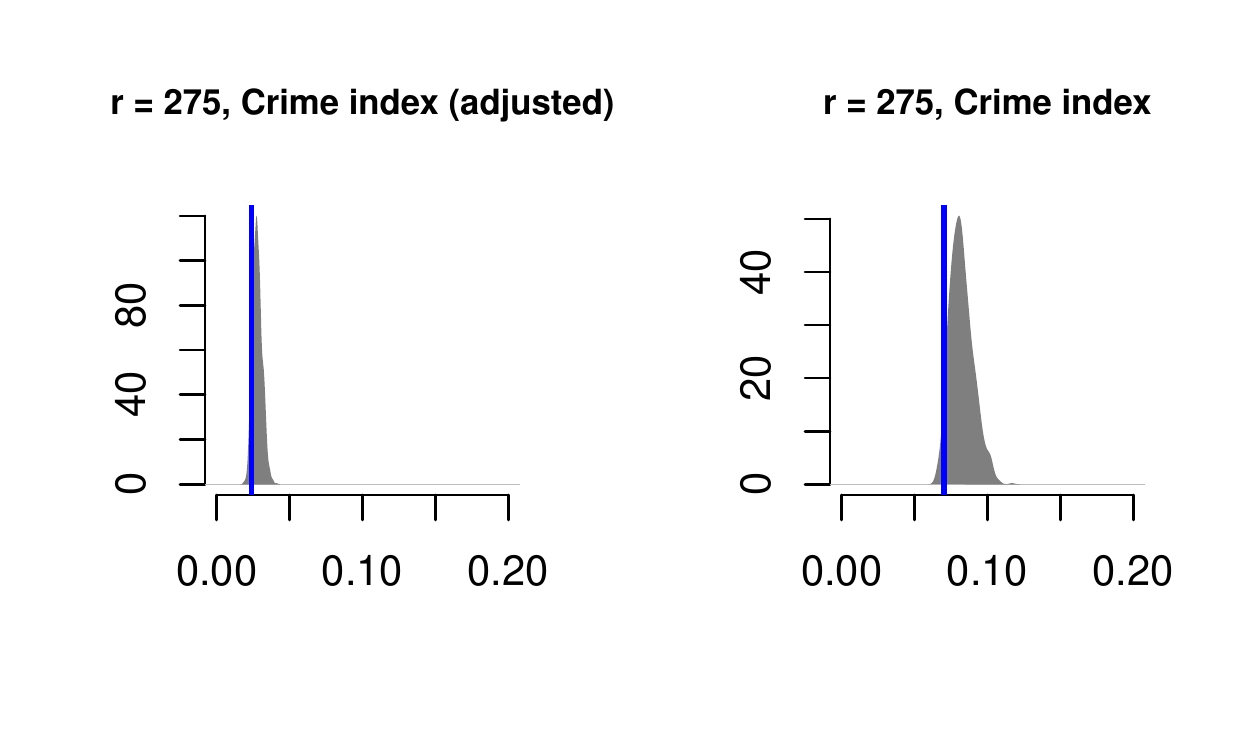}} \vspace{-10mm}
\centerline{\includegraphics[scale=0.6]{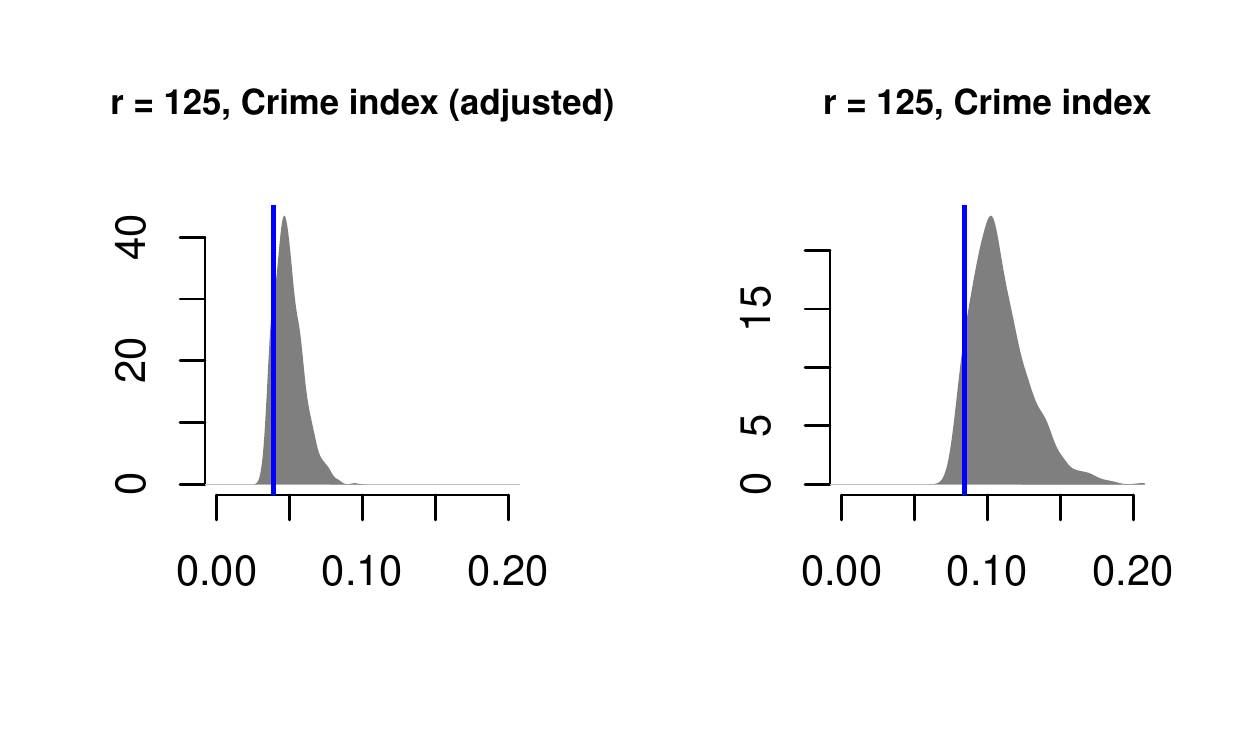} \hspace{5mm}\includegraphics[scale=0.6]{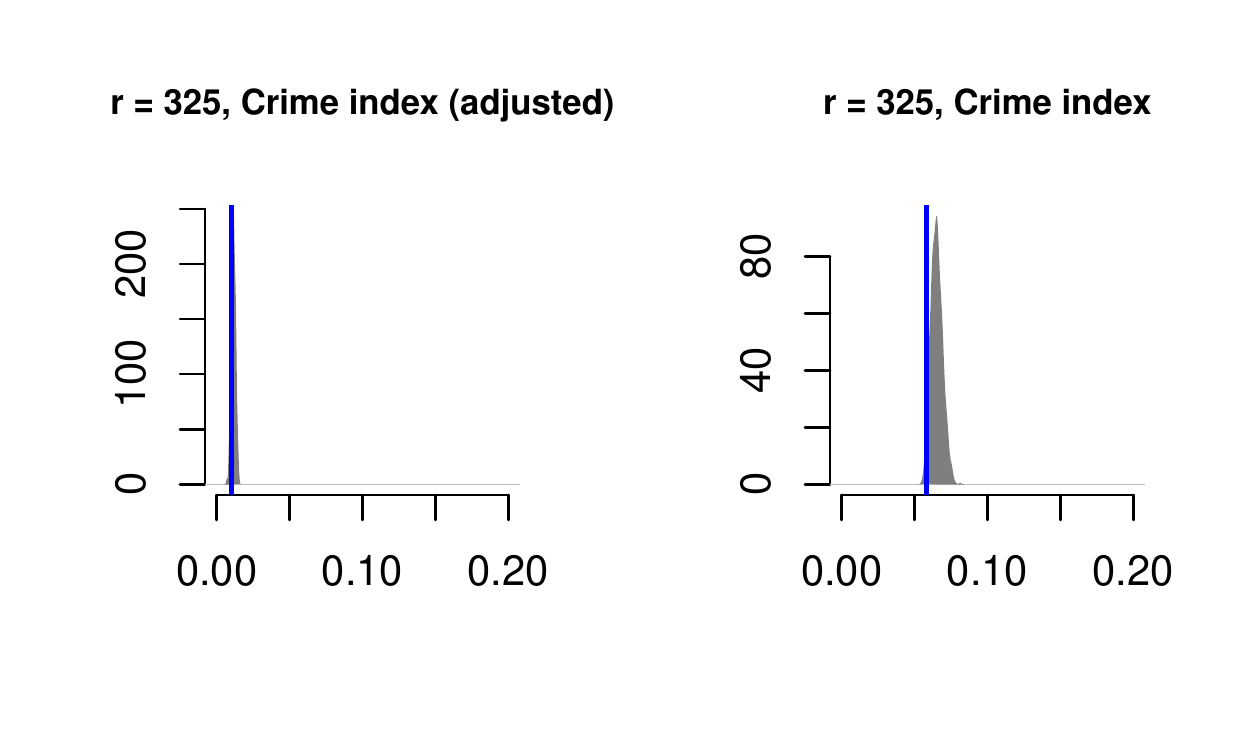} } \vspace{-10mm}
	\centerline{\includegraphics[scale=0.6]{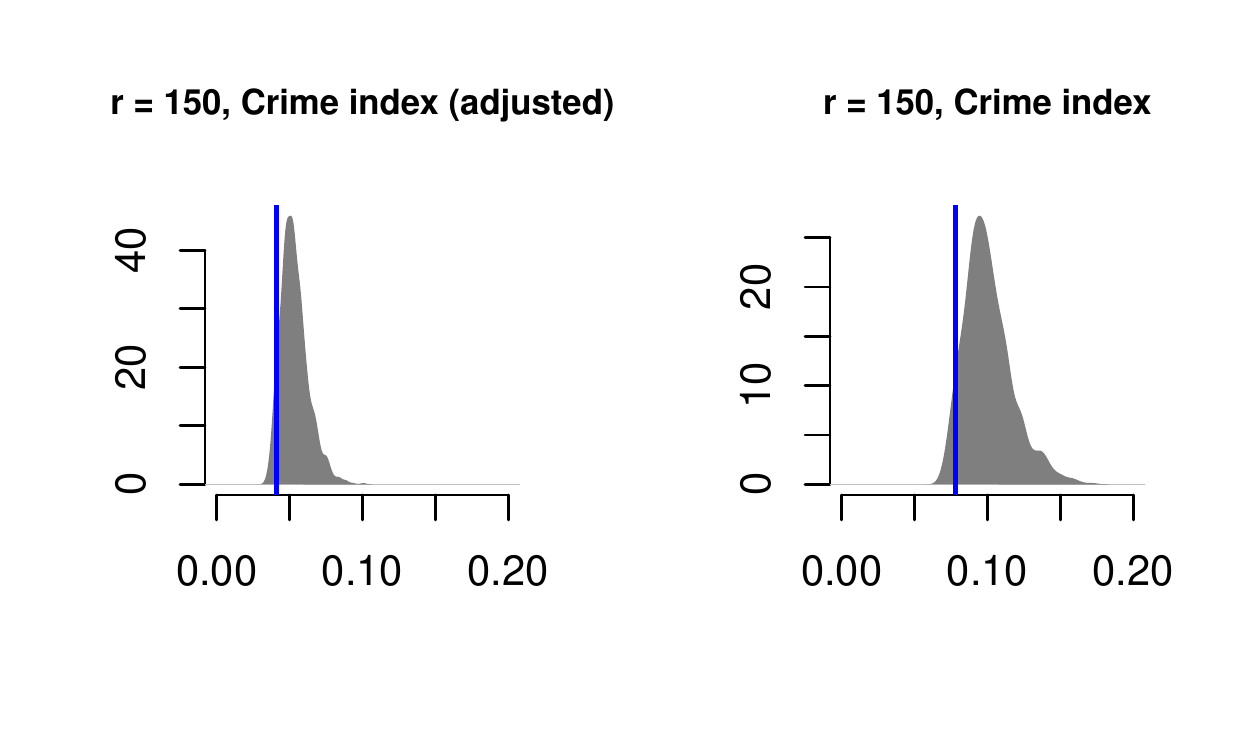} \hspace{5mm}\includegraphics[scale=0.6]{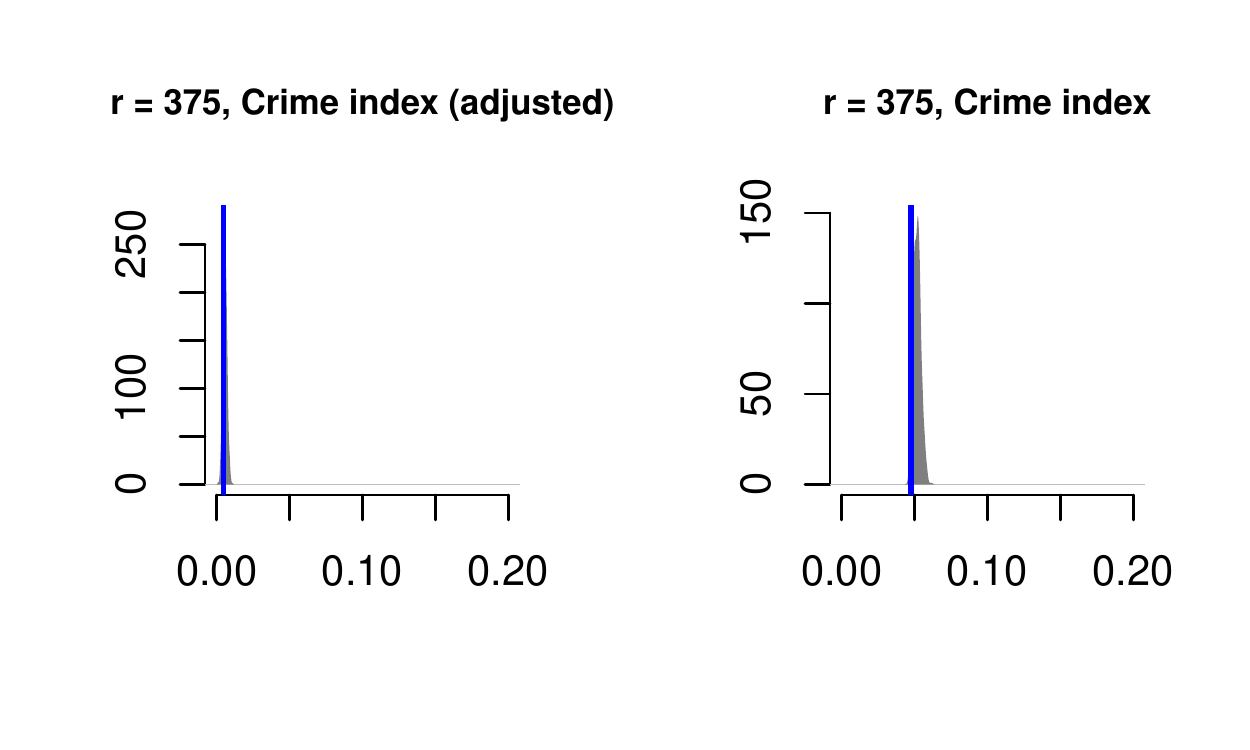}} \vspace{-10mm}
\centerline{\includegraphics[scale=0.6]{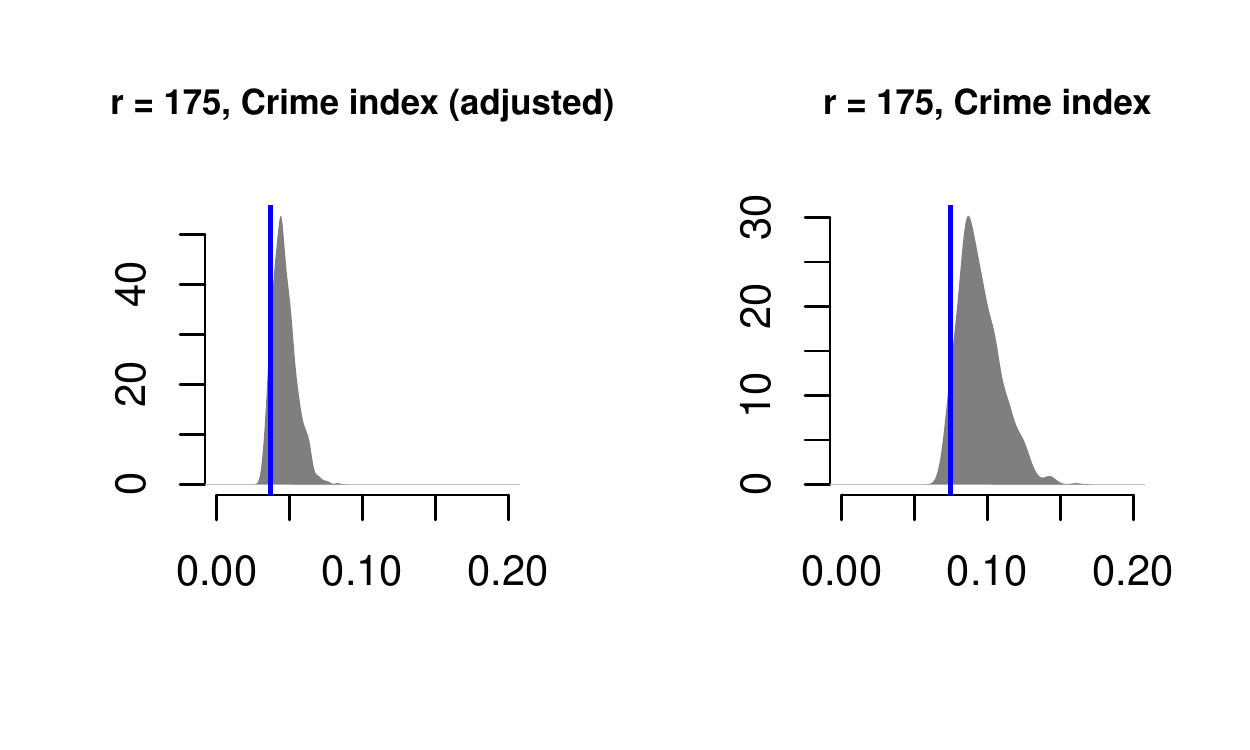} \hspace{5mm}\includegraphics[scale=0.6]{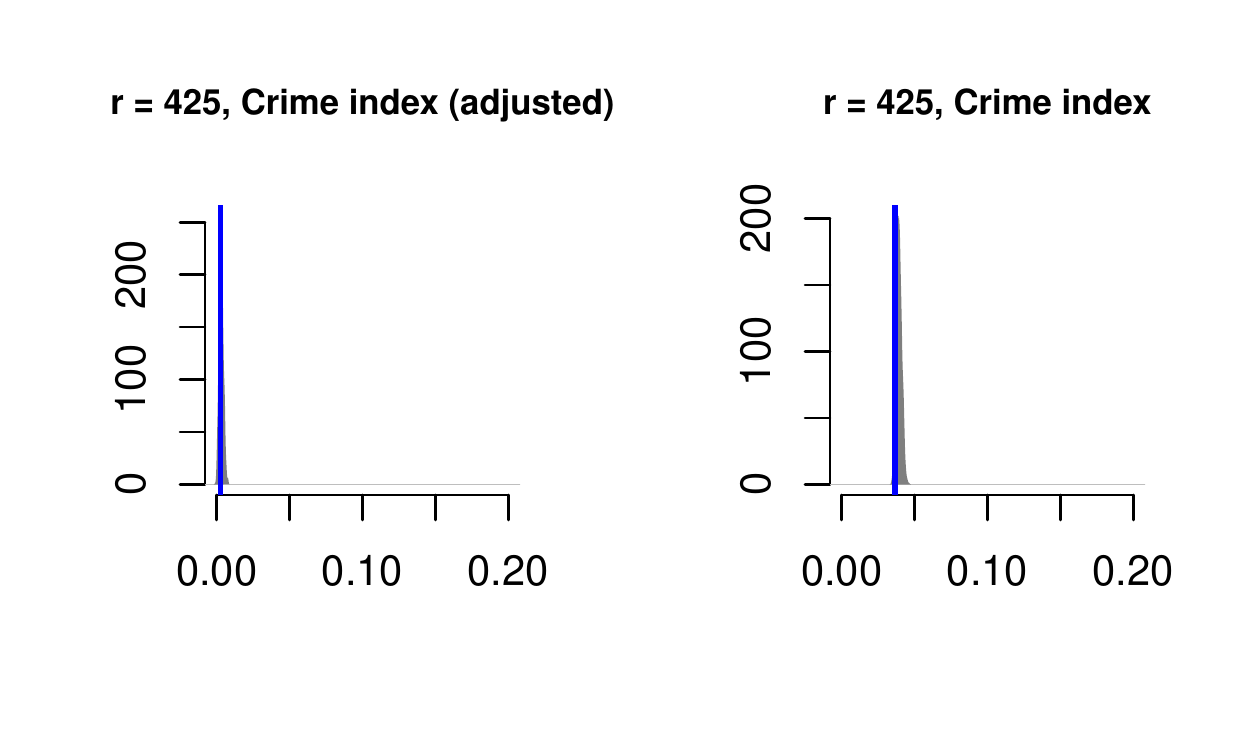}} 
	\caption{Randomization distributions and observed test statistics (blue) for 20 biclique tests. The first and third columns use the adjusted crime index as the outcome, while the second and fourth columns use the raw crime index.  The radius defining spillover exposure status varies from 75 meters (top left) to 425 meters (bottom right). Note that the randomization distribution using the adjusted crime index have lower variance and are centered at smaller positive values compared to their raw index counterparts.  However, the p-values from the tests are largely same.}
	\label{app:randdist_fig}
\end{figure}

\newpage

\section{Testing the general intersection hypothesis, $\Hinter$}\label{appendix:Hex}

\subsection{Proof of Theorem~4}

\begin{theorem}
\ThmTwo
\end{theorem}
\begin{proof}
The main difference with Theorem~\ref{thm:main} is that the test conditions on a biclique decomposition, $\Cset$, that is comprised of multiple null exposure graphs, $G(z; \Zset)$.
Given $\Cset$, as in Theorem~\ref{thm:main}, our conditioning mechanism satisfies:
$$
P(C | \Zobs) = \Ind{\Zobs \in \Zset(C)} \Ind{C\in\Cset},
$$
and so the derivations in Step 2 of Theorem~\ref{thm:main} still hold. Thus, we only need to show that 
Step 1 about imputability of potential outcomes is correct too.

To see this, let $C$ be the conditioning biclique calculated from $G(\Zobs; \Zdom_0)$ for some $\Zdom_0\subseteq \Zdom$, as described in Procedure~\ref{proc:multi}. 
Take any unit $i$ and assignment $z'$ in $C$. Since $(i, z')$ is an edge in $C$, then, by definition of $G(z; \Zset)$, 
$f_i(z') = f_i(\Zobs)$. Under $\Hex$, it follows that $Y_i(z') = Y_i(\Zobs)$.
\end{proof}

\subsection{Multi-null exposure graph}\label{appendix_HI}

Here, we show how to extend the test for $\Hex$ in order to test $\Hinter$ in its more general form~\eqref{eq:H0_inter}.
%
Below we define formally the concept of multi-null exposure graph. 

\begin{definition}[Multi-null exposure graph]
\label{def:multi}
Consider the intersection hypothesis $\Hinter$ of Equation~\eqref{eq:H0_inter}.
For any unit $i\in\Udom$ and $z\in\Zset\subseteq\Zdom$,  let $A(i, z)$ be the unique set $\Fset\in\Inter$ such that $f_i(z)\in\Fset$, if such set exists; 
otherwise, let $A(i, z) = \{ \}$. 
Also, let  $G(z; \Zset) = (V, E)$ be the graph such that $V = \Udom \cup \Zset$ and $E =\{ (i, z') \in \Udom \times \Zset: f_i(z') \in A(i, z)\}$.
Then, $G(z; \Zset)$ is the multi-null exposure graph of $\Hinter$ with respect to $z\in\Zdom, \Zset\subseteq\Zdom$.
\end{definition}

This definition is a generalization with respect to the definition based on $\Hex$. The main difference is that Definition~\ref{def:multi} may leave out some units-assignment pairs if the corresponding potential outcomes cannot be imputed under $\Hinter$. 
To test $\Hinter$, Procedure~\ref{proc:multi} then needs to be slightly updated  by 
initializing $\Zdom_0$ not as $\Zdom$ but as 
$$
\Zdom_0 \gets \Zdom \setminus \left\{ z \in\Zdom : f_i(z)\notin\bigcup_{j=1}^J \Fset_j~\text{for all}~i\in\Udom\right\}.
$$
Definition~\ref{def:multi} of the multi-null exposure graphs $G(z; \Zset)$ can also be used 
to test $\Hinter$ based on the method of \citet{athey2018exact} 
as described in Section~\ref{sec:multi_athey}.
\color{black}

\singlespacing
\small
\bibliographystyle{apalike}
\bibliography{refs.bib}

\end{document}